\documentclass[11pt]{article}
\pdfoutput=1
\usepackage[margin=2.5cm,top=2cm,bottom=2cm,centering]{geometry}
\usepackage[utf8]{inputenc}
\usepackage[english]{babel}
\usepackage{graphicx}
\usepackage{amsmath}
\usepackage{amsfonts}
\usepackage{amssymb}
\usepackage{amsthm}
\usepackage{mathrsfs}
\usepackage{enumerate}
\usepackage{hyperref}
\usepackage[dvipsnames]{xcolor}
\usepackage{tikz}
\usepackage[hypcap]{caption}
\usepackage{bm}
\usepackage{cancel}
\usepackage{cite}
\usepackage{lscape}
\usepackage{slashed}
\usetikzlibrary{calc}
\usetikzlibrary{decorations.pathmorphing,decorations.pathreplacing,decorations.markings}
\usetikzlibrary{arrows,shapes,positioning}
\usetikzlibrary{decorations.markings}
\usetikzlibrary{patterns}
\usetikzlibrary{lindenmayersystems}

\hypersetup{
colorlinks=true,
citecolor=red,
filecolor=black,
linkcolor=blue,
urlcolor=blue,
pdfauthor={Adrien Poncelet},
pdftitle={Schramm's formula for multiple loop-erased random walks},
pdfstartview={XYZ null null 1.50},
bookmarksopen=true
}

\tikzset{
  on each segment/.style={
    decorate,
    decoration={
      show path construction,
      moveto code={},
      lineto code={
        \path [#1]
        (\tikzinputsegmentfirst)--(\tikzinputsegmentlast);
      },
      curveto code={
        \path [#1] (\tikzinputsegmentfirst)
        .. controls
        (\tikzinputsegmentsupporta) and (\tikzinputsegmentsupportb)
        ..
        (\tikzinputsegmentlast);
      },
      closepath code={
        \path [#1]
        (\tikzinputsegmentfirst)--(\tikzinputsegmentlast);
      },
    },
  },
  mid arrow/.style={postaction={decorate,decoration={
        markings,
        mark=at position 0.75 with {\arrow[#1]{stealth}}
      }}},
}

\numberwithin{equation}{section}

\newcommand*\xbar[1]{%
  \hbox{%
    \vbox{%
      \hrule height 0.5pt 
      \kern0.3ex
      \hbox{%
        \kern-0.1em
        \ensuremath{#1}%
        \kern-0.0em
      }%
    }%
  }%
} 

\newcommand{\C}{\mathbb{C}}
\newcommand{\Z}{\mathbb{Z}}
\newcommand{\E}{\mathcal{E}}
\newcommand{\G}{\mathcal{G}}
\newcommand{\N}{\mathcal{N}}
\newcommand{\V}{\mathcal{V}}
\newcommand{\Gr}{\mathbf{G}}
\newcommand{\Zr}{\mathbf{Z}}
\newcommand{\bs}{\backslash}

\newcommand{\diff}{\text{d}}
\newcommand{\w}{\mathrm{w}}
\newcommand{\vsig}{\vec{\sigma}}
\newcommand{\vmu}{\vec{\mu}}
\newcommand{\vt}{\vec{\tau}}
\newcommand{\vtp}{\vec{\tau}^{\,\prime}}
\newcommand{\pf}{\mathrm{pf}}
\newcommand{\x}{\mathrm{x}}
\newcommand{\y}{\mathrm{y}}
\newcommand{\s}{\mathrm{s}}

\renewcommand{\i}{\text{i}}
\renewcommand{\l}{\ell}
\renewcommand{\P}{\mathbb{P}}

\renewcommand\Re{\operatorname{Re}}
\renewcommand\Im{\operatorname{Im}}
\renewcommand{\ge}{\geqslant}
\renewcommand{\le}{\leqslant}
\renewcommand{\L}{\text{L}}
\newcommand{\R}{\text{R}}

\newtheorem{thm}{Theorem}[section]

\newtheorem{prop}[thm]{Proposition}
\newtheorem{lm}[thm]{Lemma}
\newtheorem{defn}[thm]{Definition}
\newtheorem{cor}[thm]{Corollary}

\title{\bf Schramm's formula for multiple\\loop-erased random walks\vspace{-0.5cm}}
\author{\normalsize \textsc{Adrien Poncelet}\medskip\\
{\normalsize
\begin{minipage}{0.95\textwidth}
\begin{center}
\textit{Universit\'e catholique de Louvain\\Institut de recherche en math\'ematique et physique\\Chemin du Cyclotron 2, 1348 Louvain-la-Neuve, Belgium}\\
\medskip
\href{mailto:adrien.poncelet@uclouvain.be}{\normalsize\texttt{adrien.poncelet@uclouvain.be}}
\end{center}
\end{minipage}}
}
\date{}

\begin{document}
\maketitle

\begin{abstract}
We revisit the computation of the discrete version of Schramm's formula for the loop-erased random walk derived by Kenyon. The explicit formula in terms of the Green function relies on the use of a complex connection on a graph, for which a line bundle Laplacian is defined. We give explicit results in the scaling limit for the upper half-plane, the cylinder and the Möbius strip. Schramm's formula is then extended to multiple loop-erased random walks.

\end{abstract}


\section{Introduction}
\label{sec1}

Over the last few decades, significant interest has been shown in two-dimensional statistical lattice models at their critical point, both in the mathematics and physics communities. In the scaling limit, that is, in the limit where the mesh of the lattice tends to zero, a range of models have been shown to exhibit conformal invariance. In particular, there has been extensive research concerning the description of paths and interfaces that naturally arise between boundary points. For a single curve of that type, said to be chordal, rigorous results in the scaling limit have been obtained for a variety of models on simply connected domains: percolation \cite{Smi01,CN07}, the loop-erased random walk and the uniform spanning tree \cite{Sch00,LSW04}, the discrete Gaussian free field \cite{SS09,SS13} and the critical Ising and FK-Ising models \cite{CS12,CDCHKS14}. The limiting curves obtained in these models are particular cases of the so-called chordal \emph{Schramm-Loewner evolution with parameter $\kappa\ge 0$} ($\text{SLE}_{\kappa}$), introduced in \cite{Sch00}; the value of $\kappa$ depending on the lattice model. In particular, $\text{SLE}_2$ has been proved to be the scaling limit of the loop-erased random walk and the chemical path in the uniform spanning tree \cite{LSW04} (both models possessing the same probability distribution \cite{Pem91,Wil96}). The extension of $\text{SLE}_{\kappa}$ to multiply connected domains has been discussed for instance in \cite{BF04,Zha04,Law11}. Similarly, multiple interfaces and paths in lattice models are expected to be described by multiple SLE curves in interaction (for instance, conditioned not to intersect each other) \cite{BBK05,Dub06,Kyt06,KL07,Law09,Izy17,PW17b,Wu17}. An alternative description of the continuum limit of lattice models is given by \emph{conformal field theory} (CFT), in which observables are identified as correlation functions of local operators. The connection with the SLE picture for the description of boundary paths and interfaces has been given in \cite{BB02,FW03}.

A natural question regarding multiple curves between fixed boundary points consists in the computation of crossing or connection probabilities. The former have been discussed for percolation \cite{Car92} and the critical Ising model \cite{ASA02}, while the latter have been computed for the loop-erased random walk and the uniform spanning tree \cite{KW11a,KKP17}, the double dimer model \cite{KW11a} and the discrete Gaussian free field \cite{PW17b}.

For a single curve, one of the simplest results in the SLE framework is \emph{Schramm's formula}, which gives the probability that a point $z$ lies to the left (or right) of an $\text{SLE}_{\kappa}$ curve between fixed boundary points \cite{Sch01}. Its explicit form on the upper half-plane for a random curve $\gamma$ starting at 0 and growing to infinity is given by
\begin{equation}
\P[\text{$z$ lies to the left of $\gamma$}]=\frac{1}{2}-\frac{\Gamma(4/\kappa)}{\sqrt{\pi}\,\Gamma\left(\frac{8-\kappa}{2\kappa}\right)}\frac{x}{y}\,{}_2 F_1\left(\frac{1}{2},\frac{4}{\kappa};\frac{3}{2};-\frac{x^2}{y^2}\right),
\label{Schramm_SLE}
\end{equation}
where $z=x+\i y$ and ${}_2 F_1$ is the ordinary hypergeometric function. This formula has been extended in several directions, namely to curves in doubly connected domains for $\kappa=2$ \cite{Hag09a} and $\kappa=4$ \cite{HBB10}, and for the left-passage probability with respect to two marked points $z_1,z_2$ for $\kappa=8/3$ \cite{SC09,BV13}. The generalization to multiple curves passing left or right of one marked point has also been investigated in \cite{GC05}, in which the passage probabilities were explicitly computed for two curves only, for $\kappa=0,2,4,8/3,8$ (with more rigorous proofs given later in \cite{LV17}).

Most of the results regarding Schramm's formula, including the original paper \cite{Sch01}, have been directly obtained in the scaling limit of lattice models. A notable exception is \cite{IP12}, in which an explicit expression for the left-passage probability with respect to a specific type of marked points is given on the lattice, for the percolation model. In this article, we consider the spanning forest model and the loop-erased random walk on a planar graph. We give a discrete version of Schramm's formula for paths in spanning forests in terms of the Green function of the graph (this problem was briefly discussed in \cite{Ken11,KW15}). Due to the well-known correspondence with paths in spanning trees \cite{Pem91,Wil96}, the formula also holds for loop-erased random walks. We compute several explicit expressions in the scaling limit for various domains and boundary conditions, and compare some of our results to known $\text{SLE}_2$ calculations. We then establish a generalization of Schramm's formula for multiple curves between boundary points on a planar graph, for all possible connectivities.

The paper is organized as follows. In Section~\ref{sec2}, we recall graph-theoretical definitions and properties associated with spanning forests, as well as their generalization for graphs equipped with a connection, namely so-called cycle-rooted groves, introduced in \cite{KW15}. In Section~\ref{sec3}, we define a measure on paths in spanning forests and cycle-rooted groves. We give a combinatorial expression of Schramm's formula for paths between two boundary vertices with respect to a marked face. We list our explicit results in the scaling limit in Section~\ref{sec4}, in which we consider the upper half-plane, the cylinder and the Möbius strip, for wired or free boundary conditions. We generalize Schramm's formula for $n$ paths with respect to a single face in Section~\ref{sec5}. The case $n=2$ is exposed in details before the introduction of the general formalism for $n>2$. In Section~\ref{sec6}, we explicitly show that the measures on loop-erased random walks and random paths in spanning forests coincide. The interpretation of partition functions for paths as conformal correlators via the SLE/CFT correspondence is recalled in Section~\ref{sec7}, in which we check the consistency of some of our results on the upper half-plane. Finally, two appendices collect the definitions of Jacobi's elliptic and theta functions, together with some of their properties relevant to computations presented in Section~\ref{sec4}, and the proofs of the propositions contained in Section~\ref{sec5}.


\section{Spanning forests and cycle-rooted groves}
\label{sec2}

In this section, we first recall definitions and classical results of graph theory. We give an overview of some new techniques pertaining to graphs equipped with a connection, introduced in \cite{For93,Ken11,KW15}, which are relevant for the computations of later sections. For the most part, we follow the formalism of \cite{KW15}.


\subsection{The matrix-tree theorem}
\label{sec2.1}

Let $\G$ be an unoriented connected simple graph, with a set of vertices $\mathcal{V}$ and edges $\E$. Let us denote by $(u,v)$ and $\{u,v\}$ the oriented edge and the unoriented edge from $u$ to $v$, respectively (with $\{v,u\}=\{u,v\}$). Let $s$ be an additional vertex, called the \emph{root}, connected to a subset $\mathcal{D}\subset\V$ by a set of edges $\E_s$. We denote by $\G_s$ the extended graph with vertices $\mathcal{V}\cup\{s\}$ and edges $\E\cup\E_s$. A spanning tree on $\G_s$ is a connected acyclic subgraph that contains all the vertices of $\mathcal{V}\cup\{s\}$. In what follows, we shall make the distinction between a spanning tree and a \emph{rooted} spanning tree, in which all edges point toward a single vertex, usually taken to be the root $s$. We define a conductance function $c:\E\cup\E_s\to \mathbb{R}^*_+$, associating a weight $c_{u,v}=c_{v,u}$ with each unoriented edge $\{u,v\}\in\E$. We further choose a unit conductance for the edges $\{u,s\}$ in $\E_s$: $c_{u,s}=c_{s,u}=1$. The weight of a spanning tree $\mathcal{T}$ on $\G_s$ is given by
\begin{equation}
\w(\mathcal{T})=\prod_{\{u,v\}\in\mathcal{T}}c_{u,v}.
\label{ST_weight}
\end{equation}
The standard Laplacian of the graph $\G_s$ is denoted by $\Delta_{\G_s}$ and is defined for any $u,v\in\V\cup\{s\}$ by
\begin{equation}
\left(\Delta_{\G_s}\right)_{u,v}=\begin{cases}
\sum_{w:\{u,w\}\in\E\cup\E_s}c_{u,w}&\quad\text{if $v=u$,}\\
-c_{u,v}&\quad\text{if $\{u,v\}\in\E\cup\E_s$,}\\
0&\quad\text{otherwise.}
\end{cases}
\end{equation}
One sees that $\Delta_{\G_s}$ is singular since it has a nondegenerate\footnote{Recall that the algebraic multiplicity of the zero eigenvalue of the standard Laplacian is equal to the number of connected components of the graph.} zero eigenvalue, whose corresponding eigenspace is generated by $(1,1,\ldots,1)^{\text{t}}$ (see Section \ref{sec2.3} for further discussion). Therefore, it is customary to define $\Delta^{(s)}_{\G_s}$, the Laplacian of $\G_s$ with Dirichlet boundary conditions at $s$, as the submatrix of $\Delta_{\G_s}$ from which the row and column indexed by $s$ have been removed. In what follows, we shall drop the subscript $\G_s$ and the superscript $s$, and simply refer to this Dirichlet Laplacian as $\Delta$. Moreover, we shall call \emph{wired} the vertices of $\mathcal{D}$ (i.e. those that are connected to the root by an edge in $\E_s$), and \emph{free} or \emph{unwired} the vertices in $\V\bs\mathcal{D}$.

In the following sections, we shall consider graphs $\G=(\V,\E)$ embedded on surfaces bounded by Jordan curves (in fact, we shall mainly deal with planar graphs embedded on the plane). We call \emph{boundary vertices} the vertices in $\V$ belonging to these curves, and take $\mathcal{D}$ to be a subset thereof. A graph in which all boundary vertices belong to $\mathcal{D}$ (resp. $\V\bs\mathcal{D}$) will be said to have a \emph{wired boundary} (resp. \emph{free boundary}). As we shall see in Section~\ref{sec4} for several graphs, the discrete Green function with wired (resp. free) boundary conditions converges to the continuum Green function with Dirichlet (resp. Neumann) boundary conditions in the scaling limit. More generally, we shall consider graphs with uniform conditions on each boundary component separately, for which this convergence property of the Green function also holds.

A well-known result in graph theory, the \emph{matrix-tree theorem} (also known as Kirchhoff's theorem), states that the determinant of the Dirichlet Laplacian is equal to the weighted sum of spanning trees on $\G_s$.
\clearpage
\begin{thm}[\hspace{-0.025cm}\cite{Kir47}\hspace{-0.025cm}]
Let $\mathscr{T}$ denote the set of spanning trees on the graph $\G_s$, and let $\Delta$ be the Dirichlet Laplacian of $\G_s$. Then
\begin{equation}
\det\Delta=\sum_{\mathcal{T}\in\mathscr{T}}\w(\mathcal{T}),
\end{equation}
where the weight of a spanning tree is defined by Eq.~\eqref{ST_weight}.
\label{MTT}
\end{thm}

More generally, we shall consider spanning forests on $\G_s$ (that is, acyclic spanning subgraphs). The natural weight of a spanning forest is given by the product of the weights of its tree components. The generalization of Kirchhoff's theorem for spanning forests, called the \emph{all-minors matrix-tree theorem} is as follows.
\begin{thm}[\hspace{-0.025cm}\cite{Cha82}, Theorem 1\hspace{-0.025cm}]
In addition to $s$, let $R=\{r_1,\ldots,r_n\}$ and $S=\{s_1,\ldots,s_n\}$ be two disjoint subsets of $\V$, and $\sigma=r_1 s_1|\ldots|r_n s_n$ be a pairing of these vertices (the vertical bars separating the $n$ pairs). One denotes by $Z[\sigma]$ the weighted sum of spanning forests on $\G_s$ in which each of the $n{+}1$ trees contains either the root $s$ or a single pair $\{r_i,s_i\}$, $1\le i\le n$. Then the following sum rule holds:
\begin{equation}
\det\Delta\det G_{R}^{S}=\sum_{\rho\in\mathrm{S}_n}\epsilon(\rho)\,Z[r_1 s_{\rho(1)}|\ldots|r_n s_{\rho(n)}],
\end{equation}
where the sum is over all permutations of the symmetric group on $n$ objects and $\epsilon(\rho)$ is the signature of $\rho$. The notation $G_{R}^{S}$ refers to the restriction of the Green matrix $G=\Delta^{-1}$ to the rows (resp. columns) indexed by the vertices of $R$ (resp. $S$), where $\Delta$ is the Laplacian of $\G_s$ with Dirichlet boundary conditions at $s$.
\label{AMMT}
\end{thm}


\subsection{Graphs with a connection}
\label{sec2.2}

A complex connection $\Phi$ on a graph $\G$ with edge set $\E$ is defined as a collection of nonzero complex numbers $\phi_{u,v}$ for each oriented edge $(u,v)\in\E$, called \emph{parallel transports}, that satisfy the property $\phi_{v,u}=\phi_{u,v}^{-1}$. The Laplacian with Dirichlet boundary conditions at $s$ associated with the graph $\G_s$ equipped with the connection $\Phi$, called the \emph{line bundle Laplacian} \cite{Ken11}, is denoted by $\mathbf{\Delta}=\mathbf{\Delta}^{(s)}_{\G_s,\Phi}$. It is defined as follows, for $u,v\in\V$:
\begin{equation}
\mathbf{\Delta}_{u,v}=\begin{cases}
\sum_{w:\{u,w\}\in\E\cup\E_s}c_{u,w}&\quad\text{if $v=u$,}\\
-c_{u,v}\,\phi_{u,v}^{-1}&\quad\text{if $\{u,v\}\in\E$,}\\
0&\quad\text{otherwise.}
\end{cases}
\label{LB_Delta}
\end{equation}
The particular choice $\Phi=\mathbf{1}$, that is, $\phi_{u,v}=1$ for all edges $(u,v)\in\E$, is called the \emph{trivial connection}. In that case, $\mathbf{\Delta}$ is simply the usual Dirichlet Laplacian $\Delta$. The analogue of the matrix-tree theorem for graphs with a connection counts combinatorial objects called \emph{oriented cycle-rooted spanning forests}. An OCRSF consists in the union of a single tree containing the root $s$ and \emph{cycle-rooted trees} (CRTs, also known as \emph{unicycles}), which are connected subgraphs containing a single cycle. The tree attached to the root $s$---which can possibly be degenerate or spanning---is naturally oriented toward $s$. In a cycle-rooted tree, all edges in branches point toward the unique cycle, which can be oriented in either direction (so that there is a unique outgoing arrow at each vertex of the CRT). The theorem is the following.
\begin{thm}[\hspace{-0.025cm}\cite{For93}, Theorem 1 and \cite{Ken11}, Theorem 1\hspace{-0.025cm}]
\begin{equation}
\det\mathbf{\Delta}=\sum_{\mathrm{OCRSFs}\,\mathcal{C}}\w(\mathcal{C}),\quad\text{with }\w(\mathcal{C})=\prod_{\{u,v\}\in\mathcal{C}}c_{u,v}\prod_{\text{cycles }\alpha\in\mathcal{C}}(1-\varpi_{\alpha}),
\end{equation}
where $\varpi_{\alpha}$ is the \emph{monodromy} of the cycle $\alpha=(v_0,v_1,\ldots,v_k,v_{k+1}{\equiv}v_0)$ ($k\ge 1$), defined as the product of all parallel transports along $\alpha$:
\begin{equation}
\varpi_{\alpha}=\prod_{i=0}^{k}\phi_{v_i,v_{i+1}}.
\end{equation}
\label{CRSF_thm}
\end{thm}
As the $\phi$'s are complex numbers, the starting point of a cycle is arbitrary. One also notes that cycles of length 2 have vanishing weight, due to the inverse property of parallel transports, i.e. $\phi_{v_0,v_1}=\phi_{v_1,v_0}^{-1}$.

More generally, we shall consider graphs with an even number of marked vertices, which are called \emph{nodes}. On such a graph $\G_s$ with nodes, we define an \emph{oriented cycle-rooted grove}\footnote{Our definition of a cycle-rooted grove differs slightly from that of \cite{KW15}, in which the root $s$ is considered as a node and a tree component contains one or more nodes. Moreover, we choose here to orient each component.} (OCRG) as a spanning subgraph of $\G_s$ whose components are either:
\begin{itemize}
\item a rooted tree containing $s$, in which all edges are oriented toward $s$;
\item a rooted tree containing exactly two nodes, in which all edges point toward one node, which serves as a secondary root in addition to $s$;
\item a cycle-rooted tree not containing any node nor $s$ (see Fig.~\ref{CRG_ex} for an example).
\end{itemize}
\vspace{-2mm}
We denote by $\Gamma_{\vsig}$ an OCRG with $2n$ nodes paired according to the oriented partition $\vsig={\textstyle{s_1\atop r_1}}|\cdots|{\textstyle{s_n\atop r_n}}$, that is, in which each rooted tree (except the one with $s$) contains a path from node $r_i$ to node $s_i$, for $1\le i\le n$. The weight of an OCRG $\Gamma_{\vsig}$ is given by the same formula as for an OCRSF (see Theorem~\ref{CRSF_thm}), with an extra factor accounting for (the inverse of) the product of parallel transports $\phi_{r_i\to s_i}$ along the unique path from node $r_i$ to node $s_i$ in each rooted tree $(1\le i\le n)$:
\begin{equation}
\w(\Gamma_{\vsig})=\prod_{\{u,v\}\in\Gamma_{\vsig}}c_{u,v}\prod_{\text{cycles }\alpha\in\Gamma_{\vsig}}(1-\varpi_{\alpha})\times\prod_{i=1}^{n}\phi_{r_i\to s_i}^{-1}.
\label{grove_weight}
\end{equation}
We use the notation $\Zr[\vsig]$ for the partition function for all OCRGs whose nodes are distributed in rooted trees according to the oriented pairing $\vsig$. They satisfy the following theorem, which generalizes Theorem~\ref{AMMT}, and in which $\Gr=\mathbf{\Delta}^{-1}$ is called the \emph{line bundle Green function}:
\begin{thm}[\hspace{-0.025cm}\cite{KW15}, Theorem 4.4\hspace{-0.025cm}]
\begin{equation}
\det\mathbf{\Delta}\det\Gr_R^S=\sum_{\rho\in\mathrm{S}_n}\epsilon(\rho)\,\Zr\big[{\textstyle{s_{\rho(1)}\atop r_1}}|\cdots|{\textstyle{s_{\rho(n)}\atop r_n}}\big].
\end{equation}
\label{grove_thm}
\end{thm}
\vspace{-0.75cm}
As a matter of notation, we shall write for simplicity $\Gr\textstyle{{s_1,\ldots,s_n}\atop{r_1,\ldots,r_n}}$ for $\Gr_R^S$ when listing the elements of $R,S$ explicitly. A particular case of interest we shall use in later sections consists in a connection depending on a single parameter $z\in\C^*$, namely a collection of parallel transports $\phi_{u,v}$ taking their values in $\{1,z,z^{-1}\}$ for any $(u,v)\in\E$. As we shall see below, Schramm's formula for paths in spanning forests can be expressed as ratios of partition functions for groves in the limit of a trivial connection, $\Phi\to\mathbf{1}$, that is, for $z\to 1$. The main tool for computing such quantities is Theorem \ref{grove_thm}, which requires knowledge of the line bundle Green function $\Gr(z)=\mathbf{\Delta}(z)^{-1}$. Let us recall the first two terms of its perturbative expansion around $z=1$, given in \cite{KW15}:
\begin{equation}
\Gr_{u,v}=G_{u,v}+(z-1)G'_{u,v}+\ldots,\quad\text{with }G'_{u,v}=\sum_{(k,\l):\,\phi_{k,\l}=z}c_{k,\l}\big(G_{u,\l}\,G_{k,v}-G_{u,k}\,G_{\l,v}\big),
\label{Gp_def}
\end{equation}
where $G=\Delta^{-1}$ is the usual Green function of $\G_s$ with Dirichlet boundary conditions at $s$, and $G'$ is called the \emph{derivative of the Green function}. Higher-order terms can be expressed in terms of the standard Green function $G$ as well, but will not be required for the computations presented in this paper.

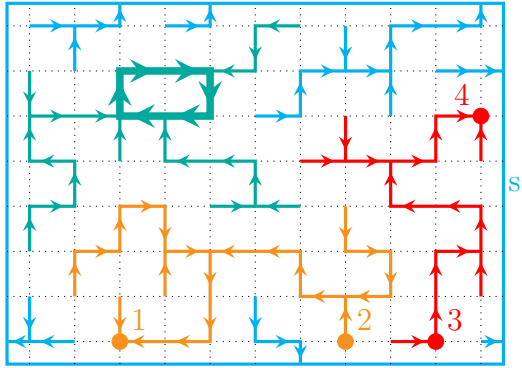
\begin{figure}[h]
\centering
\begin{tikzpicture}[scale=0.60,font=\large,very thick]
\draw[ProcessBlue] (-0.5,-0.5) rectangle (10.5,7.5);
\draw[step=1cm,dotted,thin] (-0.5,-0.5) grid (10.5,7.5);
\node[ProcessBlue] at (10.75,3.5) {\large {s}};
\path[draw,postaction={on each segment={mid arrow}},ProcessBlue] (0,7)--(1,7)--(2,7)--(2,7.5);
\path[draw,postaction={on each segment={mid arrow}},ProcessBlue] (1,6)--(1,7);
\path[draw,postaction={on each segment={mid arrow}},ProcessBlue] (3,7)--(4,7)--(4,7.5);
\path[draw,postaction={on each segment={mid arrow}},ProcessBlue] (1,0)--(0,0)--(-0.5,0);
\path[draw,postaction={on each segment={mid arrow}},ProcessBlue] (0,1)--(0,0);
\path[draw,postaction={on each segment={mid arrow}},ProcessBlue] (5,1)--(5,0)--(6,0)--(6,-0.5);
\path[draw,postaction={on each segment={mid arrow}},ProcessBlue] (10,0)--(10.5,0);
\path[draw,postaction={on each segment={mid arrow}},ProcessBlue] (5,5)--(6,5)--(6,6)--(7,6)--(8,6)--(8,7)--(9,7)--(10,7)--(10,7.5);
\path[draw,postaction={on each segment={mid arrow}},ProcessBlue] (7,7)--(7,6);
\path[draw,postaction={on each segment={mid arrow}},ProcessBlue] (8,5)--(8,6);
\path[draw,postaction={on each segment={mid arrow}},ProcessBlue] (9,6)--(10,6)--(10.5,6);
\path[draw,postaction={on each segment={mid arrow}},line width=0.1cm,Emerald] (3,5)--(2,5)--(2,6)--(3,6)--(4,6)--(4,5)--(3,5);
\path[draw,postaction={on each segment={mid arrow}},Emerald] (0,6)--(0,5)--(1,5)--(2,5);
\path[draw,postaction={on each segment={mid arrow}},Emerald] (2,4)--(2,5);
\path[draw,postaction={on each segment={mid arrow}},Emerald] (0,2)--(0,3)--(1,3)--(1,4)--(0,4)--(0,5);
\path[draw,postaction={on each segment={mid arrow}},Emerald] (6,3)--(5,3)--(5,4)--(4,4)--(3,4)--(3,5);
\path[draw,postaction={on each segment={mid arrow}},Emerald] (6,7)--(5,7)--(5,6)--(4,6)--(4,5)--(3,5);
\path[draw,postaction={on each segment={mid arrow}},Emerald] (4,3)--(5,3);
\path[draw,postaction={on each segment={mid arrow}},BurntOrange] (1,1)--(1,2)--(2,2)--(2,3)--(3,3)--(3,2)--(4,2);
\path[draw,postaction={on each segment={mid arrow}},BurntOrange] (7,3)--(7,2)--(8,2)--(8,1)--(7,1)--(6,1)--(6,2)--(5,2)--(4,2)--(4,1)--(4,0)--(3,0)--(2,0);
\path[draw,postaction={on each segment={mid arrow}},BurntOrange] (2,1)--(2,0);
\path[draw,postaction={on each segment={mid arrow}},BurntOrange] (3,1)--(3,2);
\path[draw,postaction={on each segment={mid arrow}},BurntOrange] (7,0)--(7,1);
\path[draw,postaction={on each segment={mid arrow}},red] (8,0)--(9,0)--(9,1)--(9,2)--(10,2);
\path[draw,postaction={on each segment={mid arrow}},red] (10,1)--(10,2)--(10,3)--(9,3)--(8,3)--(8,4);
\path[draw,postaction={on each segment={mid arrow}},red] (10,4)--(10,5);
\path[draw,postaction={on each segment={mid arrow}},red] (6,4)--(7,4)--(8,4)--(9,4)--(9,5)--(10,5);
\path[draw,postaction={on each segment={mid arrow}},red] (7,5)--(7,4);
\filldraw[BurntOrange] (2,0) circle (0.15cm) node[above right] {$1$};
\filldraw[BurntOrange] (7,0) circle (0.15cm) node[above right] {$2$};
\filldraw[red] (9,0) circle (0.15cm) node[above right] {$3$};
\filldraw[red] (10,5) circle (0.15cm) node[above left] {$4$};
\end{tikzpicture}
\caption{Oriented cycle-rooted grove on a rectangular grid with wired boundary, where the root $s$ is drawn as a box surrounding the grid. The graph contains four nodes paired according to $\vsig={\textstyle{1\atop 2}}\big|{\textstyle{4\atop 3}}$. The grove has four connected components: three trees $\mathcal{T}_s$, $\mathcal{T}_{2\to 1}$, $\mathcal{T}_{3\to 4}$ and a cycle-rooted tree without nodes, whose cycle is highlighted with a heavy line. All the edges in $\mathcal{T}_{2\to 1}$ (resp. $\mathcal{T}_{3\to 4}$) are oriented toward node 1 (resp. node 4).}
\label{CRG_ex}
\end{figure}


\subsection{Graphs with free boundary conditions}
\label{sec2.3}

Up to this point, we have considered a connected graph $\G$ with a root $s$, for which the Dirichlet Laplacian $\Delta$ is invertible. For such a graph, a subset $\mathcal{D}$ of the vertices $\V$ are wired to $s$, while the vertices of $\V\bs\mathcal{D}$ are free. A particular case occurs when $\mathcal{D}=\varnothing$, namely, if all vertices are free (in the scaling limit, this corresponds to uniform Neumann boundary conditions, see Section~\ref{sec4}). As the root $s$ is isolated from all the vertices of $\G$, it may be eliminated altogether. Although this situation may seem the most natural from the graph-theoretical point of view, the singularity of the standard Laplacian $\Delta_0=\Delta_{\G}$ raises technical difficulties, as we shall see below (the subscript 0 serves as a reminder that $\Delta_0$ has a nondegenerate zero eigenvalue).

Let $\left\{|f_{\lambda}\rangle\right\}$ denote the collection of eigenvectors of $\Delta_0$, such that $\Delta_0|f_{\lambda}\rangle=\lambda|f_{\lambda}\rangle$ and $\langle f_\lambda|f_{\lambda}\rangle=1$. Let $V_0=\langle\,|f_0\rangle\,\rangle$ denote the one-dimensional eigenspace associated with the eigenvalue $\lambda_0=0$, and $\mathcal{P}_0=|f_0\rangle\langle f_0|$ the projector onto $V_0$. Explicitly, $(\mathcal{P}_0)_{u,v}=1/\N$ for any $u,v\in\V$, where $\N=|\V|$. We define the \emph{regularized Green function} of the graph $\G$ as
\begin{equation}
G=\sum_{\lambda\neq 0}\frac{1}{\lambda}|f_{\lambda}\rangle\langle f_{\lambda}|,
\label{Gf_reg}
\end{equation}
which satisfies the relations $\Delta_0 G=\mathbb{I}-\mathcal{P}_0=G\Delta_0$. Using Eq.~\eqref{Gf_reg} and the eigendecomposition of $\Delta_0$, one readily finds that the following relations hold for any nonzero parameter $t$,
\begin{equation}
\det(\Delta_0+t\mathcal{P}_0)=t\det\hspace{0mm}'\Delta_0,\quad(\Delta_0+t\mathcal{P}_0)^{-1}=G+t^{-1}\mathcal{P}_0,
\end{equation}
where $\det'\Delta_0$ is the pseudodeterminant of $\Delta_0$ (i.e. the product of its nonzero eigenvalues).

We now show how to adapt Theorems \ref{AMMT} and \ref{grove_thm} to graphs with free boundary conditions, in terms of the regularized Green function $G$. Note that for a generic connection $\Phi$, the line bundle Laplacian $\mathbf{\Delta}_0$ is nonsingular, so its inverse $\Gr_0$ is well defined. It is therefore possible to write a version of Theorem~\ref{grove_thm} in terms of these two matrices. However, it may often be difficult to control the singular behavior of $\Gr_0$ in the limit $\Phi\to\mathbf{1}$. In such a case, we can consider a small perturbation of the line bundle Laplacian, $\mathbf{\Delta}_0+t\mathcal{P}_0$, and denote by $\Gr_t$ its inverse. The analogue of Theorems \ref{AMMT} and \ref{grove_thm} is given by
\begin{prop}
Let $R=\{r_1,\ldots,r_n\}$, $S=\{s_1,\ldots,s_n\}$ be disjoint subsets of vertices of $\G$. We denote by $\left(\Gr_t\right)_R^S$ the submatrix of $\Gr_t$ whose rows and columns are indexed by $R$ and $S$, respectively. Then
\begin{equation}
\lim_{t\to 0}\det(\mathbf{\Delta}_0+t\mathcal{P}_0)\det\left(\Gr_t\right)_R^S=\sum_{\rho\in\mathrm{S}_n}\epsilon(\rho)\,\Zr\big[\textstyle{s_{\rho(1)}\atop r_1}|\cdots|\textstyle{s_{\rho(n)}\atop r_n}\big].
\label{grove_eq_free}
\end{equation}
In the limit $\Phi\to\mathbf{1}$, the preceding equation reads
\begin{equation}
\frac{1}{\N}\det\hspace{0mm}'\Delta_0\times\det\begin{pmatrix}1 & G_{r_1,s_2}{-}G_{r_1,s_1} & \cdots & G_{r_1,s_n}{-}G_{r_1,s_1}\\ \vdots & \vdots & & \vdots\\ 1 & G_{r_n,s_2}{-}G_{r_n,s_1} & \cdots & G_{r_n,s_n}{-}G_{r_n,s_1}\end{pmatrix}=\sum_{\rho\in\mathrm{S}_n}\epsilon(\rho)\,Z[r_1 s_{\rho(1)}|\ldots|r_n s_{\rho(n)}],
\end{equation}
where $\det\hspace{0mm}'\Delta_0/\N$ is equal to the weighted sum of spanning trees on $\G$ by Kirchhoff's theorem.
\label{grove_thm2}
\end{prop}
\vspace{-0.75cm}
\begin{proof}
Using Cramer's formula for minors, we find that
\begin{equation}
\lim_{t\to 0}\det(\mathbf{\Delta}_0+t\mathcal{P}_0)\det\left(\Gr_t\right)_R^S=\lim_{t\to 0}(-1)^{\Sigma R+\Sigma S}\det(\mathbf{\Delta}_0+t\mathcal{P}_0)_{\V\bs S}^{\V\bs R}=(-1)^{\Sigma R+\Sigma S}\det(\mathbf{\Delta}_0)_{\V\bs S}^{\V\bs R},
\end{equation}
where $\Sigma R$ (resp. $\Sigma S$) denotes the sum of the indices of the columns of $\mathbf{\Delta}_0+t\mathcal{P}_0$ indexed by elements of $R$ (resp. $S$). The latter (signed) determinant corresponds precisely to the right-hand side of Eq.~\eqref{grove_eq_free} (see Theorem 4.4 in \cite{KW15}). In the limit $\Phi\to\mathbf{1}$, the left-hand side of Eq.~\eqref{grove_eq_free} yields
\begin{equation*}
\lim_{t\to 0}\det\hspace{0mm}'\Delta_0\times t\det(G+t^{-1}\mathcal{P}_0)_R^S\,,
\end{equation*}
which can be evaluated by subtracting the first column of $(G+t^{-1}\mathcal{P}_0)_R^S$ from its other columns to isolate the $t^{-1}$ dependence in the first column only (recall that $(\mathcal{P}_0)_{u,v}=1/\N$ for any $u,v\in\V$).
\end{proof}

For concrete applications in the next two sections, we shall use Proposition~\ref{grove_thm2} for a connection $\Phi$ with parallel transports $1$, $z$ or $z^{-1}$. In particular, we shall need to compute the first-order expansion of the line bundle Green function $\Gr_t$ around $z=1$, which reads
\begin{equation}
\begin{split}
\left(\Gr_t\right)_{u,v}&=\left(G_{u,v}+t^{-1}\N^{-1}\right)+(z-1)\sum_{(k,\l):\,\phi_{k,\l}=z}c_{k,\l}\big(G_{u,\l}\,G_{k,v}-G_{u,k}\,G_{\l,v}\big)\\
&\quad+(z-1)t^{-1}\N^{-1}\sum_{(k,\l):\,\phi_{k,\l}=z}c_{k,\l}\big(G_{u,\l}+G_{k,v}-G_{u,k}-G_{\l,v}\big)+\ldots
\end{split}
\end{equation}
We write the first-order derivative of $\Gr_t$ as follows: $\partial_z\Gr_t|_{z=1}=G'+t^{-1}\N^{-1}\widetilde{G}'$, with
\begin{equation}
\begin{split}
G'_{u,v}&=\sum_{(k,\l):\,\phi_{k,\l}=z}c_{k,\l}\big(G_{u,\l}\,G_{k,v}-G_{u,k}\,G_{\l,v}\big),\\
\widetilde{G}'_{u,v}&=\sum_{(k,\l):\,\phi_{k,\l}=z}c_{k,\l}\big(G_{u,\l}+G_{k,v}-G_{u,k}-G_{\l,v}\big).
\label{Gp_def2}
\end{split}
\end{equation}


\section{Paths in spanning forests and Schramm's formula}
\label{sec3}

In this section, we first define a measure on chemical paths in spanning forests between two fixed vertices $u_1,u_2$ of an unoriented connected graph, and then extend it to paths in oriented cycle-rooted groves. For $u_1,u_2$ located on the boundary of the graph, we study the following question about a random simple path from $u_1$ to $u_2$ with respect to the measure on forests: What is the probability that such a path leaves a marked face $f$ to its left, as illustrated in Fig.~\ref{Prb}? In what follows, we give an explicit combinatorial expression for this probability, called Schramm's formula, by taking the limit $\Phi\to\mathbf{1}$ of the corresponding probability in the grove model, in a nontrivial way. Our formula depends only on the standard Green function, which is well known for regular graphs such as those we shall consider in Section \ref{sec4}.

\begin{figure}[h]
\centering
\begin{tikzpicture}[scale=1,font=\small]

\begin{scope}
\draw[step=0.5cm,densely dotted] (-5,-3) grid (2,0);
\draw[densely dotted,fill=Aquamarine!50!white] (-0.5,-2) rectangle (0,-1.5);
\node[blue] at (-0.25,-1.75) {\Large $f$};
\draw[ultra thick,red] (-3,-3)--(-3,-2.5)--(-2.5,-2.5)--(-2.5,-1.5)--(-1,-1.5)--(-1,-2)--(-1.5,-2)--(-1.5,-2.5)--(-1,-2.5)--(-1,-3)--(0,-3);
\draw[-latex,ultra thick,red] (-2.5,-2)--(-2.5,-1.75);
\draw[ultra thick,blue] (-3,-3)--(-3.5,-3)--(-3.5,-1.5)--(-4,-1.5)--(-4,-1)--(-2,-1)--(-2,-0.5)--(-0.5,-0.5)--(-0.5,-1)--(0.5,-1)--(0.5,-1.5)--(1,-1.5)--(1,-2)--(0.5,-2)--(0.5,-2.5)--(0,-2.5)--(0,-3);
\draw[-latex,ultra thick,blue] (-2.5,-1)--(-2.125,-1);
\filldraw (-3,-3) circle (0.1cm);
\node at (-3,-3.375) {$u_1$};
\filldraw (0,-3) circle (0.1cm);
\node at (0,-3.375) {$u_2$};
\node[red] at (-2.25,-2.25) {\large L};
\node[blue] at (1.25,-1.75) {\large R};
\end{scope}

\begin{scope}[scale=0.8,shift={(8,-2)}]
\filldraw[Aquamarine!50!white] (0,0) circle (1cm);
\foreach \r in {1,1.5,...,2.5}{\foreach \t in {0,30,...,330}{
\draw[help lines,densely dotted,thick] (0,0) circle (\r);
\filldraw (\t:\r) circle (0.025cm);
\draw[help lines,dotted] (\t:1)--(\t:2.5);}}
\draw[very thick,blue] (210:2.5)to[out=120,in=270](180:2.5)to(180:1.5)to[out=90,in=240](150:1.5)to(150:2)to[out=60,in=180](90:2)to(90:1.5)to[out=0,in=150](60:1.5)to(60:1)to[out=330,in=90](0:1)to(0:2)to[out=270,in=60](330:2)to(330:2.5);
\draw[very thick,red] (210:2.5)to(210:1)to[out=300,in=150](240:1)to(240:2)to[out=330,in=180](270:2)to(270:2.5)to[out=0,in=210](300:2.5)to[out=30,in=240](330:2.5);
\draw[line width=0.06cm,-latex,blue] (180:2)--(180:1.75);
\draw[line width=0.06cm,-latex,red] (240:1.5)--(240:1.75);
\filldraw (210:2.5) circle (0.075cm) node[below left] {$u_1$};
\filldraw (330:2.5) circle (0.075cm) node[below right] {$u_2$};
\node[blue] at (135:2.9) {CW};
\node[red] at (270:2.9) {CCW};
\end{scope}

\end{tikzpicture}
\caption{On the left are two simple paths on a rectangular grid between two boundary vertices, from $u_1$ to $u_2$; the red (resp. blue) path leaving the marked face $f$ to its left (resp. right). On the right is a counterclockwise (resp. clockwise) path on an annulus, which can be viewed as leaving the central circular face to its left (resp. right).}
\label{Prb}
\end{figure}
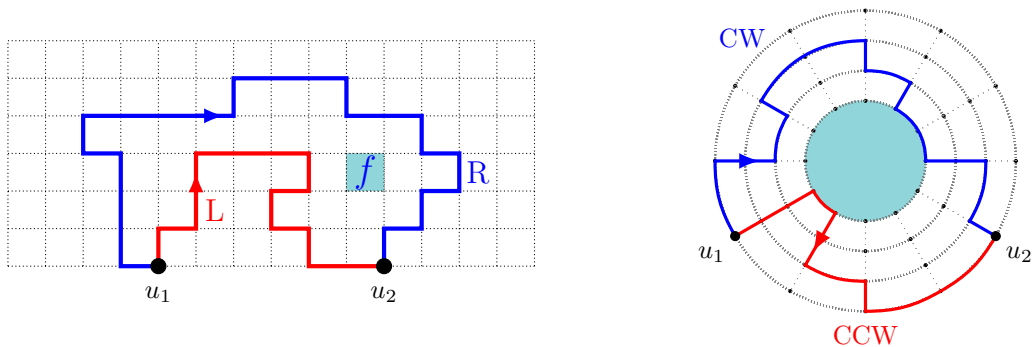


\subsection{Measures on paths in forests and groves}
\label{sec3.1}

Let $\G_s=\G\cup\{s\}$ be an unoriented connected graph with a root $s$ and let $u_1,u_2\neq s$ be two of its vertices. We consider the set of all simple unoriented paths on $\G$ between $u_1$ and $u_2$ (a path is simple if it has no self-intersection). We define the weight of such a path $\gamma:u_1\leftrightarrow u_2$ as the weighted sum of two-component spanning forests on $\G_s$ in which the entire path $\gamma$ belongs to one tree, and the vertex $s$ to the other:
\begin{equation}
\w_{\text{SF}}(\gamma)=\sum_{\text{2SFs $\mathcal{F}\supset\gamma$}}\w(\mathcal{F})=C(\gamma)\det\Delta^{(\gamma)},
\label{SF_weight}
\end{equation}
where $C(\gamma)$ is the product of all the conductances along $\gamma$. Here $\Delta^{(\gamma)}$ is the submatrix of the Dirichlet Laplacian $\Delta=\Delta_{\G_s}^{(s)}$, in which the rows and columns indexed by the vertices of $\gamma$ have been removed. The second equality in \eqref{SF_weight} stems from the fact that 2-component spanning forests on $\G_s$ that contain a fixed path $\gamma:u_1\leftrightarrow u_2$ of length $n$ are in one-to-one correspondence with $(n{+}2)$-component spanning forests on $\G_{s,\gamma}$, in which all the vertices of $\gamma$ are chosen as roots alongside $s$ (see Fig.~\ref{2SF_ex}). The partition function for all paths between $u_1$ and $u_2$ in two-component spanning forests is given by $Z[12]\equiv Z[u_1 u_2]$ as defined in Section~\ref{sec2.1}, and reads by virtue of Theorem~\ref{AMMT}
\begin{equation}
Z[12]=\sum_{\gamma:u_1\leftrightarrow u_2}\w_{\text{SF}}(\gamma)=G_{1,2}\det\Delta,
\label{Z12}
\end{equation}
where $G_{1,2}\equiv G_{u_1,u_2}=(\Delta^{-1})_{u_1,u_2}$. The probability distribution on simple paths in two-component spanning forests is therefore given by
\begin{equation}
\P_{\text{SF}}(\gamma:u_1\leftrightarrow u_2)=\frac{\w_{\text{SF}}(\gamma)}{\sum_{\gamma:u_1\leftrightarrow u_2}\w_{\text{SF}}(\gamma)}=\frac{C(\gamma)\det\Delta^{(\gamma)}}{G_{1,2}\det\Delta}.
\label{P_SF}
\end{equation}
In what follows, we shall consider \emph{oriented} simple paths on $\G$ \emph{from} $u_1$ \emph{to} $u_2$. It is therefore natural to define a probability measure on such paths as follows:
\begin{equation}
\P_{\text{SF}}(\gamma:u_1\to u_2)=\P_{\text{SF}}(\gamma^*:u_1\leftrightarrow u_2),
\end{equation}
where $\gamma$ is obtained from $\gamma^*$ by orienting all of its edges toward $u_2$. Note that this implies that an oriented path $\gamma$ and the reverse path $\gamma^{-1}$ have the same probability.

\begin{figure}[h]
\centering
\begin{tikzpicture}[scale=0.5,font=\large]

\begin{scope}[very thick,blue!45!white]
\draw (-0.5,-0.5) rectangle (10.5,7.5);
\draw[step=1cm,densely dotted,line width=0.7pt] (-0.5,-0.5) grid (10.5,7.5);
\node[blue] at (10.8,3.5) {\large {s}};
\draw (2,7.5)--(2,7)--(0,7);
\draw (1,7)--(1,6);
\draw (4,7.5)--(4,7)--(3,7)--(3,6);
\draw (-0.5,0)--(1,0);
\draw (0,0)--(0,1);
\draw (10.5,0)--(10,0);
\draw (10,7.5)--(10,7)--(8,7)--(8,6)--(6,6)--(6,5)--(5,5);
\draw (10.5,6)--(9,6);
\draw (8,5)--(8,6);
\draw (7,7)--(7,6);
\draw (2,5)--(2,6)--(4,6);
\draw (2,4)--(2,5)--(0,5)--(0,6);
\draw (0,5)--(0,4)--(1,4)--(1,3)--(0,3)--(0,2);
\draw (3,5)--(3,4)--(5,4)--(5,3)--(6,3);
\draw (3,5)--(4,5)--(4,6)--(5,6)--(5,7)--(6,7);
\draw (4,3)--(5,3);
\draw[ForestGreen] (6,0)--(5,0)--(5,1)--(4,1);
\draw[ForestGreen] (1,1)--(1,2)--(2,2)--(2,3)--(3,3)--(3,2)--(4,2);
\draw[ForestGreen] (2,1)--(2,0);
\draw[ForestGreen] (3,1)--(3,2);
\draw[ForestGreen] (7,0)--(7,1);
\draw[ForestGreen] (8,0)--(9,0)--(9,1)--(9,2)--(10,2);
\draw[ForestGreen] (10,1)--(10,2)--(10,3)--(9,3)--(8,3)--(8,4);
\draw[ForestGreen] (10,4)--(10,5);
\draw[ForestGreen] (7,5)--(7,4);
\draw[ForestGreen] (6,4)--(7,4);
\draw[very thick,dashed,ForestGreen] (2,0)--(3,0)--(4,0)--(4,1)--(4,2)--(5,2)--(6,2)--(6,1)--(7,1)--(8,1)--(8,2)--(7,2)--(7,3)--(7,4)--(8,4)--(9,4)--(9,5)--(10,5);
\filldraw[ForestGreen] (2,0) circle (0.15cm) node[above left] {$u_1$};
\filldraw[ForestGreen] (10,5) circle (0.15cm) node[above left] {$u_2$};
\end{scope}

\begin{scope}[xshift=15cm,very thick,blue!45!white]
\draw (-0.5,-0.5) rectangle (10.5,7.5);
\draw[step=1cm,densely dotted,line width=0.7pt] (-0.5,-0.5) grid (10.5,7.5);
\node[blue] at (10.8,3.5) {\large {s}};
\draw (2,7.5)--(2,7)--(0,7);
\draw (1,7)--(1,6);
\draw (4,7.5)--(4,7)--(3,7)--(3,6);
\draw (-0.5,0)--(1,0);
\draw (0,0)--(0,1);
\draw (10.5,0)--(10,0);
\draw (10,7.5)--(10,7)--(8,7)--(8,6)--(6,6)--(6,5)--(5,5);
\draw (10.5,6)--(9,6);
\draw (8,5)--(8,6);
\draw (7,7)--(7,6);
\draw (2,5)--(2,6)--(4,6);
\draw (2,4)--(2,5)--(0,5)--(0,6);
\draw (0,5)--(0,4)--(1,4)--(1,3)--(0,3)--(0,2);
\draw (3,5)--(3,4)--(5,4)--(5,3)--(6,3);
\draw (3,5)--(4,5)--(4,6)--(5,6)--(5,7)--(6,7);
\draw (4,3)--(5,3);
\draw[ForestGreen] (6,0)--(5,0)--(5,1)--(4,1);
\draw[ForestGreen] (1,1)--(1,2)--(2,2)--(2,3)--(3,3)--(3,2)--(4,2);
\draw[ForestGreen] (2,1)--(2,0);
\draw[ForestGreen] (3,1)--(3,2);
\draw[ForestGreen] (7,0)--(7,1);
\draw[ForestGreen] (8,0)--(9,0)--(9,1)--(9,2)--(10,2);
\draw[ForestGreen] (10,1)--(10,2)--(10,3)--(9,3)--(8,3)--(8,4);
\draw[ForestGreen] (10,4)--(10,5);
\draw[ForestGreen] (7,5)--(7,4);
\draw[ForestGreen] (6,4)--(7,4);
\foreach \pos in {(3,0),(4,0),(4,1),(4,2),(5,2),(6,2),(6,1),(7,1),(8,1),(8,2),(7,2),(7,3),(7,4),(8,4),(9,4),(9,5)}{\filldraw[ForestGreen] \pos circle (0.125cm) ;}
\filldraw[ForestGreen] (2,0) circle (0.15cm) node[above left] {$u_1$};
\filldraw[ForestGreen] (10,5) circle (0.15cm) node[above left] {$u_2$};
\end{scope}

\end{tikzpicture}
\caption{On the left is a two-component spanning forest on a wired grid. Two selected vertices $u_1,u_2$ belong to a tree distinct from the one that includes the root $s$ (represented by the box surrounding the grid). The dashed line highlights the path $\gamma$ between $u_1$ and $u_2$, which contains $n{+}1$ vertices. On the right is the associated $(n{+}2)$-component spanning forest in which each vertex of $\gamma$ is the root of a (possibly degenerate) tree, in addition to $s$.}
\label{2SF_ex}
\end{figure}
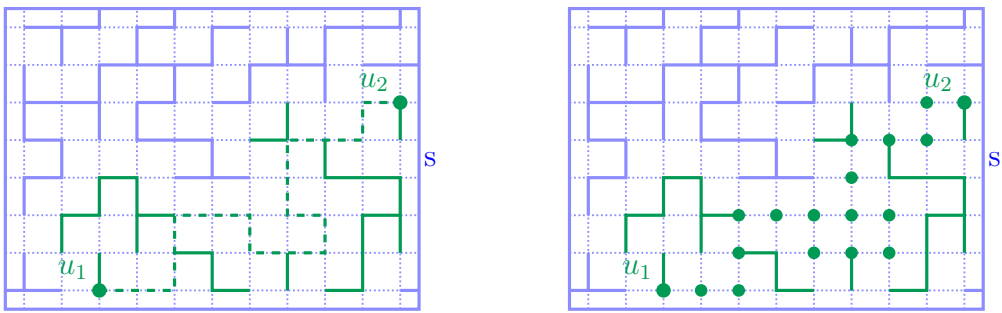

Let us now turn to paths in oriented cycle-rooted groves (OCRGs), for which we define the weight of the oriented simple path $\gamma:u_1\to u_2$ as
\begin{equation}
\w_{\text{CRG}}(\gamma)=\sum_{\text{OCRGs }\Gamma_{\vsig}\supset\gamma}\w(\Gamma_{\vsig})=C(\gamma)\phi(\gamma)^{-1}\det\mathbf{\Delta}^{(\gamma)},
\label{CRG_weight}
\end{equation}
where $\vsig={\textstyle{2\atop 1}}$, $\w(\Gamma_{\vsig})$ is given by Eq.~\eqref{grove_weight} and $\phi(\gamma)=\phi_{u_1\to u_2}$ is the product of all parallel transports along $\gamma$. Explicitly, if $\gamma=(v_i)_{0\le i\le n}$ with $v_0\equiv u_1$ and $v_n\equiv u_2$,
\begin{equation}
\phi(\gamma)=\phi_{u_1\to u_2}=\prod_{i=0}^{n-1}\phi_{v_i,v_{i+1}}.
\end{equation}
Here $\det\mathbf{\Delta}^{(\gamma)}$ is the weighted sum of OCRGs consisting in $n{+}2$ rooted trees, each including a unique vertex of $\gamma$ or $s$, and cycle-rooted trees not intersecting $\gamma\cup\{s\}$. Using Theorem \ref{grove_thm}, we find that the partition function for all paths from $u_1$ to $u_2$ is $\Zr[{\textstyle{2\atop 1}}]=\Gr_{1,2}\det\mathbf{\Delta}$, where $\Gr_{1,2}\equiv\Gr_{u_1,u_2}=(\mathbf{\Delta}^{-1})_{u_1,u_2}$. The normalized distribution yields
\begin{equation}
\P_{\text{CRG}}(\gamma:u_1\to u_2)=\frac{\w_{\text{CRG}}(\gamma)}{\sum_{\gamma:u_1\to u_2}\w_{\text{CRG}}(\gamma)}=\frac{C(\gamma)\phi(\gamma)^{-1}\det\mathbf{\Delta}^{(\gamma)}}{\Gr_{1,2}\det\mathbf{\Delta}}.
\label{P_CRG}
\end{equation}
It should be noted that $\P_{\text{CRG}}$ is \emph{not} a probability distribution for a generic connection, as it can take complex or negative values. Furthermore, $\P_{\text{CRG}}(\gamma)\neq\P_{\text{CRG}}(\gamma^{-1})$ since $\phi(\gamma)\neq\phi(\gamma^{-1})$ and $\Gr_{1,2}\neq\Gr_{2,1}$ in general. It is straightforward to establish the relation between Eqs. \eqref{P_SF} and \eqref{P_CRG} by taking the limit of a trivial connection, i.e. $\Phi\to\mathbf{1}$:
\begin{equation}
\P_{\text{SF}}(\gamma:u_1\to u_2)=\lim_{\Phi\to\mathbf{1}}\P_{\text{CRG}}(\gamma:u_1\to u_2).
\label{lim_P_CRG}
\end{equation}

Similar formulas hold for a graph $\G$ with free boundary conditions, as discussed in Section~\ref{sec2.3}. In that case, the Dirichlet Laplacians $\Delta$ and $\mathbf{\Delta}$ are replaced with the Laplacians $\Delta_0$ and $\mathbf{\Delta}_0$, and the partition functions are given by
\begin{equation}
\Zr[{\textstyle{2\atop 1}}]=\lim_{t\to 0}(\Gr_t)_{1,2}\det(\mathbf{\Delta_0}+t\mathcal{P}_0),\quad Z[12]=\frac{1}{\N}\det\hspace{0mm}'\Delta_0.
\end{equation}


\subsection{Schramm's formula}
\label{sec3.2}

Let us now address the first objective of this paper, namely writing a discrete version of Schramm's formula for paths in spanning forests on a graph with respect to a marked face $f$, in terms of the standard Green function of the graph. Such a formula was first given in \cite{Ken11,KW15} for graphs with free boundary conditions. We extend it here for graphs connected to a root, using the same technique as in \cite{Ken11,KW15}. We shall follow similar steps when discussing multiple paths later on in Section~\ref{sec5}.

In what follows, we consider a graph $\G=(\V,\E)$ embedded on a surface $\Sigma$ such that the edges in $\E$ do not intersect each other (except possibly at the vertices in $\V$). We connect a subset $\mathcal{D}$ of its boundary vertices to a root $s$ with a set of edges $\E_s$ to form the extended graph $\G_s$. Note that we do not require that $s$ or the edges in $\E_s$ belong to $\Sigma$. We select two boundary vertices (possibly in $\mathcal{D}$) as the nodes $u_1,u_2$ and a face $f$ of the graph\footnote{The reason we consider paths between boundary vertices is twofold. First, this situation is the discrete analogue of \emph{chordal} SLEs, for which explicit results for Schramm's formula have been computed. To our knowledge, the equivalent formula for radial or whole-plane SLEs is not currently known. Second, there are (infinitely) many inequivalent classes of paths between two bulk vertices with respect to a face, instead of just two when $u_1,u_2$ are on the boundary. Computing their respective probabilities using a connection as presented below would require knowledge of the full line bundle Green function, as opposed to the perturbative expansion used in most of this work.}. If $\Sigma$ is orientable, the discrete analogue of Schramm's formula \eqref{Schramm_SLE} for paths in spanning forests reads
\begin{equation}
\P_{\L}(u_1,u_2;f)=\sum_{\gamma:u_1\to u_2}\chi_{\L}(\gamma;f)\,\P_{\text{SF}}(\gamma:u_1\to u_2),
\end{equation}
where $\chi_{\L}(\gamma;f)=1$ (resp. 0) if $\gamma$ leaves $f$ to its left (resp. right), and $\P_{\text{SF}}(\gamma:u_1\to u_2)$ is given by Eq.~\eqref{P_SF}. Alternatively, if we denote by $Z_{\L}[12]$ (resp. $Z_{\R}[12]$) the weighted sum of two-component spanning forests on $\G_s$ in which the path from $u_1$ to $u_2$ leaves $f$ to its left (resp. right), the preceding equation may be rewritten as
\begin{equation}
\P_{\L}(u_1,u_2;f)=\frac{Z_{\L}[12]}{Z[12]}.
\label{Schramm_form1}
\end{equation}
While the denominator is already known \eqref{Z12}, a subtler approach is required to extract $Z_{\L}[12]$ from $Z[12]=Z_{\L}[12]+Z_{\R}[12]$ \cite{Ken11}. To do so, let us equip $\G_s$ with a connection that is trivial everywhere except on a collection of edges $\{k,\l\}$ crossed by a \emph{zipper}, that is, a path from $f$ to the outer face of $\G$ on the dual graph. We impose for convenience that the zipper intersects the clockwise boundary path from $u_1$ to $u_2$ (see Fig.~\ref{zip_ex}). We put a constant parallel transport $\phi_{k,\l}=z\in\C^*$ on the oriented edges $(k,\l)$---and $\phi_{\l,k}=z^{-1}$ in the opposite direction---in such a way that a counterclockwise cycle circling $f$ has monodromy $z$.

This specific choice of a connection allows one to distinguish between paths from $u_1$ to $u_2$ that leave $f$ to their left and to their right. Indeed, the product of parallel transports from $u_1$ to $u_2$ appearing in Eq.~\eqref{CRG_weight} is $\phi_{1\to 2}=1$ for the former, while $\phi_{1\to 2}=z^{-1}$ for the latter. We can therefore decompose the full partition function for groves with a path from $u_1$ to $u_2$ as follows:
\begin{equation}
\Zr[{\textstyle{2\atop 1}}]=\Zr_{\L}[{\textstyle{2\atop 1}}]+\Zr_{\R}[{\textstyle{2\atop 1}}],
\label{Zr_decomp}
\end{equation}
where $\phi_{1\to 2}$ is constant over the paths of each class, left and right.


\begin{figure}[h]
\centering
\begin{tikzpicture}[scale=0.7]
\draw[gray,densely dotted,line width=0.7pt] (-0.5,-0.5) grid (10.5,5.5);
\draw (-0.5,-0.5) rectangle (10.5,5.5);
\node at (10.75,2.5) {$s$};
\draw[gray,densely dotted,fill=Aquamarine!50!white] (6,3) rectangle (7,4);
\node[Aquamarine] at (6.5,2.625) {\large $f$};
\draw[line width=1pt,red] (2,0)--(3,0)--(3,1)--(2,1)--(2,2)--(3,2)--(3,3)--(4,3)--(4,2)--(6,2)--(6,1)--(5,1)--(5,0)--(7,0)--(7,1)--(8,1)--(8,0);
\draw[line width=0.06cm,-latex,blue] (0.5,3)--(0.75,3);
\node[red] at (4.5,1.5) {$\L$};
\draw[line width=0.06cm,-latex,red] (4,2.5)--(4,2.25);
\draw[line width=1pt,blue] (2,0)--(1,0)--(1,1)--(0,1)--(0,3)--(1,3)--(1,4)--(2,4)--(2,5)--(4,5)--(4,4)--(5,4)--(5,3)--(6,3)--(6,4)--(6,5)--(9,5)--(9,2)--(10,2)--(10,0)--(8,0);
\node[blue] at (1.5,3.5) {$\R$};
\filldraw (2,0) circle (0.15cm);
\filldraw (8,0) circle (0.15cm);
\node at (2,-1) {\large $u_1$};
\node at (8,-1) {\large $u_2$};
\filldraw[Aquamarine] (6.5,3.5) circle (0.125cm);
\draw[thick,Aquamarine] (6.5,3.5)--(10.5,3.5);
\foreach \x in {7,...,10}{\draw[very thick,dashed,-latex,Aquamarine] (\x,3)--(\x,4) node[above left] {$z$};}
\end{tikzpicture}
\caption{Wired rectangular grid with two marked boundary vertices $u_1,u_2$ and a marked face $f$. The solid line from $f$ to the outer face represents the zipper. The edges crossed by the zipper possess a parallel transport $z$ from the bottom up (and $z^{-1}$ from the top down). The path labeled $\L$ (resp. $\R$) from $u_1$ to $u_2$ leaves $f$ to its left (resp. right).} 
\label{zip_ex}
\end{figure}
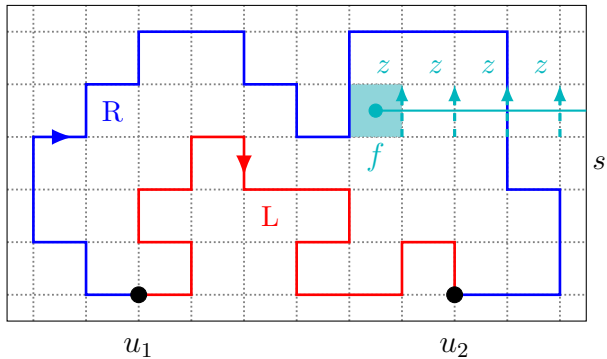

If we consider instead paths in the opposite direction, we have $\phi_{2\to 1}=1$ (resp. $\phi_{2\to 1}=z$) for paths from $u_2$ to $u_1$ that leave the face $f$ to their right (resp. left). Similarly to Eq.~\eqref{Zr_decomp}, we write
\begin{equation}
\Zr[{\textstyle{1\atop 2}}]=\Zr_{\L}[{\textstyle{1\atop 2}}]+\Zr_{\R}[{\textstyle{1\atop 2}}]
\end{equation}
for paths from $u_2$ to $u_1$, where the indices $\L,\R$ refer to the left- or right-passage with respect to $f$ in the original direction, i.e. \emph{from $u_1$ to $u_2$}. Looking at Eqs. \eqref{grove_weight} and \eqref{CRG_weight}, we see that the weight of a path $\gamma$ in OCRGs is related to that of its reverse by
\begin{equation}
\w_{\text{CRG}}(\gamma)\,\phi(\gamma)=\w_{\text{CRG}}(\gamma^{-1})\,\phi(\gamma^{-1}).
\end{equation}
As the product of parallel transports along the path from $u_1$ to $u_2$ is constant over the left and right classes separately, we find the relations
\begin{equation}
\Zr_{\L}[{\textstyle{1\atop 2}}]=\Zr_{\L}[{\textstyle{2\atop 1}}],\quad\Zr_{\R}[{\textstyle{1\atop 2}}]=z^{-2}\Zr_{\R}[{\textstyle{2\atop 1}}],
\end{equation}
which allow us to write and solve the following system for $\Zr_{\L,\R}[{\textstyle{2\atop 1}}]$:
\begin{equation}
\begin{split}
&\Gr_{1,2}\det\mathbf{\Delta}=\Zr[{\textstyle{2\atop 1}}]=\Zr_{\L}[{\textstyle{2\atop 1}}]+\Zr_{\R}[{\textstyle{2\atop 1}}],\\
&\Gr_{2,1}\det\mathbf{\Delta}=\Zr[{\textstyle{1\atop 2}}]=\Zr_{\L}[{\textstyle{2\atop 1}}]+z^{-2}\Zr_{\R}[{\textstyle{2\atop 1}}].
\end{split}
\end{equation}
The solution reads
\begin{equation}
\Zr_{\L}[{\textstyle{2\atop 1}}]=\frac{\left(\Gr_{2,1}-z^{-2}\Gr_{1,2}\right)\det\mathbf{\Delta}}{1-z^{-2}},\quad\Zr_{\R}[{\textstyle{2\atop 1}}]=\frac{\left(\Gr_{1,2}-\Gr_{2,1}\right)\det\mathbf{\Delta}}{1-z^{-2}}.
\end{equation}
Since the normalized measure $\P_{\text{CRG}}$ converges to $\P_{\text{SF}}$ in the limit of a trivial connection \eqref{lim_P_CRG}, that is, when $z\to 1$, one obtains the following combinatorial result for Schramm's formula \eqref{Schramm_form1}:
\begin{equation}
\P_{\L}(u_1,u_2;f)=\lim_{z\to 1}\frac{\left(\Gr_{2,1}-z^{-2}\Gr_{1,2}\right)\det\mathbf{\Delta}}{1-z^{-2}}\times\frac{1}{\Gr_{1,2}\det\mathbf{\Delta}}=1-\frac{G'_{1,2}}{G_{1,2}},
\label{Schramm_form2}
\end{equation}
where we have used the antisymmetry of the derivative of the Green function $G'$ defined in \eqref{Gp_def}.

The corresponding formula for graphs with free boundary conditions was given in \cite{Ken11}, where Cramer's rule was used to write
\begin{equation}
(\Gr_0)_{1,2}=\frac{\kappa}{\det\mathbf{\Delta}_0}\left\{1+(z-1)\,\widetilde{G}'_{1,2}+\ldots\right\},
\end{equation}
where $\widetilde{G}'$ is defined in \eqref{Gp_def2} and $\kappa=Z[12]$ denotes the weighted sum of spanning trees on $\G$. We present here an alternative derivation using the modified line bundle Green function $\Gr_t=(\mathbf{\Delta}_0+t\mathcal{P}_0)^{-1}$ instead of $\Gr_0=(\mathbf{\Delta}_0)^{-1}$:  
\begin{equation}
\begin{split}
\P_{\L}(u_1,u_2;f)&=\lim_{t\to 0}\lim_{z\to 1}\frac{\left((\Gr_t)_{2,1}-z^{-2}(\Gr_t)_{1,2}\right)\det(\mathbf{\Delta}_0+t\mathcal{P}_0)}{1-z^{-2}}\times\frac{1}{(\Gr_t)_{1,2}\det(\mathbf{\Delta}_0+t\mathcal{P}_0)}\\
&=1-\lim_{t\to 0}\frac{G'_{1,2}+t^{-1}\N^{-1}\widetilde{G}'_{1,2}}{G_{1,2}+t^{-1}\N^{-1}}=1-\widetilde{G}'_{1,2},
\end{split}
\label{Schramm_form3}
\end{equation}
where the properties $G'_{v,u}=-G'_{u,v}$ and $\widetilde{G}'_{v,u}=-\widetilde{G}'_{u,v}$ follow from Eq.~\eqref{Gp_def2}. A third way of obtaining Eq.~\eqref{Schramm_form3} consists in connecting $u_2$ to the root $s$, that is, defining a graph $\xbar{\G}$ such that $\xbar{\mathbf{\Delta}}_{u,v}=(\mathbf{\Delta}_0)_{u,v}+\delta_{u,2}\delta_{2,v}$. The OCRGs with a path from $u_1$ to $u_2$ on $\G$ are in one-to-one correspondence with those on $\xbar{\G}$ (in which the tree rooted at $s$ is degenerate, since it cannot contain $u_2$). Moreover, the inverse $\xbar{G}=\xbar{\Delta}^{-1}$ is well defined (due to the connection between $u_2$ and the root $s$), so Schramm's formula is given by Eq.~\eqref{Schramm_form2}:
\begin{equation}
\P_{\L}(u_1,u_2;f)=1-\frac{\xbar{G}'_{1,2}}{\xbar{G}_{1,2}},
\end{equation}
with $\xbar{G}'$ defined in \eqref{Gp_def} (in terms of $\xbar{G}$ instead of $G$). Writing the Green function $\xbar{G}$ in terms of the regularized Green function $G$ \eqref{Gf_reg} on $\G$ yields
\begin{equation}
\xbar{G}_{u,v}=1+G_{u,v}+G_{2,2}-G_{u,2}-G_{2,v}.
\end{equation}
One recovers Eq.~\eqref{Schramm_form3} by plugging this relation into the preceding equation.


\section{Explicit results for passage probabilities}
\label{sec4}

In this section, we consider rectangular grids embedded on surfaces, with unit conductances: $c_{u,v}=1$ for any $\{u,v\}\in\E$, for various combinations of wired and free boundary conditions. We compute the discrete Green function $G$ of the graph, and show that it converges in the scaling limit to the continuum Green function (up to a constant normalization factor), which we use to give an explicit expression for (the scaling limit of) $\P_{\L}(u_1,u_2;f)$. For the upper half-plane and the cylinder, we compare our results with known $\text{SLE}_2$ probabilities.


\subsection{The upper half-plane}
\label{sec4.1}

Let us first compute Schramm's formula on the discrete upper half-plane (UHP) $\G=\Z\times\mathbb{N}^*$, whose boundary consists in vertices of the form $(\x,\y{=}1)$, $\x\in\Z$. We shall consider uniform boundary conditions on such a graph, either wired or free. In the former case, all boundary vertices $(\x,1)$ are connected by an edge to a root $s$, so their degree in $\G_s$ is 4 like the bulk vertices; while in the latter case, the boundary vertices have only degree 3.

Anticipating a bit, we denote by $G^{\text{D}}_{u,v}$ and $G^{\text{N}}_{u,v}$ the discrete Green functions of the UHP with wired and free boundary conditions, respectively. Both can be expressed in terms of the Green function $G$ of the full plane $\Z^2$, via the method of images:
\begin{align}
G^{\text{D}}_{(\x_1,\y_1),(\x_2,\y_2)}&=G_{(\x_1,\y_1),(\x_2,\y_2)}-G_{(\x_1,\y_1),(\x_2,-\y_2)},\\
G^{\text{N}}_{(\x_1,\y_1),(\x_2,\y_2)}&=G_{(\x_1,\y_1),(\x_2,\y_2)}+G_{(\x_1,\y_1),(\x_2,1-\y_2)}.
\end{align}
As $\Z^2$ is invariant under translations, the function $G_{u,v}$ only depends on the difference $v-u=(\x,\y)$, and possesses the following integral representation:
\begin{equation}
G_{u,v}=G(v-u)=\int_{\pi}^{\pi}\frac{\diff\theta_1}{2\pi}\int_{\pi}^{\pi}\frac{\diff\theta_2}{2\pi}\frac{e^{\i\x\theta_1+\i\y\theta_2}}{4-2\cos\theta_1-2\cos\theta_2}.
\label{Gf_Z2_int}
\end{equation}
Due to the singularity of the integrand at the origin, the function $G(\x,\y)$ diverges. However, the difference $G(\x,\y)-G(0,0)$ is finite, and can be computed asymptotically for large values of $\x^2+\y^2$ \cite{GPP09}:
\begin{equation}
G(\x,\y)-G(0,0)=-\frac{1}{4\pi}\left\{\log(\x^2+\y^2)+2\gamma+3\log 2\right\}+\frac{\x^4-6\x^2\y^2+\y^4}{24\pi(\x^2+\y^2)^3}+\ldots,
\label{Gf_Z2}
\end{equation}
up to fourth-order terms in $\x,\y$, with $\gamma=0.577216...$ being the Euler constant.

In order to compute Schramm's formula \eqref{Schramm_form2}, we select two boundary vertices $u_i=(\x_1,1)$, $i=1,2$, with $\x_1<\x_2$, and the face $f$ whose lower left corner is $(\x_*,\y_*)$. As we are interested in taking the scaling limit of this formula, we assume that $u_1,u_2,f$ are separated from each other by large distances. Furthermore, we suppose for simplicity that $\x_1<\x_2<\x_*$, and equip $\G$ with the vertical zipper depicted in the left panel of Fig.~\ref{UHP_setup}, namely, we put a nontrivial parallel transport $z\in\C^*$ on the oriented edges $(k,k{+}(1,0))$ for $k=(\x_*,\s)$, $1\le\s\le\y_*$.

Note that imposing that $\x_*>\x_2$ and taking a vertical zipper is simply a matter of convenience, to make the computation of $G'$ as easy as possible. If rather $\x_1<\x_*<\x_2$ or $\x_*<\x_1<\x_2$, one may use a more complicated zipper that touches the boundary to the right of $u_2$, or equivalently, work with the vertical zipper and adapt the combinatorics involved in establishing Eq.~\eqref{Schramm_form2} (see the right panel of Fig.~\ref{UHP_setup}). Both choices are equally valid, and yield the same explicit form for the left-passage probability, which holds for any value of $\x_*$.

\begin{figure}[h]
\centering
\begin{tikzpicture}[scale=0.6]
\tikzstyle arrowstyle=[scale=0.9]
\tikzstyle directed=[postaction={decorate,decoration={markings,mark=at position 0.8 with {\arrow[arrowstyle]{stealth}}}}]

\begin{scope}
\draw[gray,densely dotted,line width=0.7pt] (-0.5,0) grid (10.5,7.5);
\fill[pattern=north east lines, pattern color=blue] (-0.6,-0.25) rectangle (10.6,-1);
\draw[gray,densely dotted,fill=blue!10!white] (7,4) rectangle (8,5);
\node[blue] at (7.5,5.375) {\large $f$};
\filldraw (7,4) circle (0.1cm);
\node at (6,3.75) {\small $(\x_*{,}\y_*)$};
\filldraw[red] (7.5,4.5) circle (0.15cm);
\draw[very thick,red] (7.5,4.5)--(7.5,-0.25);
\foreach \y in {0,...,4}{\draw[thick,dashed,-latex,red] (7,\y)--(8,\y) node[right] {$z$};}
\filldraw (2,0) circle (0.1cm) node[above] {$u_1$};
\filldraw (5,0) circle (0.1cm) node[above] {$u_2$};
\end{scope}

\begin{scope}[xshift=13cm]
\draw[gray,densely dotted,line width=0.7pt] (-0.5,0) grid (10.5,7.5);
\fill[pattern=north east lines, pattern color=blue] (-0.6,-0.25) rectangle (10.6,-1);
\draw[gray,densely dotted,fill=blue!10!white] (3,4) rectangle (4,5);
\node[blue] at (3.5,5.375) {\large $f$};
\filldraw[red] (3.5,4.5) circle (0.15cm);
\draw[very thick,red] (3.5,4.5)--(3.5,-0.25);
\draw[very thick,-latex,red] (3.125,2)--(4,2);
\draw[very thick,densely dashed,red] (3.5,4.5)--(8.5,4.5)--(8.5,-0.25);
\draw[very thick,densely dashed,-latex,red] (8.125,2)--(9,2);
\filldraw (2,0) circle (0.1cm) node[above] {$u_1$};
\filldraw (7,0) circle (0.1cm) node[above] {$u_2$};
\end{scope}

\end{tikzpicture}
\caption{On the left are the relative positions on the upper half-plane of the boundary vertices $u_1,u_2$ and of the face $f$ whose lower left corner is the vertex $(\x_*,\y_*)$. The zipper edges are equipped with a parallel transport $z$ in the direction of their arrow (and $z^{-1}$ in the opposite direction). On the right are two possible choices of zippers for the case $\x_1<\x_*<\x_2$ (the arrow indicates the direction of the zipper edges with parallel transport $z$). Both yield different combinatorial expressions for Schramm's formula in terms of $G,G'$, but lead to the same explicit result.}
\label{UHP_setup}
\end{figure}
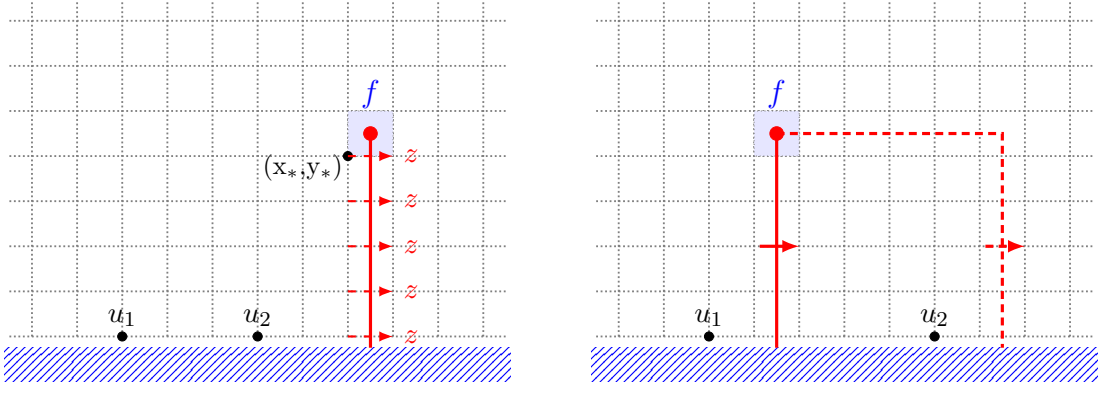

To obtain a continuum version of Schramm's formula, let us introduce new variables $x,y\in\varepsilon\Z^2$ ($\varepsilon>0$), defined by $x=\varepsilon\x$, $y=\varepsilon\y$. We shall take the scaling limit by letting $\varepsilon\to 0^+$ and $\x,\y\to\infty$ such that $x,y$ stay finite. From Laplace's equation
\begin{equation}
\begin{split}
\delta_{(\x_1,\y_1),(\x_2,\y_2)}=(\Delta G)_{(\x_1,\y_1),(\x_2,\y_2)}&=4\,G_{(\x_1,\y_1),(\x_2,\y_2)}-G_{(\x_1+1,\y_1),(\x_2,\y_2)}-G_{(\x_1-1,\y_1),(\x_2,\y_2)}\\
&\quad-G_{(\x_1,\y_1+1),(\x_2,\y_2)}-G_{(\x_1,\y_1-1),(\x_2,\y_2)},
\end{split}
\end{equation}
it follows readily that
\begin{equation}
\begin{split}
\mathfrak{G}(x_1,y_1;x_2,y_2)&=\lim_{\varepsilon\to 0^+}\left(G_{(x_1/\varepsilon,y_1/\varepsilon),(x_2/\varepsilon,y_2/\varepsilon)}-G(0,0)-\frac{1}{2\pi}\log\varepsilon+\frac{2\gamma+3\log 2}{4\pi}\right)\\
&=-\frac{1}{4\pi}\log\left[(x_1-x_2)^2+(y_1-y_2)^2\right].
\end{split}
\end{equation}
Similarly, the Green functions $G^{\text{D}}$ and $G^{\text{N}}$ converge to the usual (continuum) Green functions $\mathfrak{G}^{\text{D}},\mathfrak{G}^{\text{N}}$ with Dirichlet and Neumann boundary conditions, respectively:
\begin{equation}
\mathfrak{G}^{\text{D},\text{N}}(x_1,y_1;x_2,y_2)=-\frac{1}{4\pi}\log\left[(x_1-x_2)^2+(y_1-y_2)^2\right]\pm\frac{1}{4\pi}\log\left[(x_1-x_2)^2+(y_1+y_2)^2\right].
\label{Gf_UHP_cont}
\end{equation}

The same formulas holds if one or both the arguments of the discrete Green function $G^{\text{N}}$ lie on the boundary of the UHP, i.e. the Green function converges to $\mathfrak{G}^{\text{N}}(x_1,y_1;x_2,0)$ or $\mathfrak{G}^{\text{N}}(x_1,0;x_2,0)$, respectively. A bit more care is required for $G^{\text{D}}$, as the continuum Green function $\mathfrak{G}^{\text{D}}$ vanishes on the boundary. One finds at the lowest order in $\varepsilon$ that
\begin{equation}
\begin{split}
G^{\text{D}}_{(x_1/\varepsilon,y_1/\varepsilon),(x_2/\varepsilon,1)}&=\text{P}_1(x_1,y_1;x_2)\,\varepsilon+\mathcal{O}(\varepsilon^2),\\
G^{\text{D}}_{(x_1/\varepsilon,1),(x_2/\varepsilon,1)}&=\text{P}(x_1,x_2)\,\varepsilon^2+\mathcal{O}(\varepsilon^3),
\end{split}
\label{Pk_UHP}
\end{equation}
where $\text{P}_1$ and $\text{P}$ are respectively the \emph{Poisson kernel} and the \emph{excursion Poisson kernel}, defined in terms of the continuum Green function by
\begin{equation}
\begin{split}
\text{P}_1(x_1,y_1;x_2)&=\partial_{y_2}\mathfrak{G}^{\text{D}}(x_1,y_1;x_2,y_2)\big|_{y_2=0},\\
\text{P}(x_1,x_2)&=\partial_{y_1}\partial_{y_2}\mathfrak{G}^{\text{D}}(x_1,y_1;x_2,y_2)\big|_{y_1=y_2=0}.
\end{split}
\end{equation}

Using Eqs. \eqref{Gf_UHP_cont} and \eqref{Pk_UHP}, one may write the Green function and its derivative \eqref{Gp_def} for $\varepsilon\sim 0^+$ as follows:
\begin{align}
G^{\text{D}}_{(x_1/\varepsilon,1),(x_2/\varepsilon,1)}&=\varepsilon^2\,\text{P}(x_1,x_2)+\ldots,\\
G^{\prime\,\text{D}}_{(x_1/\varepsilon,1),(x_2/\varepsilon,1)}&=\varepsilon^2\int_{0}^{y_*}\diff s\,\big(\partial_{x_*}\text{P}_1(x_*,s;x_1)\text{P}_1(x_*,s;x_2)-\partial_{x_*}\text{P}_1(x_*,s;x_2)\text{P}_1(x_*,s;x_1)\big)+\ldots,\\
G^{\text{N}}_{(x_1/\varepsilon,1),(x_2/\varepsilon,1)}&=\left(2\,G(0,0)+\frac{1}{\pi}\log\varepsilon\right)+\mathfrak{G}^{\text{N}}(x_1,0;x_2,0)+\ldots,\\
G^{\prime\,\text{N}}_{(x_1/\varepsilon,1),(x_2/\varepsilon,1)}&=\left(2\,G(0,0)+\frac{1}{\pi}\log\varepsilon\right)\int_{0}^{y_*}\diff s\,\big(\partial_{x_*}\mathfrak{G}^{\text{N}}(x_*,s;x_1,0)-\partial_{x_*}\mathfrak{G}^{\text{N}}(x_*,s;x_2,0)\big)+\ldots
\end{align}
The integrals for both derivatives of the Green function can be carried out explicitly, and yield the following left-passage probabilities in the scaling limit:
\begin{align}
\P^{\text{D}}_{\L}(x_1,x_2;z_*)&=1+\frac{1}{\pi}\big(\arg(z_*{-}x_1)-\arg(z_*{-}x_2)\big)-\frac{1}{\pi}\frac{\Re[(z_*{-}x_1)(\bar{z}_*{-}x_2)]\Im[(z_*{-}x_1)(\bar{z}_*{-}x_2)]}{|z_*{-}x_1|^2\,|z_*{-}x_2|^2},\label{PL_UHP_D}\\
\P^{\text{N}}_{\L}(x_1,x_2;z_*)&=1+\frac{1}{\pi}\big(\arg(z_*{-}x_1)-\arg(z_*{-}x_2)\big),\label{PL_UHP_N}
\end{align}
which are illustrated in Fig.~\ref{PL_UHP_plots}. In particular, if we take $x_1=0$ and $x_2\to\infty$, Eqs.~\eqref{PL_UHP_D} and \eqref{PL_UHP_N} simplify to
\begin{equation}
\P^{\text{D}}_{\L}(z_*)=\frac{\arg z_*}{\pi}-\frac{\Re[z_*]\Im[z_*]}{\pi|z_*|^2},\quad\P^{\text{N}}_{\L}(z_*)=\frac{\arg z_*}{\pi}.
\end{equation}
The former is in agreement with the original $\text{SLE}_{2}$ result of Schramm \cite{Sch01}. For the latter, the right-passage probability $\P_{\R}^{\text{N}}=1-\P_{\L}^{\text{N}}$ is simply the harmonic measure of the segment $[x_1,x_2]$ as seen from $z_*$, as observed in \cite{Ken11}.

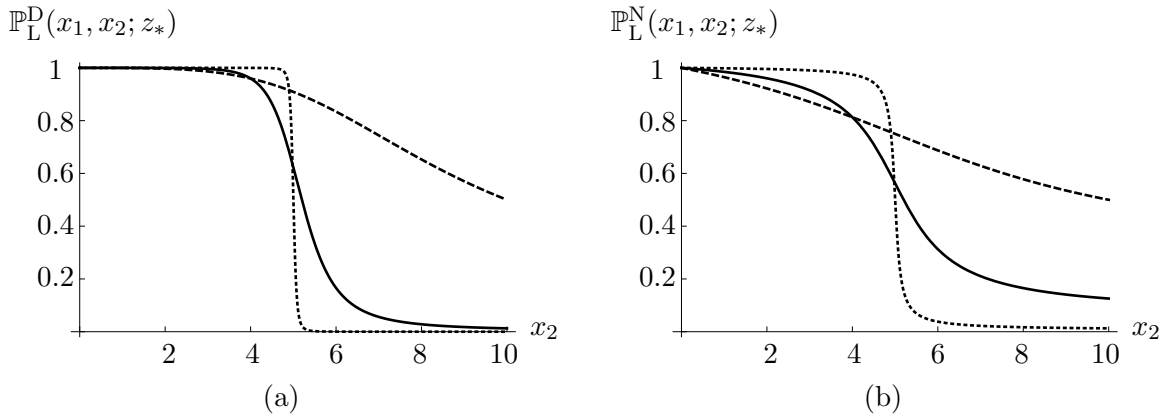
\begin{figure}[h]
\centering
\begin{tikzpicture}

\begin{scope}
\node at (0,0) {\includegraphics[scale=0.6]{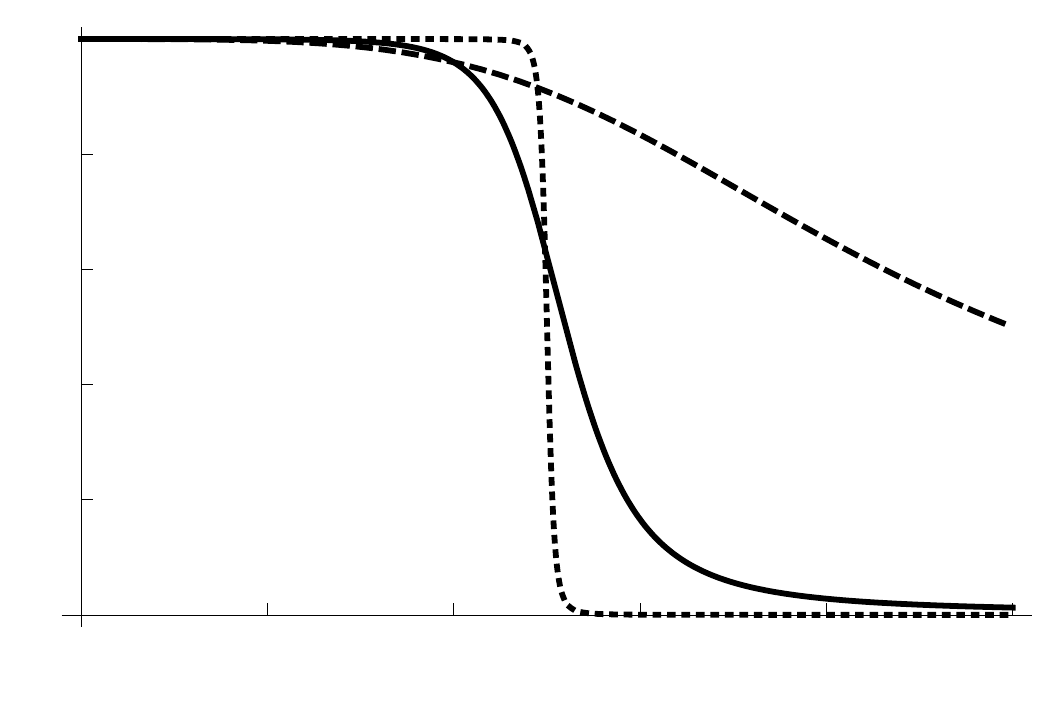}};
\node at (-2.5,2.5) {$\P^{\text{D}}_{\L}(x_1,x_2;z_*)$};
\node at (3.5,-1.6) {$x_2$};
\node at (0,-2.5) {(a)};
\node at (-1.495,-1.9) {2};
\node at (-0.382,-1.9) {4};
\node at (0.731,-1.9) {6};
\node at (1.844,-1.9) {8};
\node at (2.957,-1.9) {10};
\node at (-3,-0.872) {0.2};
\node at (-3,-0.186) {0.4};
\node at (-3,0.5) {0.6};
\node at (-3,1.186) {0.8};
\node at (-3,1.872) {1};
\end{scope}

\begin{scope}[xshift=8cm]
\node at (0,0) {\includegraphics[scale=0.6]{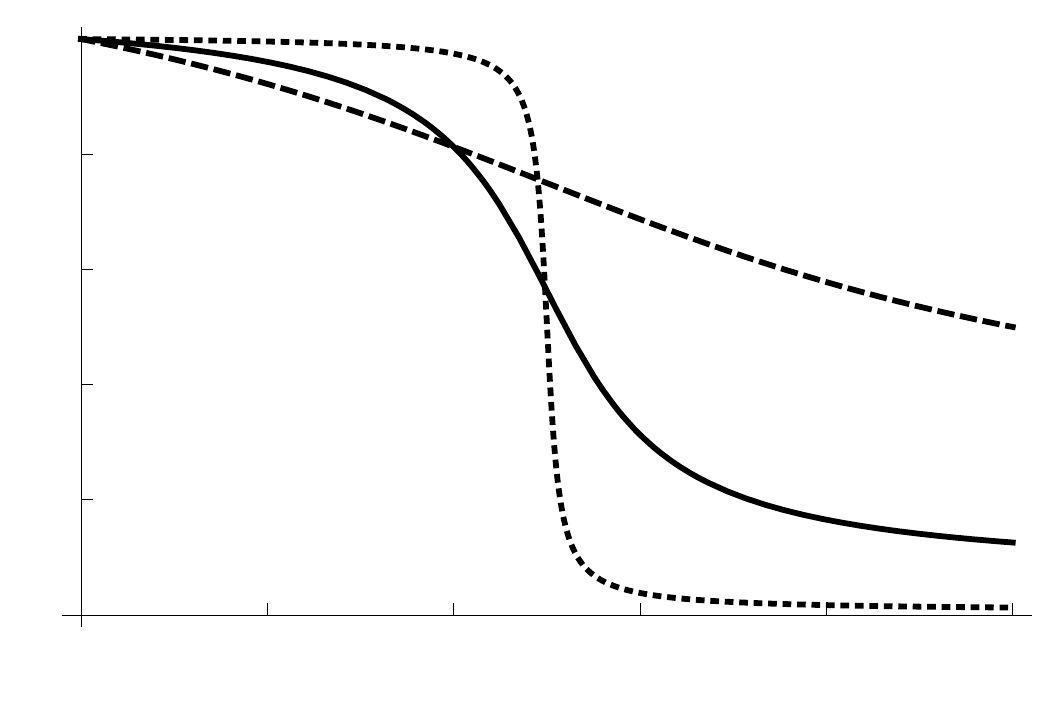}};
\node at (-2.5,2.5) {$\P^{\text{N}}_{\L}(x_1,x_2;z_*)$};
\node at (3.5,-1.6) {$x_2$};
\node at (0,-2.5) {(b)};
\node at (-1.495,-1.9) {2};
\node at (-0.382,-1.9) {4};
\node at (0.731,-1.9) {6};
\node at (1.844,-1.9) {8};
\node at (2.957,-1.9) {10};
\node at (-3,-0.872) {0.2};
\node at (-3,-0.186) {0.4};
\node at (-3,0.5) {0.6};
\node at (-3,1.186) {0.8};
\node at (-3,1.872) {1};
\end{scope}

\end{tikzpicture}
\caption{Illustration of the left-passage probability on the upper half-plane with (a) Dirichlet or (b) Neumann boundary conditions, for $x_1=0$, $x_*=5$, and $y_*=1/10$ (dotted), $y_*=1$ (solid) or $y_*=5$ (dashed).}
\label{PL_UHP_plots}
\end{figure}


\subsection{The cylinder}
\label{sec4.2}

The second application consists in a cylinder graph, namely, a rectangular grid of length $N$ and height $M$ with periodic conditions in the horizontal direction. If we denote by $(\x,\y)$ the coordinates of the vertices of the grid, $1\le\x\le N$, $1\le\y\le M$, then periodicity in the $\x$ direction implies that $(N,\y)$ and $(1,\y)$ are connected by an edge for $1\le\y\le M$.

As in Section~\ref{sec4.1}, we select two vertices $u_i=(\x_i,1)$, $i=1,2$, with $\x_1<\x_2$, and a marked face $f$, whose lower left corner is $(\x_*,\y_*)$, such that $\x_*>\x_2$. In what follows, we shall consider either wired or free conditions on each boundary component, yielding four cases in total. We shall distinguish between them by using the superscripts NN,DD,DN,ND, where the first (resp. second) letter refers to the type of conditions on the bottom (resp. top) boundary, with N standing for Neumann (free) and D for Dirichlet (wired). For each case, one could compute the discrete Green function in terms of the eigenfunctions $|f_{m,n}\rangle$ and eigenvalues $\lambda_{m,n}$ of the Laplacian (which are well known), and then use Eqs.~\eqref{Gp_def} and \eqref{Schramm_form2} to obtain Schramm's formula on the cylinder. The calculation of $G'$ through Eq.~\eqref{Gp_def} is however tedious, and yields long and cumbersome (explicit) expressions, which we have not been able to simplify when at least one boundary is wired (i.e. DD,DN,ND). We shall therefore omit the results and concentrate on the cylinder with free conditions on both boundaries (NN), for which we have obtained a simple formula.

Next we shall discuss a particular case of the left-passage probability, namely, when the marked face $f$ is the one at the top of the cylinder (or at the center of an annulus, as depicted in the right panel of Fig.~\ref{Prb}). Schramm's formula \eqref{Schramm_form2} therefore gives the probability that a random path from $u_1$ to $u_2$ winds counterclockwise around the cylinder; which we shall refer to as the \emph{positive winding probability}.

To compute it, we put a parallel transport $z=e^{\i\theta}$, $\theta\in\mathbb{R}$, on the oriented edges $((N,\y),(1,\y))$ for $1\le\y\le M$. For such a zipper, the eigenfunctions $|\mathbf{f}_{m,n}\rangle=|\mathbf{f}_{m,n}(\theta)\rangle$ and eigenvalues $\boldsymbol{\lambda}_{m,n}=\boldsymbol{\lambda}_{m,n}(\theta)$ of the line bundle Laplacian $\mathbf{\Delta}$ can be evaluated exactly, \emph{for any value of} $\theta$. Since $\mathbf{\Delta}^{\dagger}=\mathbf{\Delta}$, the line bundle Green function is given by the familiar formula
\begin{equation}
\Gr=\Gr(e^{\i\theta})=\sum_{m,n}\frac{|\mathbf{f}_{m,n}\rangle\langle\mathbf{f}_{m,n}|}{\boldsymbol{\lambda}_{m,n}}.
\label{LBL_Gf}
\end{equation}
It is then easier to compute the derivative of the Green function through
\begin{equation}
G'=\partial_z\Gr(z)\big|_{z=1}=-\i\,\partial_{\theta}\mathbf{G}(e^{\i\theta})\big|_{\theta=0}
\end{equation}
than via Eq.~\eqref{Gp_def}. We shall therefore favor the first approach to compute winding probabilities in what follows.


\subsubsection{Free boundaries}
\label{sec4.2.1}
We first consider the cylinder graph with free boundary conditions, on which we compute the left-passage probability for a random simple path from $u_1$ to $u_2$ with respect to a face $f$, using Eqs.~\eqref{Gp_def2} and \eqref{Schramm_form3}. We take here a vertical zipper going from $f$ to the bottom boundary, with parallel transport $z\in\C^*$. An illustration of the two classes of paths is provided in Fig.~\ref{Cyl_graph}.

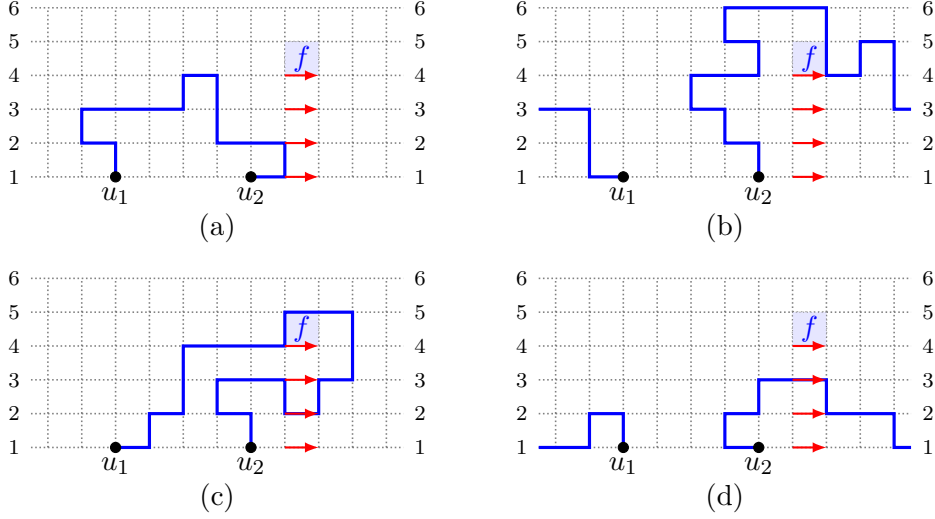
\begin{figure}[h]
\centering
\begin{tikzpicture}[scale=0.45]

\begin{scope}
\draw[gray,densely dotted,line width=0.6pt] (-0.5,0) grid (10.5,5);
\foreach \y in {1,...,6}{\node at (-1,\y-1) {\scriptsize \y};
\node at (11,\y-1) {\scriptsize \y};}
\draw[gray,densely dotted,fill=blue!10!white] (7,3) rectangle node {\textcolor{blue}{$f$}} (8,4);
\draw[very thick,blue] (2,0)--(2,1)--(1,1)--(1,2)--(3,2)--(4,2)--(4,3)--(5,3)--(5,1)--(7,1)--(7,0)--(6,0);
\foreach \y in {0,...,3}{\draw[thick,-latex,red] (7,\y)--(8,\y);}
\filldraw (2,0) circle (0.15cm) node[below] {$u_1$};
\filldraw (6,0) circle (0.15cm) node[below] {$u_2$};
\node at (5,-1.5) {(a)};
\end{scope}

\begin{scope}[shift={(15,0)}]
\draw[gray,densely dotted,line width=0.6pt] (-0.5,0) grid (10.5,5);
\foreach \y in {1,...,6}{\node at (-1,\y-1) {\scriptsize \y};
\node at (11,\y-1) {\scriptsize \y};}
\draw[gray,densely dotted,fill=blue!10!white] (7,3) rectangle node {\textcolor{blue}{$f$}} (8,4);
\draw[very thick,blue] (2,0)--(1,0)--(1,2)--(-0.5,2);
\draw[very thick,blue] (10.5,2)--(10,2)--(10,4)--(9,4)--(9,3)--(8,3)--(8,5)--(5,5)--(5,4)--(6,4)--(6,3)--(4,3)--(4,2)--(5,2)--(5,1)--(6,1)--(6,0);
\foreach \y in {0,...,3}{\draw[thick,-latex,red] (7,\y)--(8,\y);}
\filldraw (2,0) circle (0.15cm) node[below] {$u_1$};
\filldraw (6,0) circle (0.15cm) node[below] {$u_2$};
\node at (5,-1.5) {(b)};
\end{scope}

\begin{scope}[shift={(0,-8)}]
\draw[gray,densely dotted,line width=0.6pt] (-0.5,0) grid (10.5,5);
\foreach \y in {1,...,6}{\node at (-1,\y-1) {\scriptsize \y};
\node at (11,\y-1) {\scriptsize \y};}
\draw[gray,densely dotted,fill=blue!10!white] (7,3) rectangle node {\textcolor{blue}{$f$}} (8,4);
\draw[very thick,blue] (2,0)--(3,0)--(3,1)--(4,1)--(4,3)--(7,3)--(7,4)--(9,4)--(9,2)--(8,2)--(8,1)--(7,1)--(7,2)--(5,2)--(5,1)--(6,1)--(6,0);
\foreach \y in {0,...,3}{\draw[thick,-latex,red] (7,\y)--(8,\y);}
\filldraw (2,0) circle (0.15cm) node[below] {$u_1$};
\filldraw (6,0) circle (0.15cm) node[below] {$u_2$};
\node at (5,-1.5) {(c)};
\end{scope}

\begin{scope}[shift={(15,-8)}]
\draw[gray,densely dotted,line width=0.6pt] (-0.5,0) grid (10.5,5);
\foreach \y in {1,...,6}{\node at (-1,\y-1) {\scriptsize \y};
\node at (11,\y-1) {\scriptsize \y};}
\draw[gray,densely dotted,fill=blue!10!white] (7,3) rectangle node {\textcolor{blue}{$f$}} (8,4);
\draw[very thick,blue] (2,0)--(2,1)--(1,1)--(1,0)--(-0.5,0);
\draw[very thick,blue] (10.5,0)--(10,0)--(10,1)--(8,1)--(8,2)--(6,2)--(6,1)--(5,1)--(5,0)--(6,0);
\foreach \y in {0,...,3}{\draw[thick,-latex,red] (7,\y)--(8,\y);}
\filldraw (2,0) circle (0.15cm) node[below] {$u_1$};
\filldraw (6,0) circle (0.15cm) node[below] {$u_2$};
\node at (5,-1.5) {(d)};
\end{scope}

\end{tikzpicture}
\caption{Simple paths from $u_1$ to $u_2$ on a cylinder graph $\G$ realized as a rectangular grid with its left and right sides connected (the numbers on the half-edges indicate how they should be glued together to obtain $\G$). The edges with red arrows between the marked face $f$ and the lower boundary are equipped with a parallel transport $z$ (resp. $z^{-1}$) in the direction (resp. opposite direction) of the arrows. The paths are divided into two classes, according to whether they leave $f$ to their left as in panels (a) and (b), or to their right as in panels (c) and (d).}
\label{Cyl_graph}
\end{figure}

The eigenvalues and eigenfunctions of the standard Laplacian $\Delta_0$ read
\begin{align}
\lambda_{m,n}&=4-2\cos\left(\frac{\pi m}{M+1}\right)-2\cos\left(\frac{2\pi n}{N}\right),\\
f_{m,n}(\x,\y)&=\left(\frac{2-\delta_{m,0}}{MN}\right)^{1/2}e^{\i(2\pi n)\x/N}\cos\left[\frac{\pi m}{M}\left(\y-\frac{1}{2}\right)\right],
\label{Cyl_eig_NN}
\end{align}
for $0\le m\le M{-}1$ and $0\le n\le N{-}1$. In particular, $\lambda_{0,0}=0$, meaning $\Delta_0$ is singular. The regularized Green function, introduced in Section~\ref{sec2.3}, takes the form of \eqref{Gf_reg}.

As on the upper half-plane, we define new coordinates $x=\varepsilon\x$ and $y=\varepsilon\y$ living on the lattice $\varepsilon\Z^2$. We take the scaling limit $\varepsilon\to 0^+$, $M,N\to\infty$ such that $M\varepsilon\to p$ and $N\varepsilon\to 2\pi$, and compute the Green function $\mathfrak{G}$ and its derivative $\widetilde{\mathfrak{G}}'$ directly in the continuum. They read respectively
\begin{align}
\mathfrak{G}(x_1,y_1;x_2,y_2)&=\frac{1}{2\pi p}\sum_{\substack{m,n\in\Z\\(m,n)\neq(0,0)}}\frac{e^{\i n(x_1-x_2)}\cos\left(\frac{\pi m y_1}{p}\right)\cos\left(\frac{\pi m y_2}{p}\right)}{n^2+\left(\frac{\pi m}{p}\right)^2},\\
\widetilde{\mathfrak{G}}'(x_1,y_1;x_2,y_2)&=\int_{0}^{y_*}\diff s\big(\partial_{x_*}\mathfrak{G}(x_*,s;x_1,y_1)-\partial_{x_*}\mathfrak{G}(x_*,s;x_2,y_2)\big).
\end{align}
For boundary points $u_i=(x_i,0)$, we use the Poisson summation formula to write the derivative with respect  to $x_*$ as
\begin{equation}
\partial_{x_*}\mathfrak{G}(x_*,s;x_i,0)=-\frac{1}{\pi}\sum_{k\in\Z}\sum_{n=1}^{\infty}\sin[n(x_*-x_i)]\,e^{-n|2kp-s|}.
\end{equation}
The integration over $s$ is straightforward, and yields
\begin{equation}
\int_{0}^{y_*}\diff s\,\partial_{x_*}\mathfrak{G}(x_*,s;x_i,0)=-\frac{1}{2}+\frac{x_*-x_i}{2\pi}+\frac{1}{\pi}\sum_{n=1}^{\infty}\frac{\sin[n(x_*-x_i)]}{n}\frac{e^{n(p-y_*)}-e^{-n(p-y_*)}}{e^{np}-e^{-np}}.
\end{equation}
We are now in position to write explicitly Schramm's formula \eqref{Schramm_form3} for a path on the cylinder with Neumann boundary conditions, in the scaling limit:
\begin{equation}
\P^{\text{NN}}_{\L}(x_1,x_2,p;z_*)=1-\frac{x_2{-}x_1}{2\pi}-\frac{1}{\pi}\sum_{n=1}^{\infty}\frac{\sin[n(x_*{-}x_1)]-\sin[n(x_*{-}x_2)]}{n}\frac{e^{n(p-y_*)}-e^{-n(p-y_*)}}{e^{np}-e^{-np}},
\label{PL_cyl_NN}
\end{equation}
where $z_*=x_*+\i y_*$. We illustrate this probability in Fig.~\ref{P+_NN_cyl_plots}. Taking the limit $y_*\to p$ yields the positive winding probability, simply given by
\begin{equation}
\lim_{y_*\to p}\P^{\text{NN}}_{\L}(x_1,x_2,p;z_*)=\P^{\text{NN}}_{+}(x_2-x_1)=1-\frac{x_2-x_1}{2\pi}.
\label{PNN}
\end{equation}
Surprisingly, this winding probability does \emph{not} depend on the height $p$ of the cylinder. The corresponding probabilities for other combinations of boundary conditions, on the other hand, \emph{do} depend on $p$, as we shall see in Sections~\ref{sec4.2.2} and \ref{sec4.2.3}. We have not found an explanation for this property so far.

\begin{figure}[h]
\centering
\begin{tikzpicture}

\begin{scope}
\node at (0,0) {\includegraphics[scale=0.7]{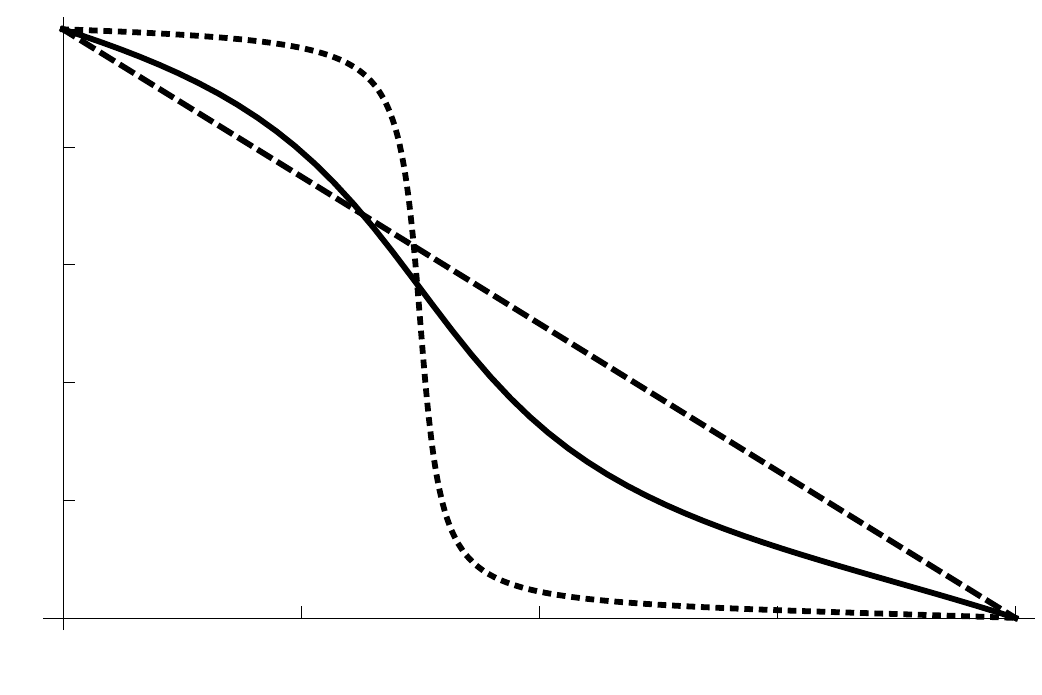}};
\node at (-2.7,3) {$\P^{\text{NN}}_{\L}(x_1,x_2,p;z_*)$};
\node at (3.9,-1.8) {$x_2$};
\node at (-1.4,-2.2) {$\pi/2$};
\node at (0.2,-2.2) {$\pi$};
\node at (1.8,-2.2) {$3\pi/2$};
\node at (3.4,-2.2) {$2\pi$};
\node at (-3.6,-1) {$0.2$};
\node at (-3.6,-0.2) {$0.4$};
\node at (-3.6,0.6) {$0.6$};
\node at (-3.6,1.4) {$0.8$};
\node at (-3.6,2.2) {$1$};
\end{scope}

\begin{scope}[xshift=8cm]
\node at (0,0) {\includegraphics[scale=0.7]{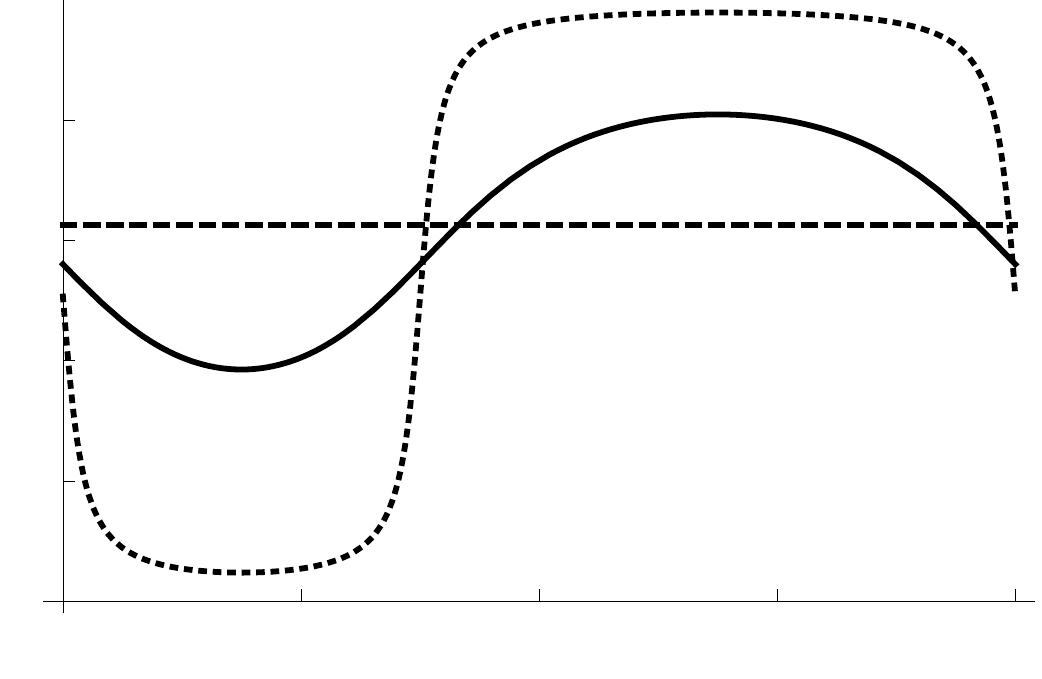}};
\node at (-2.7,3) {$\P^{\text{NN}}_{\L}(x_1,x_2,p;z_*)$};
\node at (4,-1.8) {$x_*$};
\node at (-1.4,-2.2) {$\pi/2$};
\node at (0.2,-2.2) {$\pi$};
\node at (1.8,-2.2) {$3\pi/2$};
\node at (3.4,-2.2) {$2\pi$};
\node at (-3.6,-1) {$0.2$};
\node at (-3.6,-0.2) {$0.4$};
\node at (-3.6,0.6) {$0.6$};
\node at (-3.6,1.4) {$0.8$};
\node at (-3.6,2.2) {$1$};
\end{scope}

\end{tikzpicture}
\caption{Illustration of the left-passage probability on a cylinder of perimeter $2\pi$ and height $p=10$ with Neumann conditions on both boundaries. We assume the random paths to start at $x_1=0$, and consider three distinct values for the vertical coordinate of the marked point $z_*$: $y_*=1/10$ (dotted), $y_*=1$ (solid) and $y_*=10$ (dashed). On the left is the left-passage probability as a function of the position of the endpoint $x_2$ of the paths, for $x_*=3\pi/4$. On the right is the same probability as a function of the horizontal coordinate $x_*$ of $z_*$, for $x_2=3\pi/4$.}
\label{P+_NN_cyl_plots}
\end{figure}
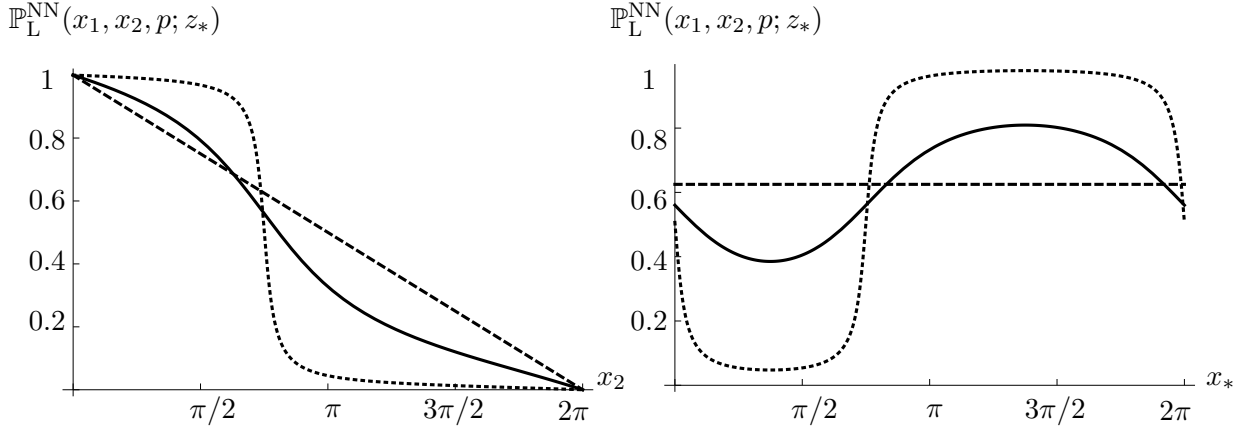


\subsubsection{Wired boundaries}
\label{sec4.2.2}
As a second calculation on the cylinder, let us now determine the positive winding probability (i.e. the left-passage probability with $\y_*=p$) when both the top and bottom boundaries of the cylinder are wired to the root $s$. We proceed here differently from in Section~\ref{sec4.2.1}, in the sense that we compute the line bundle Green function $\Gr=\Gr(e^{\i\theta})$ directly, for any value of $\theta$ (a feat we have not been able to reproduce for a zipper connected to a generic face $f$ of the graph). We choose here a zipper with parameter $z=e^{\i\theta}$ crossing the edges of the form $((N,\y),(1,\y))$ for $1\le\y\le M$.

The eigenvalue equation for the line bundle Laplacian $\mathbf{\Delta}$ reads
\begin{equation}
\mathbf{\Delta}\mathbf{f}(\x,\y)=4\mathbf{f}(\x,\y)-\mathbf{f}(\x{-}1,\y)-\mathbf{f}(\x{+}1,\y)-\mathbf{f}(\x,\y{-}1)-\mathbf{f}(\x,\y{+}1)=\boldsymbol{\lambda}\,\mathbf{f}(\x,\y)
\end{equation}
for $1\le\x\le N$, $1\le\y\le M$. The eigenfunctions must further satisfy the following boundary conditions:
\begin{equation}
\mathbf{f}(\x,0)=0=\mathbf{f}(\x,M+1),\quad\mathbf{f}(0,\y)=e^{\i\theta}\mathbf{f}(N,\y),\quad\mathbf{f}(N+1,\y)=e^{-\i\theta}\mathbf{f}(1,\y).
\end{equation}
One readily finds that the eigenvalues and normalized eigenfunctions are given by
\begin{align}
\boldsymbol{\lambda}_{m,n}&=4-2\cos\left(\frac{\pi m}{M+1}\right)-2\cos\left(\frac{2\pi n-\theta}{N}\right),\\
\mathbf{f}_{m,n}(\x,\y)&=\left(\frac{2}{(M+1)N}\right)^{1/2}e^{\i(2\pi n-\theta)\x/N}\sin\left(\frac{\pi m\y}{M+1}\right),
\end{align}
for $1\le m\le M$ and $0\le n\le N{-}1$. In the same scaling limit as in Section~\ref{sec4.2.1}, the line bundle Green function $\Gr$ converges to the continuum Green function $\mathfrak{G}$ of a cylinder of perimeter $2\pi$ and height $p$ with Dirichlet conditions on both boundaries,
\begin{equation}
\mathfrak{G}(x_1,y_1;x_2,y_2;\theta)=\frac{1}{2\pi p}\sum_{m\in\Z^*}\sum_{n\in\Z}\frac{e^{\i\left(n-\frac{\theta}{2\pi}\right)(x_1-x_2)}\sin\left(\frac{\pi m y_1}{p}\right)\sin\left(\frac{\pi m y_2}{p}\right)}{\left(n-\frac{\theta}{2\pi}\right)^2+\left(\frac{\pi m}{p}\right)^2}.
\end{equation}
For boundary points $u_i=(x_i,0)$, $i=1,2$, the scaling limit of Schramm's formula may be written in terms of the excursion Poisson kernel $\text{P}(x,p;\theta)$ (see Section~\ref{sec4.1} on the UHP), as 
\begin{equation}
\P^{\text{DD}}_{+}(x,p)=1+\i\frac{1}{\text{P}(x,p;\theta)}\frac{\diff\text{P}(x,p;\theta)}{\diff\theta}\bigg|_{\theta=0},
\end{equation}
where $x=x_2-x_1$, the subscript $+$ refers to the positive (counterclockwise) winding, and the superscripts stand for the type of boundary conditions on the bottom and top of the cylinder, respectively.

Recall that the excursion Poisson kernel is given by the normal derivative of the Green function $\mathfrak{G}$ at the boundary points $(x_1,0)$ and $(x_2,0)$. It reads here
\begin{equation}
\text{P}(x,p;\theta)\equiv\lim_{\delta_1,\delta_2\to 0^+}\frac{1}{\delta_1\delta_2}\mathfrak{G}(x_1,\delta_1;x_2,\delta_2;\theta)=\frac{\pi}{4p^2}\sum_{k\in\Z}\frac{e^{\i k\theta}}{\sinh^2\left[(2\pi k-x)\frac{\pi}{2p}\right]},
\end{equation}
where we used the Poisson summation formula to write the last equality. Its zeroth and first orders in $\theta$ can be written in terms of Jacobi's theta functions (which we recall in Appendix~\ref{a1}) and yield the following winding probability:
\begin{equation}
\P^{\text{DD}}_{+}(x,p)=1-\frac{1}{2\pi}\frac{1}{\partial_x F_1(x,p)}\big(x\,\partial_x+p\,\partial_p\big)F_1(x,p),
\label{P+_cyl_DD}
\end{equation}
in agreement with the $\text{SLE}_2$ computation of \cite{Hag09a,Hag09b}. Here $F_1(x,p)$ is given by
\begin{equation}
F_1(x,p)=x+p\frac{\vartheta'_1(x/2,e^{-p})}{\vartheta_1(x/2,e^{-p})}=x+p\cot\left(\frac{x}{2}\right)+4p\sum_{n=1}^{\infty}\frac{\sin(nx)}{e^{2np}-1}.
\end{equation}
We refer to Fig.~\ref{P+_cyl} for an illustration of this winding probability for different moduli $p$.

In the limit of long cylinders $p\to\infty$, one recovers the equivalent of Schramm's formula for $\kappa=2$ \cite{Sch01} in our geometry,
\begin{equation}
\P^{\text{DD}}_{+}(x,p)\simeq 1-\frac{x-\sin x}{2\pi}+\frac{1}{\pi p}\sin x\sin^2(x/2).
\label{PDD_large_p}
\end{equation}
For very thin cylinders, i.e. for $p\sim 0^+$, one obtains, using the modular properties of theta functions (recalled in Appendix~\ref{a1}),
\begin{equation}
\P^{\text{DD}}_{+}(x,p)\simeq\frac{e^{-\pi(2\pi-x)/p}}{e^{-\pi x/p}+e^{-\pi(2\pi-x)/p}}\simeq\begin{cases}
1-e^{-2\pi(\pi-x)/p}\quad&\text{if $0<x<\pi$,}\\
\frac{1}{2}\quad&\text{if $x=\pi$,}\\
e^{-2\pi(x-\pi)/p}\quad&\text{if $\pi<x<2\pi$.}
\end{cases}
\label{PDD_small_p}
\end{equation}
which converges to the Heaviside step function. An intuitive explanation for this fact was given in \cite{Hag09a}, using the correspondence between paths in spanning forests and loop-erased random walks \cite{Pem91}: for very thin cylinders, the loop erasure of walks may be neglected at leading order, and the Heaviside function can be obtained by a simple Brownian motion calculation in the scaling limit.


\subsubsection{Mixed boundaries}
\label{sec4.2.3}
The calculation of the positive winding probability on a cylinder with wired and free conditions on its bottom and top boundaries, respectively, is very similar to that of $\P^{\text{DD}}_{+}$. Therefore, we merely state the result in the scaling limit \cite{Hag09b}:
\begin{equation}
\P^{\text{DN}}_{+}(x,p)=1-\frac{1}{2\pi}\frac{1}{\partial_x F_2(x,p)}\big(x\,\partial_x+p\,\partial_p\big)F_2(x,p),
\label{P+_cyl_DN}
\end{equation}
with the function $F_2(x,p)$ defined in terms of Jacobi's theta and elliptic functions by
\begin{equation}
F_2(x,p)=p\,\vartheta_3^2\,\text{cs}\big(\vartheta_3^2\,x/2,\vartheta_2^2/\vartheta_3^2\big)=p\cot\left(\frac{x}{2}\right)-4p\sum_{n=1}^{\infty}\frac{\sin(nx)}{e^{2np}+1},
\end{equation}
where $\vartheta_a\equiv\vartheta_a(0,e^{-p})$ for $a=2,3$. For large cylinders, we find
\begin{equation}
\P^{\text{DN}}_{+}(x,p)\simeq 1-\frac{x-\sin x}{2\pi}+\frac{8}{\pi}\sin x\sin^2(x/2)\,p\,e^{-2p},
\end{equation}
which yields the same result as Eq.~\eqref{PDD_large_p} in the limit $p\to\infty$. Interestingly, the first correction is exponential here, as opposed to the $1/p$ correction computed above for the pure Dirichlet case. When $p\to 0^+$, the winding probability converges to the Heaviside function as follows:
\begin{equation}
\P^{\text{DN}}_{+}(x,p)\simeq 1-\left(1+e^{\pi(\pi-x)/p}\right)^{-1}\simeq\begin{cases}
1-e^{-\pi(\pi-x)/p}\quad&\text{if $0<x<\pi$,}\\
\frac{1}{2}\quad&\text{if $x=\pi$,}\\
e^{-\pi(x-\pi)/p}\quad&\text{if $\pi<x<2\pi$.}
\end{cases}
\end{equation}
Both asymptotic limits of $\P^{\text{DN}}_{+}$ are represented in Fig.~\ref{P+_cyl}.

For the opposite choice of boundary conditions, namely free and wired on the bottom and top boundaries, respectively, we obtain the following Green function in the scaling limit,
\begin{equation}
\mathfrak{G}(x_1,0;x_2,0;\theta)=\frac{1}{2\pi}\sum_{n\in\Z}\frac{e^{\i\left(n-\frac{\theta}{2\pi}\right)(x_1-x_2)}}{n-\frac{\theta}{2\pi}}\tanh\left[\left(n-\frac{\theta}{2\pi}\right)p\right].
\end{equation}
Its leading order for $\theta\sim 0^+$ yields, with $x=x_2-x_1$,
\begin{equation}
\begin{split}
\mathfrak{G}(x,p)&\equiv\mathfrak{G}(x_1,0;x_2,0;\theta)\big|_{\theta=0}=\frac{p}{2\pi}+\frac{1}{\pi}\sum_{n=1}^{\infty}\frac{\cos(nx)\tanh(np)}{n}\\
&=\frac{1}{\pi}\log\left(\frac{2\,\vartheta_4^2(x/2,e^{-2p})}{\vartheta_1(x/2,e^{-p})\vartheta_2(0,e^{-p})}\right).
\end{split}
\end{equation}
The derivative of the Green function can, upon visual inspection, be rewritten as
\begin{equation}
\mathfrak{G}'(x,p)\equiv-\i\,\partial_{\theta}\mathfrak{G}(x_1,0;x_2;0;\theta)\big|_{\theta=0}=\frac{x}{2\pi}\mathfrak{G}(x,p)+\frac{1}{2\pi}(p\,\partial_p-1)\int_{0}^{x}\diff t\,\mathfrak{G}(t,p).
\end{equation}
It follows that Schramm's formula for the Neumann-Dirichlet case reads, in the scaling limit,
\begin{equation}
\P^{\text{ND}}_{+}(x,p)=1-\frac{1}{2\pi}\frac{1}{\partial_x F_3(x,p)}\big(x\,\partial_x+p\,\partial_p-1\big)F_3(x,p),
\label{P+_cyl_ND}
\end{equation}
where $F_3(x,p)$ is defined by
\begin{equation}
F_3(x,p)=\int_{0}^{x}\diff t\,\mathfrak{G}(t,p).
\end{equation}
We may compute the large- and small-$p$ expansions as above, using the approximations
\begin{align}
\mathfrak{G}(x,p)&\simeq\frac{p}{2\pi}-\frac{1}{\pi}\log[2\sin(x/2)]-\frac{2}{\pi}\cos x\,e^{-2p}&\quad\text{for $p\gg 1$,}\\
\mathfrak{G}(x,p)&\simeq\frac{2}{\pi}e^{-\pi x/(2p)}+\frac{2}{\pi}e^{-\pi(2\pi-x)/(2p)}&\quad\text{for $p\ll 1$.}
\end{align}
In the limits $p\to\infty$ and $p\to 0^+$, the winding probability converges to $1-x/(2\pi)$ and to the Heaviside function $1-\Theta(x-\pi)$, respectively (see Fig.~\ref{P+_cyl}).

\begin{figure}[h]
\centering
\begin{tikzpicture}

\begin{scope}
\node at (0,0) {\includegraphics[scale=0.6]{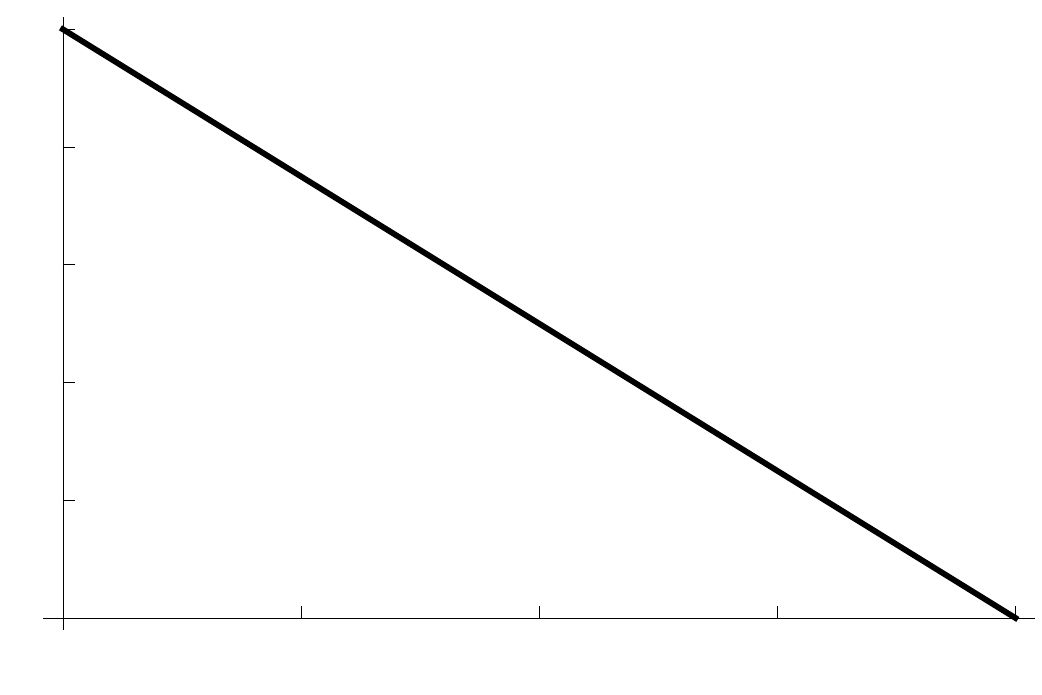}};
\node at (-2.6,2.6) {$\P^{\text{DD}}_{+}(x,p)$};
\node at (3.4,-1.6) {$x$};
\node at (-1.3,-2) {$\pi/2$};
\node at (0.2,-2) {$\pi$};
\node at (1.6,-2) {$3\pi/2$};
\node at (3,-2) {$2\pi$};
\node at (-3.2,-0.85) {$0.2$};
\node at (-3.2,-0.15) {$0.4$};
\node at (-3.2,0.55) {$0.6$};
\node at (-3.2,1.25) {$0.8$};
\node at (-3.2,1.95) {$1$};
\end{scope}

\begin{scope}[xshift=8cm]
\node at (0,0) {\includegraphics[scale=0.6]{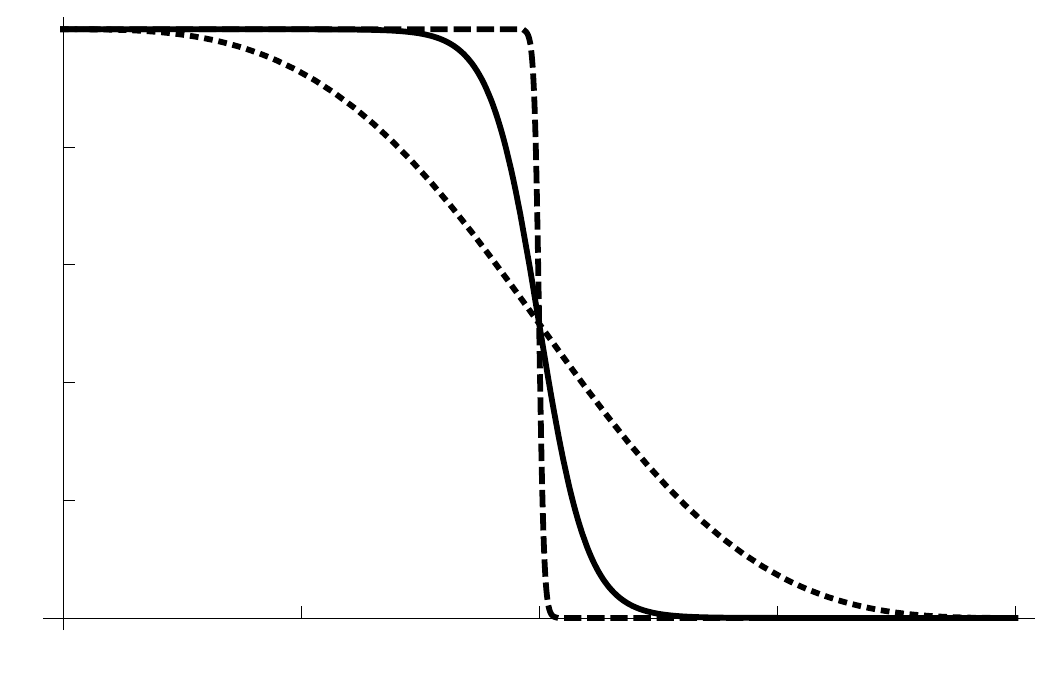}};
\node at (-2.6,2.6) {$\P^{\text{DD}}_{+}(x,p)$};
\node at (3.4,-1.6) {$x$};
\node at (-1.3,-2) {$\pi/2$};
\node at (0.2,-2) {$\pi$};
\node at (1.6,-2) {$3\pi/2$};
\node at (3,-2) {$2\pi$};
\node at (-3.2,-0.85) {$0.2$};
\node at (-3.2,-0.15) {$0.4$};
\node at (-3.2,0.55) {$0.6$};
\node at (-3.2,1.25) {$0.8$};
\node at (-3.2,1.95) {$1$};
\end{scope}

\begin{scope}[yshift=-6cm]
\node at (0,0) {\includegraphics[scale=0.6]{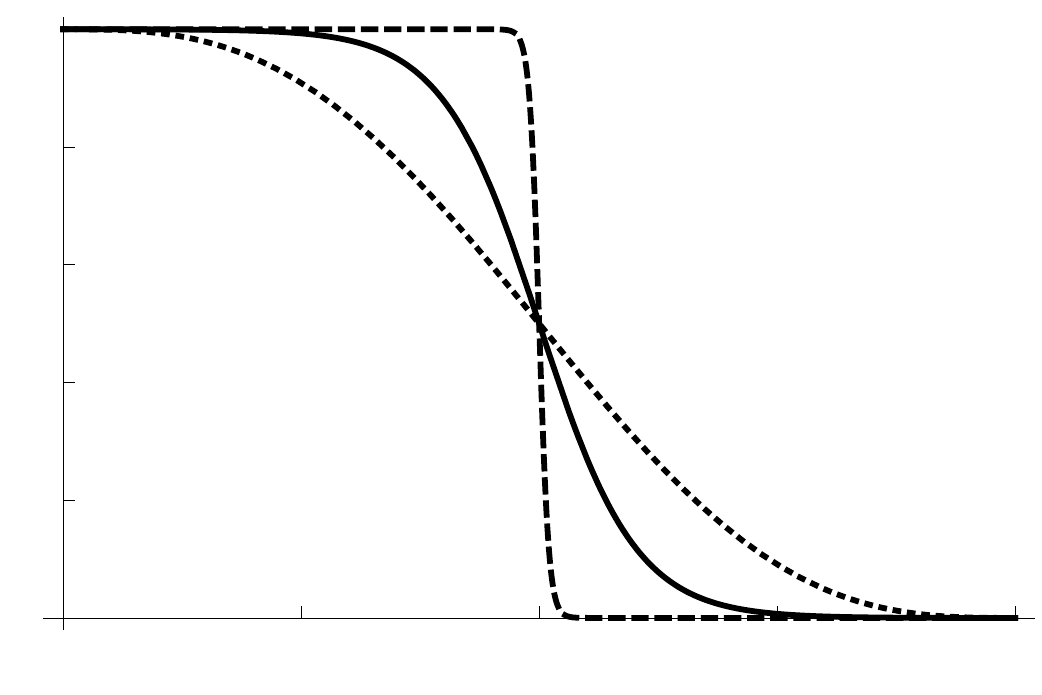}};
\node at (-2.6,2.6) {$\P^{\text{DN}}_{+}(x,p)$};
\node at (3.4,-1.6) {$x$};
\node at (-1.3,-2) {$\pi/2$};
\node at (0.2,-2) {$\pi$};
\node at (1.6,-2) {$3\pi/2$};
\node at (3,-2) {$2\pi$};
\node at (-3.2,-0.85) {$0.2$};
\node at (-3.2,-0.15) {$0.4$};
\node at (-3.2,0.55) {$0.6$};
\node at (-3.2,1.25) {$0.8$};
\node at (-3.2,1.95) {$1$};
\end{scope}

\begin{scope}[xshift=8cm,yshift=-6cm]
\node at (0,0) {\includegraphics[scale=0.6]{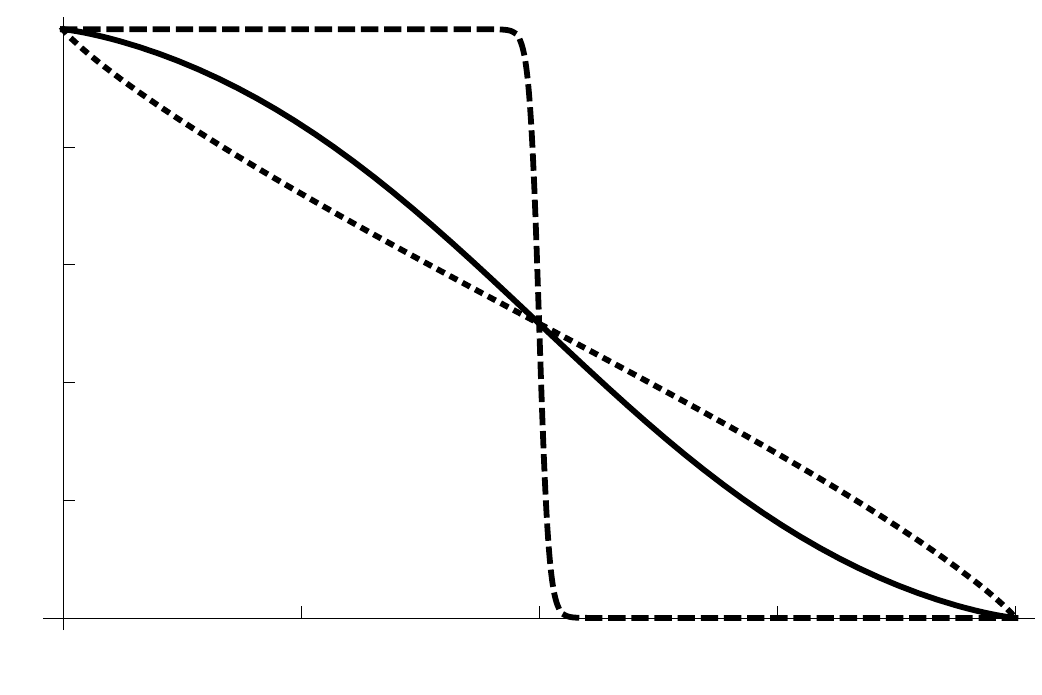}};
\node at (-2.6,2.6) {$\P^{\text{ND}}_{+}(x,p)$};
\node at (3.4,-1.6) {$x$};
\node at (-1.3,-2) {$\pi/2$};
\node at (0.2,-2) {$\pi$};
\node at (1.6,-2) {$3\pi/2$};
\node at (3,-2) {$2\pi$};
\node at (-3.2,-0.85) {$0.2$};
\node at (-3.2,-0.15) {$0.4$};
\node at (-3.2,0.55) {$0.6$};
\node at (-3.2,1.25) {$0.8$};
\node at (-3.2,1.95) {$1$};
\end{scope}

\end{tikzpicture}
\caption{Illustration of the four positive winding probability on the cylinder with Dirichlet and/or Neumann conditions on its lower and upper boundaries, for $p=1/10$ (dashed), $p=1$ (solid) and $p=10$ (dotted).}
\label{P+_cyl}
\end{figure}
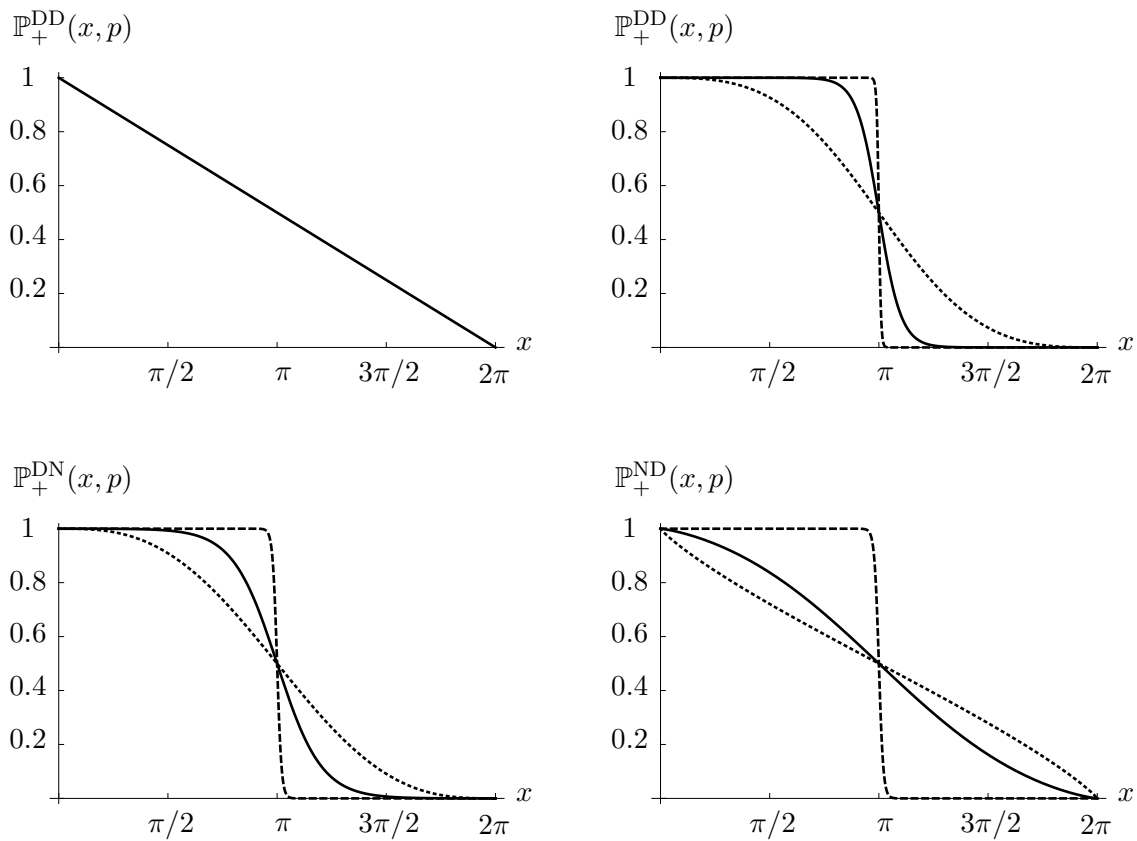


\subsection{The Möbius strip}
\label{sec4.3}

Our next example for Schramm's formula consists in a grid embedded on a Möbius strip. The nonorientability of the surface brings out two complications. First, the concepts of ``left-passage'' and ``counterclockwise'' are ill defined on such a surface. Second, simple paths between two boundary vertices on a Möbius strip are distributed into not \emph{two} but \emph{three} distinct topological classes, as illustrated in Fig.~\ref{Moebius_paths}. An adaptation of Eq.~\eqref{Schramm_form2} is therefore needed to compute their respective winding probabilities.

The graph $\G$ we consider here is an $N\times M$ rectangle, with vertex coordinates $(\x,\y)$ for $1\le\x\le N$ and $1\le\y\le M$. We choose wired or free conditions on both the top ($\y=M$) and bottom ($\y=1$) sides of the rectangle, and twisted periodic conditions in the horizontal direction, namely, $(\x{+}N,\y)\sim(\x,M{+}1{-}\y)$. The latter are implemented by connecting each vertex $(N,\y)$ of the rectangle to $(1,M{+}1{-}\y)$ by an edge, for $1\le\y\le M$. Furthermore, we equip these edges with a parallel transport $\phi_{(N+1,\y),(1,M+1-\y)}=e^{\i\theta}$, with $\theta\in\mathbb{R}$.

With respect to the zipper, there are four possible ways to select two boundary vertices $u_i=(\x_i,\y_i)$, with $1\le\x_i\le N$ and $i=1,2$, as each $\y_i$ may take the value $1$ or $M$. It easy to see that choosing $\y_1=\y_2=1$ is equivalent to $\y_1=\y_2=M$, and that $\y_1=1,\y_2=M$ is equivalent to $\y_1=M,\y_2=1$. Let us therefore pick $u_1=(\x_1,1)$, and discuss the two cases $u_2=(\x_2,1)$ (with $\x_1<\x_2$) and $u_2=(\x_2,M)\simeq(\x_2+N,1)$ with $1\le\x_2\le N$. We start with the former, and consider oriented cycle-rooted groves on the graph that contain a path from $u_1$ to $u_2$. In such groves, the path from $u_1$ to $u_2$ winds around the strip zero, one or two times, yielding a product of parallel transports $\phi_{1\to 2}=1,e^{-\i\theta},e^{-2\i\theta}$, respectively (see Fig.~\ref{Moebius_paths}). The noncontractible cycles $\alpha$ appearing in such groves are of two types, winding once or twice around the strip, and therefore pick up a monodromy factor $\varpi_{\alpha}=e^{\i\theta},e^{2\i\theta}$ or their inverses (depending on the orientation of the cycles). The contractible cycles, on the other hand, have a trivial monodromy, and so do not appear in OCRGs counted by the partition function $\Zr[{\textstyle{2\atop 1}}]$.

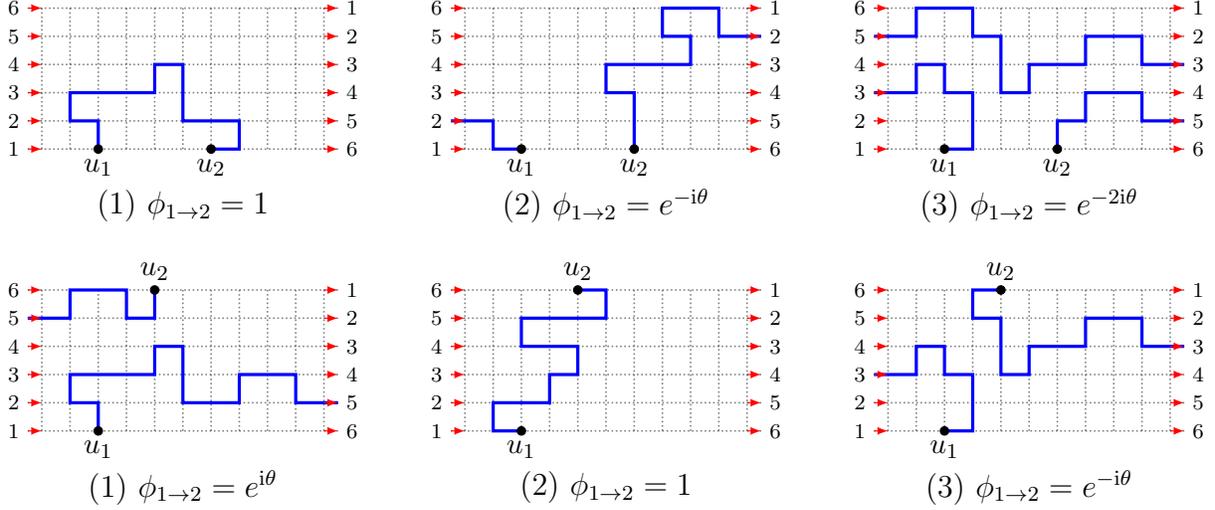
\begin{figure}[h]
\centering
\begin{tikzpicture}[scale=0.375]

\begin{scope}
\draw[gray,densely dotted,line width=0.6pt] (-0.5,0) grid (10.5,5);
\foreach \y in {1,...,6}{\node at (-1,\y-1) {\scriptsize \y};
\node at (11,6-\y) {\scriptsize \y};}
\draw[very thick,blue] (2,0)--(2,1)--(1,1)--(1,2)--(3,2)--(4,2)--(4,3)--(5,3)--(5,1)--(7,1)--(7,0)--(6,0);
\filldraw (2,0) circle (0.15cm) node[below] {$u_1$};
\filldraw (6,0) circle (0.15cm) node[below] {$u_2$};
\foreach \y in {0,...,5}{\draw[-latex,red] (10,\y)--(10.5,\y);
\draw[-latex,red] (-0.5,\y)--(0,\y);}
\node at (5,-2) {\large $(1)$ $\phi_{1\to 2}=1$};
\end{scope}

\begin{scope}[xshift=15cm]
\draw[gray,densely dotted,line width=0.6pt] (-0.5,0) grid (10.5,5);
\foreach \y in {1,...,6}{\node at (-1,\y-1) {\scriptsize \y};
\node at (11,6-\y) {\scriptsize \y};}
\draw[very thick,blue] (2,0)--(1,0)--(1,1)--(-0.5,1);
\draw[very thick,blue] (10.5,4)--(9,4)--(9,5)--(7,5)--(7,4)--(8,4)--(8,3)--(5,3)--(5,2)--(6,2)--(6,0);
\filldraw (2,0) circle (0.15cm) node[below] {$u_1$};
\filldraw (6,0) circle (0.15cm) node[below] {$u_2$};
\foreach \y in {0,...,5}{\draw[-latex,red] (10,\y)--(10.5,\y);
\draw[-latex,red] (-0.5,\y)--(0,\y);}
\node at (5,-2) {\large $(2)$ $\phi_{1\to 2}=e^{-\i\theta}$};
\end{scope}

\begin{scope}[xshift=30cm]
\draw[gray,densely dotted,line width=0.6pt] (-0.5,0) grid (10.5,5);
\foreach \y in {1,...,6}{\node at (-1,\y-1) {\scriptsize \y};
\node at (11,6-\y) {\scriptsize \y};}
\draw[very thick,blue] (2,0)--(3,0)--(3,2)--(2,2)--(2,3)--(1,3)--(1,2)--(-0.5,2);
\draw[very thick,blue] (10.5,3)--(9,3)--(9,4)--(7,4)--(7,3)--(5,3)--(5,2)--(4,2)--(4,4)--(3,4)--(3,5)--(1,5)--(1,4)--(-0.5,4);
\draw[very thick,blue] (10.5,1)--(9,1)--(9,2)--(7,2)--(7,1)--(6,1)--(6,0);
\filldraw (2,0) circle (0.15cm) node[below] {$u_1$};
\filldraw (6,0) circle (0.15cm) node[below] {$u_2$};
\foreach \y in {0,...,5}{\draw[-latex,red] (10,\y)--(10.5,\y);
\draw[-latex,red] (-0.5,\y)--(0,\y);}
\node at (5,-2) {\large $(3)$ $\phi_{1\to 2}=e^{-2\i\theta}$};
\end{scope}

\begin{scope}[yshift=-10cm]
\draw[gray,densely dotted,line width=0.6pt] (-0.5,0) grid (10.5,5);
\foreach \y in {1,...,6}{\node at (-1,\y-1) {\scriptsize \y};
\node at (11,6-\y) {\scriptsize \y};}
\draw[very thick,blue] (2,0)--(2,1)--(1,1)--(1,2)--(3,2)--(4,2)--(4,3)--(5,3)--(5,1)--(7,1)--(7,2)--(9,2)--(9,1)--(10.5,1);
\draw[very thick,blue] (-0.5,4)--(1,4)--(1,5)--(3,5)--(3,4)--(4,4)--(4,5);
\filldraw (2,0) circle (0.15cm) node[below] {$u_1$};
\filldraw (4,5) circle (0.15cm) node[above] {$u_2$};
\foreach \y in {0,...,5}{\draw[-latex,red] (10,\y)--(10.5,\y);
\draw[-latex,red] (-0.5,\y)--(0,\y);}
\node at (5,-2) {\large $(1)$ $\phi_{1\to 2}=e^{\i\theta}$};
\end{scope}

\begin{scope}[xshift=15cm,yshift=-10cm]
\draw[gray,densely dotted,line width=0.6pt] (-0.5,0) grid (10.5,5);
\foreach \y in {1,...,6}{\node at (-1,\y-1) {\scriptsize \y};
\node at (11,6-\y) {\scriptsize \y};}
\draw[very thick,blue] (2,0)--(1,0)--(1,1)--(3,1)--(3,2)--(4,2)--(4,3)--(2,3)--(2,4)--(5,4)--(5,5)--(4,5);
\filldraw (2,0) circle (0.15cm) node[below] {$u_1$};
\filldraw (4,5) circle (0.15cm) node[above] {$u_2$};
\foreach \y in {0,...,5}{\draw[-latex,red] (10,\y)--(10.5,\y);
\draw[-latex,red] (-0.5,\y)--(0,\y);}
\node at (5,-2) {\large $(2)$ $\phi_{1\to 2}=1$};
\end{scope}

\begin{scope}[xshift=30cm,yshift=-10cm]
\draw[gray,densely dotted,line width=0.6pt] (-0.5,0) grid (10.5,5);
\foreach \y in {1,...,6}{\node at (-1,\y-1) {\scriptsize \y};
\node at (11,6-\y) {\scriptsize \y};}
\draw[very thick,blue] (2,0)--(3,0)--(3,2)--(2,2)--(2,3)--(1,3)--(1,2)--(-0.5,2);
\draw[very thick,blue] (10.5,3)--(9,3)--(9,4)--(7,4)--(7,3)--(5,3)--(5,2)--(4,2)--(4,4)--(3,4)--(3,5)--(4,5);
\filldraw (2,0) circle (0.15cm) node[below] {$u_1$};
\filldraw (4,5) circle (0.15cm) node[above] {$u_2$};
\foreach \y in {0,...,5}{\draw[-latex,red] (10,\y)--(10.5,\y);
\draw[-latex,red] (-0.5,\y)--(0,\y);}
\node at (5,-2) {\large $(3)$ $\phi_{1\to 2}=e^{-\i\theta}$};
\end{scope}

\end{tikzpicture}
\caption{Möbius graph $\G$ drawn as a rectangular grid with twisted periodic boundary conditions in the horizontal direction (the numbers on the half-edges indicate how they should be glued together to obtain the graph $\G$). The oriented edges equipped with a nontrivial parallel transport $e^{\i\theta}$ are drawn with an arrow. Each panel depicts a representative of each class of simple paths between two boundary vertices $u_1,u_2$. There are three distinct classes of paths, whether $u_2$ is located on the bottom boundary of the rectangle (in the top row) or on the top boundary (in the bottom row).}
\label{Moebius_paths}
\end{figure}

More precisely, let us observe that an OCRG in the first or third class (i.e. with $\phi_{1\to 2}=1,e^{-2\i\theta}$) contains at most one cycle winding once around the strip, and any number of cycles winding twice. On the other hand, the OCRGs in the second class (with $\phi_{1\to 2}=e^{-\i\theta}$) are quite peculiar, as the existence of the path from $u_1$ to $u_2$ forbids any noncontractible cycle. This crucial observation allows one to write the partition function explicitly for all paths in groves from $u_1$ to $u_2$ as
\begin{equation}
\begin{split}
\Zr[{\textstyle{2\atop 1}}]&=\Gr_{1,2}\det\mathbf{\Delta}=\sum_{\text{OCRGs }\Gamma_{\vsig}}\prod_{\text{cycles }\alpha\in\Gamma_{\vsig}}(1-\varpi_{\alpha})\times\phi_{1\to 2}^{-1}\\
&=\sum_{k=0}^{\infty}(2{-}2\cos 2\theta)^k\left\{N^{(1)}_k+N^{(3)}_k e^{2\i\theta}+\left(\widetilde{N}^{(1)}_k+\widetilde{N}^{(3)}_k e^{2\i\theta}\right)(2{-}2\cos\theta)\right\}+N^{(2)}_0 e^{\i\theta},
\label{Moebius_Z12}
\end{split}
\end{equation}
where $\vsig={\textstyle{2\atop 1}}$, and $N^{(j)}_k$ (resp. $\widetilde{N}^{(j)}_k$) denotes the number of unoriented CRGs\footnote{The weight of an unoriented cycle $\alpha$ being defined as the sum of the weights in both directions, i.e. $(1-\varpi_{\alpha})+(1-\varpi_{\alpha^{-1}})=2-\varpi_{\alpha}-\varpi_{\alpha}^{-1}$.} in class $j$ with $k$ cycles winding twice around the strip and no cycle (resp. one cycle) winding once. In the limit $\theta\to 0$, Eq.~\eqref{Moebius_Z12} yields
\begin{equation}
Z[12]=\lim_{\theta\to 0}\Zr[{\textstyle{2\atop 1}}]=N^{(1)}_0+N^{(2)}_0+N^{(3)}_0,
\end{equation}
where $N^{(j)}_0$ is the number of spanning forests on $\G$ in class $j$, whose knowledge is required to compute winding probabilities on the Möbius strip:
\begin{equation}
\P_{\text{SF}}(\text{a random simple path from $u_1$ to $u_2$ is in class $j$})\equiv\P_{j}(u_1,u_2)=\frac{N^{(j)}_0}{N^{(1)}_0+N^{(2)}_0+N^{(3)}_0}.
\end{equation}

In order to extract the values of $N^{(j)}_0$ from Eq.~\eqref{Moebius_Z12}, we first use the zeroth- and first-order terms in $\theta$ (as on the cylinder), which yield respectively
\begin{equation}
\Zr[{\textstyle{2\atop 1}}]\big|_{\theta=0}=N^{(1)}_0+N^{(2)}_0+N^{(3)}_0,\quad\partial_{\theta}\Zr[{\textstyle{2\atop 1}}]\big|_{\theta=0}=\i\,N^{(2)}_0+2\i\,N^{(3)}_0.
\label{Moebius_eq1}
\end{equation}
To write a third independent equation, we evaluate Eq.~\eqref{Moebius_Z12} at $\theta=\pi/2$:
\begin{equation}
\Im\Zr[{\textstyle{2\atop 1}}]\big|_{\theta=\pi/2}=N^{(2)}_0.
\label{Moebius_eq2}
\end{equation}
Using Eqs. \eqref{Moebius_eq1} and \eqref{Moebius_eq2}, we find the following combinatorial expressions for winding probabilities on the Möbius strip:
\begin{equation}
\begin{split}
\P_1(u_1,u_2)&=1+\frac{\i}{2}\partial_{\theta}\log\Zr[{\textstyle{2\atop 1}}]\big|_{\theta=0}-\frac{1}{2}\frac{\Im\Zr[{\textstyle{2\atop 1}}]\big|_{\theta=\pi/2}}{\Zr[{\textstyle{2\atop 1}}]\big|_{\theta=0}},\\
\P_2(u_1,u_2)&=\frac{\Im\Zr[{\textstyle{2\atop 1}}]\big|_{\theta=\pi/2}}{\Zr[{\textstyle{2\atop 1}}]\big|_{\theta=0}},\\
\P_3(u_1,u_2)&=-\frac{\i}{2}\partial_{\theta}\log\Zr[{\textstyle{2\atop 1}}]\big|_{\theta=0}-\frac{1}{2}\frac{\Im\Zr[{\textstyle{2\atop 1}}]\big|_{\theta=\pi/2}}{\Zr[{\textstyle{2\atop 1}}]\big|_{\theta=0}},
\label{Moebius_prob}
\end{split}
\end{equation}
where $\Zr[{\textstyle{2\atop 1}}]=\Gr_{1,2}\det\mathbf{\Delta}$ can be computed in terms of the eigenvalues and eigenfunctions of $\mathbf{\Delta}=\mathbf{\Delta}(\theta)$, which we give below for the Möbius graph with either wired or free boundary conditions.

If instead we take the endpoint of the path on the top boundary of the rectangle, $u_2=(\x_2,M)$, the product of parallel transports along the path from $u_1$ to $u_2$ (with respect to the same zipper) reads $\phi_{1\to 2}=e^{\i\theta},1,e^{-\i\theta}$ for the first, second and third classes of paths, respectively (see the bottom row of Fig.~\ref{Moebius_paths}). Consequently, the winding probabilities are given by
\begin{equation}
\begin{split}
\widehat{\P}_1(u_1,u_2)&=\frac{1}{2}+\frac{\i}{2}\partial_{\theta}\log\Zr[{\textstyle{2\atop 1}}]\big|_{\theta=0}-\frac{1}{2}\frac{\Re\Zr[{\textstyle{2\atop 1}}]\big|_{\theta=\pi/2}}{\Zr[{\textstyle{2\atop 1}}]\big|_{\theta=0}},\\
\widehat{\P}_2(u_1,u_2)&=\frac{\Re\Zr[{\textstyle{2\atop 1}}]\big|_{\theta=\pi/2}}{\Zr[{\textstyle{2\atop 1}}]\big|_{\theta=0}},\\
\widehat{\P}_3(u_1,u_2)&=\frac{1}{2}-\frac{\i}{2}\partial_{\theta}\log\Zr[{\textstyle{2\atop 1}}]\big|_{\theta=0}-\frac{1}{2}\frac{\Re\Zr[{\textstyle{2\atop 1}}]\big|_{\theta=\pi/2}}{\Zr[{\textstyle{2\atop 1}}]\big|_{\theta=0}},
\label{Moebius_prob2}
\end{split}
\end{equation}
where the hat serves as a reminder that the endpoint $u_2$ is located on the top boundary of the rectangle. Although the combinatorial forms \eqref{Moebius_prob} and \eqref{Moebius_prob2} of the $\P_{j}$'s and the $\widehat{\P}_j$'s differ depending on the position of $u_2$, we shall show explicitly that they are related to one another by $\x_2\to\x_2{+}N$ for both choices of boundary conditions (wired or free), reflecting the periodicity of the Möbius strip.


\subsubsection{Wired boundary}
\label{sec4.3.1}
Let us start with a Möbius graph in which all boundary vertices are wired to the root $s$. Its Dirichlet line bundle Laplacian $\mathbf{\Delta}$ has the following eigenvalues and eigenfunctions:
\begin{align}
\boldsymbol{\lambda}_{m,n}&=4-2\cos\left(\frac{\pi m}{M+1}\right)-2\cos\left(\frac{2\pi n+\pi(m+1)-\theta}{N}\right),\\
\mathbf{f}_{m,n}(\x,\y)&=\left(\frac{2}{MN}\right)^{1/2}e^{\i(2\pi n+\pi(m+1)-\theta)\x/N}\sin\left(\frac{\pi m\y}{M+1}\right),
\end{align}
for $1\le m\le M$ and $0\le n\le N{-}1$. As in Section~\ref{sec4.2} (for the cylinder), we introduce the variables $x=\varepsilon\x$ and $y=\varepsilon\y$ on the lattice $\varepsilon\Z^2$. We take the scaling limit $\varepsilon\to 0^+$ and $M,N\to\infty$ such that $M\varepsilon\to p$ and $N\varepsilon\to 2\pi$. In that limit, the discrete Green function converges to the continuum Green function of a Möbius strip of width $p$ and perimeter $4\pi$, given by
\begin{equation}
\mathfrak{G}(x_1,y_1;x_2,y_2;\theta)=\frac{1}{2\pi p}\sum_{m\in\Z^*}\sum_{n\in\Z}\frac{e^{\i\left(n+\frac{m+1}{2}-\frac{\theta}{2\pi}\right)(x_1-x_2)}\sin\left(\frac{\pi m y_1}{p}\right)\sin\left(\frac{\pi m y_2}{p}\right)}{\left(n+\frac{m+1}{2}-\frac{\theta}{2\pi}\right)^2+\left(\frac{\pi m}{p}\right)^2}.
\end{equation}
It is useful to separate the series over $m$ into two parts, according to the parity of $m$, to compute the associated excursion Poisson kernel, which reads at $(x_1,0),(x_2,0)$:
\begin{equation}
\text{P}(x,p;\theta)=\lim_{\delta,\varepsilon\to 0^+}\frac{1}{\delta\varepsilon}\mathfrak{G}(x_1,\delta;x_2,\varepsilon;\theta)=\frac{\pi}{2p^2}\sum_{k\in\Z}\frac{e^{\i k\theta}\left\{(-1)^k+\cosh\left[(2\pi k-x)\pi/p\right]\right\}}{\sinh^2\left[(2\pi k-x)\pi/p\right]},
\label{Moebius_Pkern}
\end{equation}
where $x=x_2-x_1$. If on the other hand $u_2=(x_2,p)$, we find the excursion Poisson kernel
\begin{equation}
\widehat{\text{P}}(x,p;\theta)=\lim_{\delta,\varepsilon\to 0^+}\frac{1}{\delta\varepsilon}\mathfrak{G}(x_1,\delta;x_2,p-\varepsilon;\theta)=\frac{\pi}{2p^2}\sum_{k\in\Z}\frac{e^{\i k\theta}\left\{(-1)^k-\cosh\left[(2\pi k-x)\pi/p\right]\right\}}{\sinh^2\left[(2\pi k-x)\pi/p\right]}.
\label{Moebius_Pkern2}
\end{equation}

In the continuum, one finds the following expressions for the winding probabilities on the strip:
\begin{equation}
\begin{split}
\P^{\text{D}}_1(x,p)&=1+\frac{\i}{2}\frac{\partial_{\theta}\text{P}(x,p;0)}{\text{P}(x,p;0)}-\frac{1}{2}\frac{\det\mathbf{\Delta}(\pi/2)}{\det\mathbf{\Delta}(0)}\frac{\Im\text{P}(x,p;\pi/2)}{\text{P}(x,p;0)},\\
\P^{\text{D}}_2(x,p)&=\frac{\det\mathbf{\Delta}(\pi/2)}{\det\mathbf{\Delta}(0)}\frac{\Im\text{P}(x,p;\pi/2)}{\text{P}(x,p;0)},\\
\P^{\text{D}}_3(x,p)&=-\frac{\i}{2}\frac{\partial_{\theta}\text{P}(x,p;0)}{\text{P}(x,p;0)}-\frac{1}{2}\frac{\det\mathbf{\Delta}(\pi/2)}{\det\mathbf{\Delta}(0)}\frac{\Im\text{P}(x,p;\pi/2)}{\text{P}(x,p;0)},
\label{Moebius_prob_cont}
\end{split}
\end{equation}
when $u_2=(x_2,0)$ is on the bottom boundary (similar formulas hold when $u_2=(x_2,p)$). Comparing Eqs.~\eqref{Moebius_Pkern} and \eqref{Moebius_Pkern2}, we find that $\text{P}(x{+}2\pi,p;\theta)=-e^{\i\theta}\widehat{\text{P}}(x,p;\theta)$, from which the relations
\begin{equation}
\widehat{\P}^{\text{D}}_j(x,p)=\P^{\text{D}}_j(x+2\pi,p)
\end{equation}
follow. It is worth noting that the excursion Poisson kernel also satisfies the identity $\text{P}(4\pi{-}x,p;\theta)=e^{2\i\theta}\text{P}(x,p;-\theta)$, implying that
\begin{equation}
\P^{\text{D}}_j(4\pi-x,p)=\P^{\text{D}}_{4-j}(x,p)
\label{Moebius_sym}
\end{equation}
for $1\le j\le 3$. It suffices therefore to compute the winding probabilities for $0\le x\le 2\pi$, which we assume is the case for the rest of this section.

Let us now come back the excursion Poisson kernel \eqref{Moebius_Pkern}. Its zeroth and first orders in $\theta$ can be written in terms of Jacobi's theta and elliptic functions (whose definitions are recalled in the Appendix~\ref{a1}) as follows, provided $0\le x\le 2\pi$:
\begin{equation}
\text{P}(x,p;0)=-\frac{1}{2\pi p}\partial_x F(x,p),\quad\partial_{\theta}\text{P}(x,p;0)=-\frac{\i}{4\pi^2 p}\left(x\,\partial_x+p\,\partial_p\right)F(x,p),
\end{equation}
with the function $F(x,p)$ given by
\begin{equation}
\begin{split}
F(x,p)&=\frac{p}{2}\vartheta_3^2\left(\text{ns}(\vartheta_3^2 x/2,\vartheta_2^2/\vartheta_3^2)+\text{cs}(\vartheta_3^2 x/2,\vartheta_2^2/\vartheta_3^2)\right)=\frac{p}{2}\cot(x/4)-2p\sum_{n=1}^{\infty}\frac{\sin(nx/2)}{(-1)^n e^{np/2}+1}\\
&=\pi\coth(\pi x/(2p))+4\pi\sum_{n=1}^{\infty}\frac{\sinh(n\pi x/p)}{(-1)^n e^{2\pi^2 n/p}+1},
\label{Moebius_F}
\end{split}
\end{equation}
with $\vartheta_a\equiv\vartheta_a(0,e^{-p/2})$ for $a=2,3$, and where we used the modular properties of Jacobi's theta and elliptic functions to write the last equality. The two series representations of $F(x,p)$ will be used below to compute the asymptotics of $\P^{\text{D}}_j(x,p)$ for $p\to\infty$ and $p\to 0^+$.

Similarly, the imaginary part of the excursion Poisson kernel at $\theta=\pi/2$ can be recast into the form
\begin{equation}
\Im\text{P}(x,p;\pi/2)=-\frac{1}{2\pi p}\partial_x\widetilde{F}(x,p),
\end{equation}
with the auxiliary function $\widetilde{F}(x,p)$ defined as
\begin{equation}
\begin{split}
\widetilde{F}(x,p)&=\frac{p}{2}\vartheta_2^2\,\text{cd}(\vartheta_3^2 x/4,\vartheta_2^2/\vartheta_3^2)=2p\sum_{n=0}^{\infty}\frac{(-1)^n e^{(2n+1)p/4}\cos[(2n+1)x/4]}{e^{(2n+1)p/2}-1}\\
&=\pi+4\pi\sum_{n=1}^{\infty}\frac{(-1)^n e^{2\pi^2 n/p}\cosh(n\pi x/p)}{e^{4\pi^2 n/p}+1}.
\label{Moebius_Ft}
\end{split}
\end{equation}

In addition to the excursion Poisson kernel, knowledge of the determinant of the Laplacian is also required to compute winding probabilities on the Möbius strip using \eqref{Moebius_prob_cont}. For the discrete graph $\G$, it reads
\begin{equation}
\det\mathbf{\Delta}(\theta)=\prod_{m=1}^{M}\prod_{n=0}^{N-1}\left[4-2\cos\left(\frac{\pi m}{M+1}\right)-2\cos\left(\frac{2\pi n+\pi(m+1)-\theta}{N}\right)\right].
\end{equation}
A similar determinant was computed in \cite{Hag09a} for $\theta=0$; we follow essentially the same procedure here. We start by introducing the variable $t_m$ defined by the relation $\cosh t_m=2-\cos(\pi m/(M{+}1))$. For $m\ll M$, $t_m=\pi m/(M{+}1)+\ldots$ at leading order. This new variable enables us to compute the product over $n$:
\begin{equation}
\begin{split}
&\prod_{n=0}^{N-1}\left(2\cosh t_m-2\cos\left(\frac{2\pi n+\pi(m+1)-\theta}{N}\right)\right)\\
&=\prod_{n=0}^{N-1}e^{t_m}\left(1-e^{-t_m+\i(2\pi n+\pi(m+1)-\theta)/N}\right)\left(1-e^{-t_m-\i(2\pi n+\pi(m+1)-\theta)/N}\right)\\
&=e^{N t_m}\left(1-e^{-N t_m+\i\pi(m+1)+\i\theta}\right)\left(1-e^{-N t_m+\i\pi(m+1)-\i\theta}\right),
\end{split}
\end{equation}
where the second equality comes from the factorization $q^N-1=\prod_{j=0}^{N-1}(q-e^{\i 2\pi j/N})$. The logarithm of the determinant therefore reads
\begin{equation}
\begin{split}
\log\det\mathbf{\Delta}(\theta)&=\log\left(4\sin^2(\theta/2)\right)+\sum_{m=1}^{M-1}N\,t_m\\
&\hspace{1cm}+\sum_{m=1}^{M-1}\log\left[\left(1-e^{-Nt_m+\i\pi(m+1)+\i\theta}\right)\left(1-e^{-Nt_m+\i\pi(m+1)-\i\theta}\right)\right],
\label{det_asy}
\end{split}
\end{equation}
where the first term on the right-hand side comes from the contribution $m=0$ in $\det\mathbf{\Delta}(\theta)$. The second term may be evaluated perturbatively by applying the Euler-Maclaurin formula: 
\begin{equation}
\sum_{m=1}^{M}t_m=\frac{4G}{\pi}(M+1)-\frac{1}{2}\log(3+2\sqrt{2})-\frac{\pi}{12M}+\ldots
\end{equation}
up to corrections of order $1/M^2$. Here $G=0.915966...$ is Catalan's constant. In the last term of Eq.~\eqref{det_asy}, the main contribution to the sum comes from the values of $m\ll M$, so we can replace $t_m$ with $\pi m/M$. We may therefore express the determinant of the Laplacian as follows in the scaling limit:
\begin{equation}
\begin{split}
\det\mathbf{\Delta}(\theta)&\simeq\exp\left(\frac{4G}{\pi}(M+1)N-\frac{1}{2}\log(3+2\sqrt{2})N\right)e^{-\pi N/(12M)}\times 4\sin^2(\theta/2)\\
&\hspace{1cm}\times\prod_{m=1}^{\infty}\left(1-e^{-N\pi m/M+\i\pi(m+1)+\i\theta}\right)\left(1-e^{-N\pi m/M+\i\pi(m+1)-\i\theta}\right).
\end{split}
\end{equation}
It should be noted that the first exponential factor diverges as $M,N\to\infty$. However, since it is a mere constant and Eq.~\eqref{Moebius_prob_cont} only involves ratios of this determinant, we may discard it to define a \emph{regularized} determinant, which we can write in terms of Dedekind's eta function $\eta(q)$ and Jacobi's theta function $\vartheta_2(z,q)$ as
\begin{equation}
\left(\det\mathbf{\Delta}(\theta)\right)_{\text{reg}}=\frac{\vartheta_2(\theta/2,\i\,e^{-\pi^2/p})}{2\cos(\theta/2)\,\eta(-\i\,e^{-\pi^2/p})}.
\end{equation}
To calculate explicitly the winding probabilities \eqref{Moebius_prob}, we need to compute both its logarithmic derivative at $\theta=0$ and the ratio of determinants evaluated at $\pi/2$ and $0$. It is straightforward to see that the former vanishes, while the latter is equal to
\begin{equation}
\begin{split}
\frac{\left(\det\mathbf{\Delta}(\pi/2)\right)_{\text{reg}}}{\left(\det\mathbf{\Delta}(0)\right)_{\text{reg}}}&=\sqrt{2}\,\frac{\vartheta_2(\pi/4,e^{-2\pi^2/p})\vartheta_4(\pi/4,e^{-2\pi^2/p})}{\vartheta_2(0,e^{-2\pi^2/p})\vartheta_4(0,e^{-2\pi^2/p})}\\
&=\sqrt{2}\,e^{-p/16}\frac{\vartheta_2(\i p/8,e^{-p/2})\vartheta_4(\i p/8,e^{-p/2})}{\vartheta_2(0,e^{-p/2})\vartheta_4(0,e^{-p/2})}\equiv\sqrt{2}\,\Theta_{2,4}(p).
\label{Moebius_Th}
\end{split}
\end{equation}
Putting all the pieces together yields the following formulas for winding probabilities on the Möbius strip with Dirichlet boundary conditions:
\begin{equation}
\begin{split}
\P^{\text{D}}_1(x,p)&=1-\frac{1}{4\pi}\frac{(x\,\partial_x+p\,\partial_p)F(x,p)}{\partial_x F(x,p)}-\frac{1}{\sqrt{2}}\,\Theta_{2,4}(p)\frac{\partial_x\widetilde{F}(x,p)}{\partial_x F(x,p)},\\
\P^{\text{D}}_2(x,p)&=\sqrt{2}\,\Theta_{2,4}(p)\frac{\partial_x\widetilde{F}(x,p)}{\partial_x F(x,p)},\quad\P^{\text{D}}_3(x,p)=1-\P^{\text{D}}_1(x,p)-\P^{\text{D}}_2(x,p),
\end{split}
\end{equation}
with $F,\widetilde{F},\Theta_{2,4}$ defined by Eqs. \eqref{Moebius_F}, \eqref{Moebius_Ft} and \eqref{Moebius_Th}, respectively. These probabilities are illustrated in Fig.~\ref{Moebius_D_plot}, in the asymptotic cases $p\to\infty$ and $p\to 0^+$. For the former limit, we find for $p\gg 1$ that
\begin{equation}
\begin{split}
\frac{(x\,\partial_x+p\,\partial_p)F(x,p)}{\partial_x F(x,p)}&\simeq x-2\sin(x/2)+8\sin(x/2)\sin^2(x/4)\,p\,e^{-p/2},\\
\Theta_{2,4}(p)&\simeq\frac{1}{2}e^{p/16}+\frac{1}{2}e^{-7p/16},\\
\frac{\partial_x\widetilde{F}(x,p)}{\partial_x F(x,p)}&\simeq 4\sin^3(x/4)\,e^{-p/4}-8\left(2+\cos(x/2)+\cos x\right)\sin^3(x/4)\,e^{-3p/4}.
\end{split}
\end{equation}
For large strips, the winding probabilities therefore read
\begin{equation}
\begin{split}
\P^{\text{D}}_1(x,p)&\simeq 1-\frac{x-2\sin(x/2)}{4\pi}-\sqrt{2}\sin^3(x/4)\,e^{-3p/16},\quad\P^{\text{D}}_2(x,p)\simeq 2\sqrt{2}\sin^3(x/4)\,e^{-3p/16},\\
\P^{\text{D}}_3(x,p)&=1-\P^{\text{D}}_1(x,p)-\P^{\text{D}}_2(x,p).
\end{split}
\end{equation}
In contrast, we also consider the case $p\ll 1$ of very thin strips, for which we use the modular properties of Jacobi's theta and elliptic functions to obtain the expansions
\begin{equation}
\begin{split}
\frac{(x\,\partial_x+p\,\partial_p)F(x,p)}{\partial_x F(x,p)}&\simeq 2\pi\frac{\partial_x\widetilde{F}(x,p)}{\partial_x F(x,p)}\simeq\frac{2\pi\,e^{-\pi(2\pi-x)/p}}{e^{-\pi x/p}+e^{-\pi(2\pi-x)/p}},\\
\Theta_{2,4}(p)&\simeq\frac{1}{\sqrt{2}}+\sqrt{2}\,e^{-2\pi^2/p}.
\end{split}
\end{equation}
We discuss these asymptotic results for $p\to\infty$ and $p\to 0^+$ below.

\begin{figure}[h]
\centering
\begin{tikzpicture}

\begin{scope}
\node at (0,0) {\includegraphics[scale=0.65]{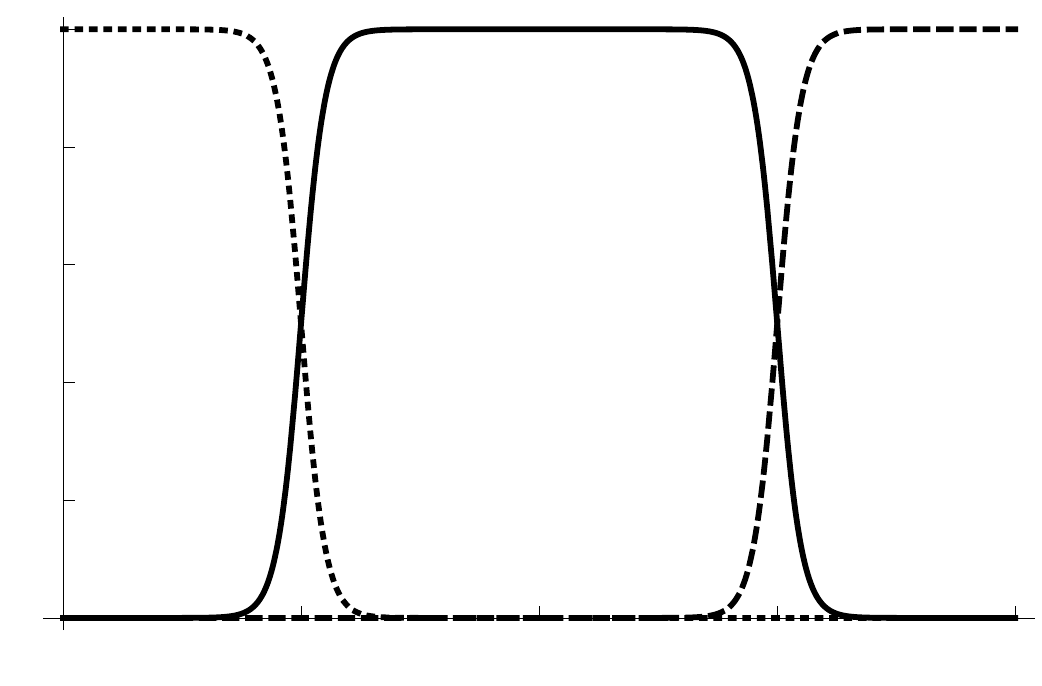}};
\node at (3.7,-1.7) {$x$};
\node at (-2.5,2.5) {$\P^{\text{D}}_{j}(x,p{=}1)$};
\node at (-1.4,-2.1) {$\pi$};
\node at (0.15,-2.1) {$2\pi$};
\node at (1.7,-2.1) {$3\pi$};
\node at (3.25,-2.1) {$4\pi$};
\node at (-3.3,-1) {$0.2$};
\node at (-3.3,-0.25) {$0.4$};
\node at (-3.3,0.5) {$0.6$};
\node at (-3.3,1.25) {$0.8$};
\node at (-3.3,2) {$1$};
\end{scope}

\begin{scope}[xshift=8cm]
\node at (0,0) {\includegraphics[scale=0.65]{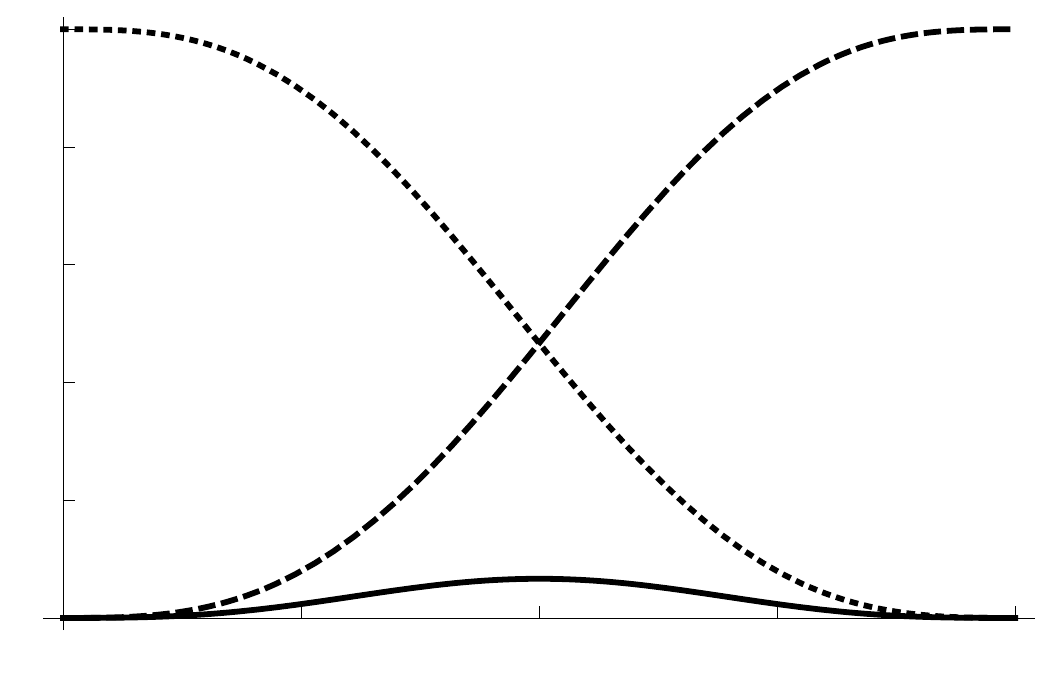}};
\node at (3.7,-1.7) {$x$};
\node at (-2.5,2.5) {$\P^{\text{D}}_{j}(x,p{=}20)$};
\node at (-1.4,-2.1) {$\pi$};
\node at (0.15,-2.1) {$2\pi$};
\node at (1.7,-2.1) {$3\pi$};
\node at (3.25,-2.1) {$4\pi$};
\node at (-3.3,-1) {$0.2$};
\node at (-3.3,-0.25) {$0.4$};
\node at (-3.3,0.5) {$0.6$};
\node at (-3.3,1.25) {$0.8$};
\node at (-3.3,2) {$1$};
\end{scope}

\end{tikzpicture}
\caption{Winding probabilities of a random simple path from $u_1=(x_1,0)$ to $u_2=(x_2,0)$ on a Möbius strip with Dirichlet boundary conditions, as functions of the separation $x=x_2-x_1$. The probabilities for the three winding classes ($1\le j\le 3$) are drawn with a dotted, solid, and dashed line, respectively.}
\label{Moebius_D_plot}
\end{figure}


\subsubsection{Free boundary}
\label{sec4.3.2}
We proceed similarly for the Möbius strip with free boundary conditions. Using the same zipper as for the wired strip, we find the eigenvalues and eigenfunctions of $\mathbf{\Delta}=\mathbf{\Delta}(\theta)$ to be given by
\begin{align}
\boldsymbol{\lambda}_{m,n}&=4-2\cos\left(\frac{\pi m}{M}\right)-2\cos\left(\frac{2\pi n+\pi m-\theta}{N}\right),\\
\mathbf{f}_{m,n}(\x,\y)&=\left(\frac{2-\delta_{m,0}}{MN}\right)^{1/2}e^{\i(2\pi n+\pi m-\theta)\x/N}\cos\left[\frac{\pi m}{M}\left(\y-\frac{1}{2}\right)\right],
\end{align}
for $0\le m\le M{-}1$ and $0\le n\le N{-}1$. The corresponding Green function in the scaling limit satisfies Neumann boundary conditions at $y=0$ and $y=p$, and reads
\begin{equation}
\mathfrak{G}(x_1,y_1;x_2,y_2;\theta)=\frac{1}{2\pi p}\sum_{m,n\in\Z}\frac{e^{\i\left(n+\frac{m}{2}-\frac{\theta}{2\pi}\right)(x_1-x_2)}\cos\left(\frac{\pi m y_1}{p}\right)\cos\left(\frac{\pi m y_2}{p}\right)}{\left(n+\frac{m}{2}-\frac{\theta}{2\pi}\right)^2+\left(\frac{\pi m}{p}\right)^2}.
\end{equation}
For boundary points, i.e. $y_1=y_2=0$, the Green function simplifies to
\begin{equation}
\mathfrak{G}(x,p;\theta)=\frac{1}{4\pi}\sum_{n\in\Z}\frac{e^{-\i\left(n-\frac{\theta}{2\pi}\right)x}\coth\left[\left(n-\frac{\theta}{2\pi}\right)\frac{p}{2}\right]}{n-\frac{\theta}{2\pi}}+\frac{1}{4\pi}\sum_{n\in\Z}\frac{e^{-\i\left(n+\frac{1}{2}-\frac{\theta}{2\pi}\right)x}\tanh\left[\left(n+\frac{1}{2}-\frac{\theta}{2\pi}\right)\frac{p}{2}\right]}{n+\frac{1}{2}-\frac{\theta}{2\pi}},
\label{Moebius_Gf_N}
\end{equation}
where $x=x_2-x_1$. For $\theta\to 0^+$, the Green function is singular because $\det\mathbf{\Delta}(\theta)\to 0$. More precisely, its Laurent series reads, up to regular terms,
\begin{equation}
\mathfrak{G}(x,p;\theta)=\frac{2\pi}{p\theta^2}+\frac{\i x}{p\theta}+\ldots
\end{equation}
As for the strip with Dirichlet boundary conditions, it suffices to study the winding probabilities for $0\le x\le 2\pi$. Indeed, the Green function satisfies the relations
\begin{align}
&\Im\mathfrak{G}(4\pi-x,p;\pi/2)=\Im\mathfrak{G}(x,p;\pi/2),\\
&\mathfrak{G}(4\pi-x,p;\theta)=\mathfrak{G}(x,p;\theta)+\frac{2\i(2\pi-x)}{p\theta}+\ldots,
\end{align}
implying that $\P^{\text{N}}_{j}(x)=\P^{\text{N}}_{4-j}(4\pi-x)$ for any $x$ between $2\pi$ and $4\pi$, for $1\le j\le 3$.

The regularized determinant is obtained by taking the same steps as in Section~\ref{sec4.2.3} and yields
\begin{equation}
\left(\det\mathbf{\Delta}(\theta)\right)_{\text{reg}}=\frac{2\sin(\theta/2)\,\vartheta_1(\theta/2,\i\,e^{-\pi^2/p})}{\eta(\i\,e^{-\pi^2/p})}.
\label{Moebius_det_N}
\end{equation}
Using Eqs. \eqref{Moebius_Gf_N} and \eqref{Moebius_det_N}, we find after some algebra the following winding probabilities on the Möbius strip with Neumann boundary conditions, in the scaling limit:
\begin{equation}
\P^{\text{N}}_1(x,p)=1-\frac{x}{4\pi}-H(x,p),\quad\P^{\text{N}}_2(x,p)=2H(x,p),\quad\P^{\text{N}}_3(x,p)=\frac{x}{4\pi}-H(x,p),
\end{equation}
where the auxiliary function $H(x,p)$ is defined by
\begin{equation}
H(x,p)=-\frac{\i}{\sqrt{2}\pi}e^{-p/16}\frac{\vartheta_1(\i p/8,e^{-p/2})\vartheta_3(\i p/8,e^{-p/2})}{\vartheta'_1(0,e^{-p/2})\vartheta_3(0,e^{-p/2})}\sum_{n\in\Z}\frac{\sin[(n+\tfrac{1}{4})x]}{(n+\tfrac{1}{4})\sinh[(n+\tfrac{1}{4})p]}.
\end{equation}
In the limit of large and thin strips, we find respectively for $0\le x\le 2\pi$:
\begin{align}
H(x,p)&\simeq\frac{2\sqrt{2}}{\pi}\sin(x/4)\,e^{-3p/16},\\
H(x,p)&\simeq\frac{x}{4\pi}-\frac{p}{2\pi^2}e^{-\pi(2\pi-x)/p}.
\end{align}
An illustration of winding probabilities on the Möbius strip with Neumann boundary conditions is provided in Fig.~\ref{Moebius_N_plot}.

\begin{figure}[h]
\centering
\begin{tikzpicture}

\begin{scope}
\node at (0,0) {\includegraphics[scale=0.65]{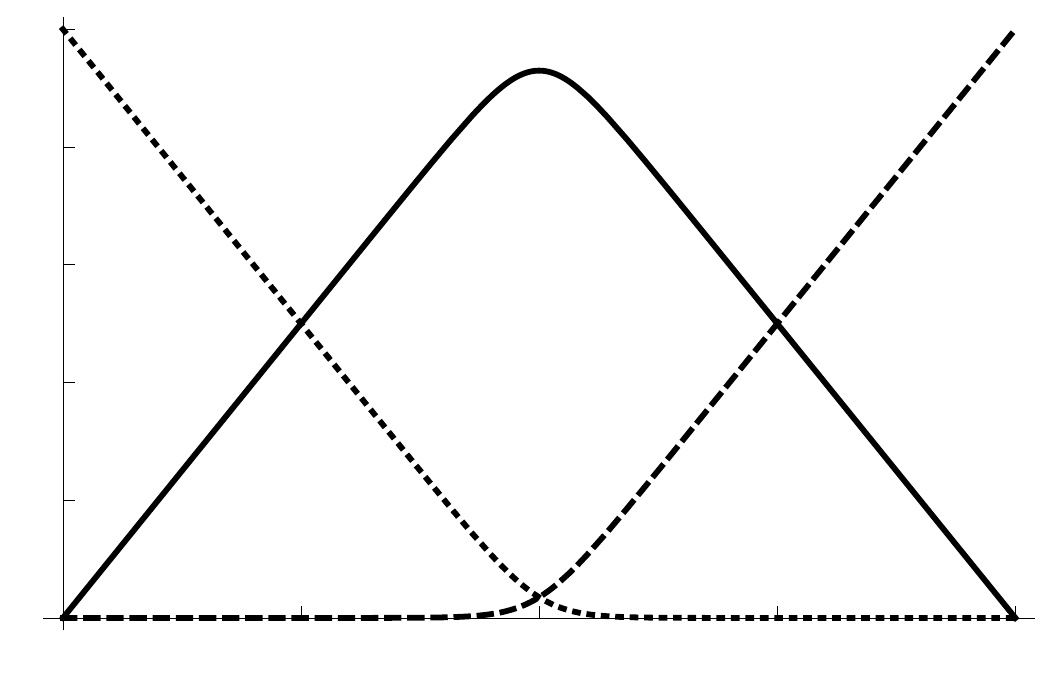}};
\node at (3.7,-1.7) {$x$};
\node at (-2.5,2.5) {$\P^{\text{N}}_{j}(x,p{=}1)$};
\node at (-1.4,-2.1) {$\pi$};
\node at (0.15,-2.1) {$2\pi$};
\node at (1.7,-2.1) {$3\pi$};
\node at (3.25,-2.1) {$4\pi$};
\node at (-3.3,-1) {$0.2$};
\node at (-3.3,-0.25) {$0.4$};
\node at (-3.3,0.5) {$0.6$};
\node at (-3.3,1.25) {$0.8$};
\node at (-3.3,2) {$1$};
\end{scope}

\begin{scope}[xshift=8cm]
\node at (0,0) {\includegraphics[scale=0.65]{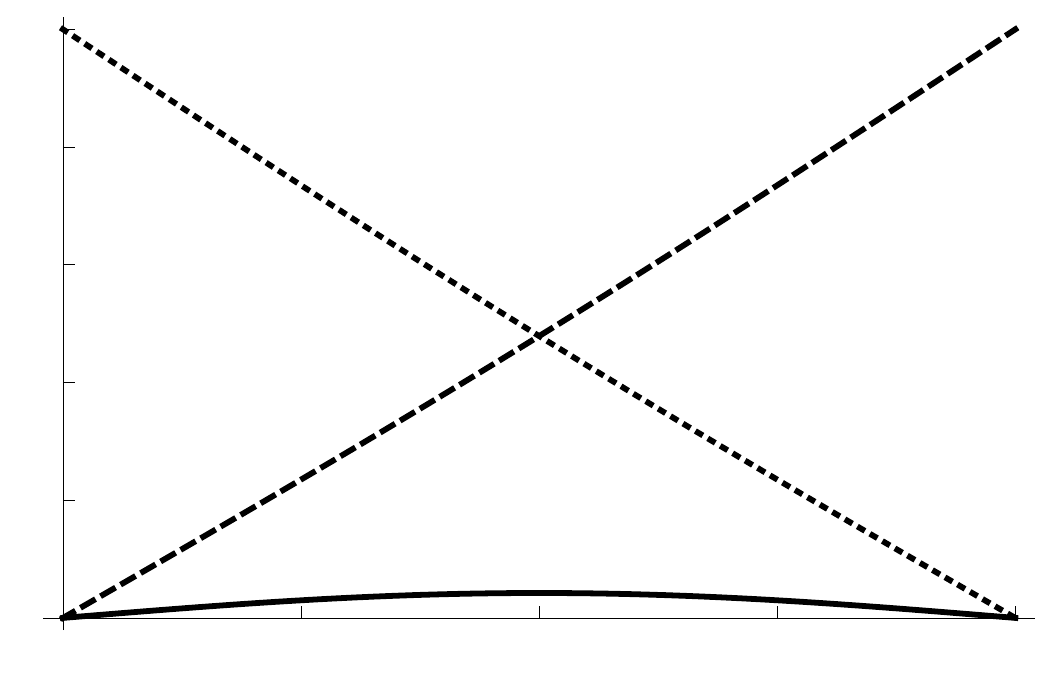}};
\node at (3.7,-1.7) {$x$};
\node at (-2.5,2.5) {$\P^{\text{N}}_{j}(x,p{=}20)$};
\node at (-1.4,-2.1) {$\pi$};
\node at (0.15,-2.1) {$2\pi$};
\node at (1.7,-2.1) {$3\pi$};
\node at (3.25,-2.1) {$4\pi$};
\node at (-3.3,-1) {$0.2$};
\node at (-3.3,-0.25) {$0.4$};
\node at (-3.3,0.5) {$0.6$};
\node at (-3.3,1.25) {$0.8$};
\node at (-3.3,2) {$1$};
\end{scope}

\end{tikzpicture}
\caption{Winding probabilities of a random simple path from $u_1=(x_1,0)$ to $u_2=(x_2,0)$ on a Möbius strip with Neumann boundary conditions, as functions of the separation $x=x_2-x_1$. The probabilities for the three winding classes ($1\le j\le 3$) are drawn with a dotted, solid, and dashed line, respectively.}
\label{Moebius_N_plot}
\end{figure}
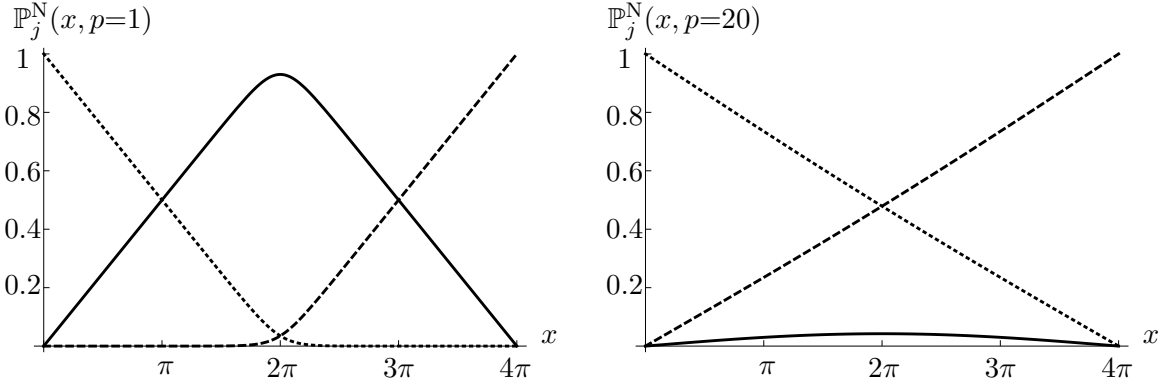

\subsubsection*{Asymptotics}
\vspace{-0.3cm}
For both types of boundary conditions, Dirichlet or Neumann, we have computed the winding probabilities in the limit of very large ($p\to\infty$) or very thin ($p\to 0^+$) strips. In the former case, we find that $\P_2(x,p)\to 0$ as $p\to\infty$, for $0<x<4\pi$. The intuitive interpretation is the following: all paths in the second class must cross the strip along its width $p$, whereas the paths in the first and third classes do not. As there are many more (typically shorter) paths in the first and third classes than in the second one, the winding probability $\P_2$ is suppressed in the limit $p\to\infty$. The two remaining probabilities can then be thought of as counterclockwise (class 1) and clockwise (class 3) probabilities. Indeed, we find precisely the same formulas for the Möbius strip and for the cylinder (see Eqs.~\eqref{PNN} and \eqref{PDD_large_p}), up to the rescaling $x\to 2x$ (for both boundary conditions):
\begin{equation}
\begin{split}
\P^{\text{D}}_1(x,\infty)&=1-\frac{x/2-\sin(x/2)}{2\pi},\quad\P^{\text{D}}_3(x,\infty)=1-\P^{\text{D}}_1(x,\infty),\\
\P^{\text{N}}_1(x,\infty)&=1-\frac{x/2}{2\pi},\quad\P^{\text{N}}_3(x,\infty)=1-\P^{\text{N}}_1(x,\infty).
\end{split}
\end{equation}
In the limit of thin strips $p\to 0^+$, we argue once again that shorter paths are favored over longer ones, which implies in particular that $\P_1(x,0)=0$ for $0<x<2\pi$ and $\P_3(x,0)=0$ for $2\pi<x<4\pi$. All winding probabilities converge to the same functions as on the cylinder, namely Heaviside functions (resp. piecewise linear functions) for Dirichlet (resp. Neumann) boundary conditions.


\section{Passage probabilities for multiple paths}
\label{sec5}

In Section~\ref{sec3}, we have defined a measure on simple paths in spanning forests, extended it to oriented cycle-rooted groves, and computed the marginal probability associated with the left- or right-passage of a random path with respect to a marked face of a graph $\G$ embedded on an orientable surface, i.e. Schramm's formula. The resulting probability \eqref{Schramm_form2} is expressed in terms of the standard Green function of the graph, whose explicit expression is well known for regular graphs such as the upper half-plane or the cylinder, considered in Section~\ref{sec4}.

Let us now generalize this formula to multiple nonintersecting simple paths between $2n$ nodes located on the boundary of the outer face of a planar graph $\G$. We denote the collection of these special vertices by $\mathcal{U}=\{u_1,\ldots,u_{2n}\}$ or simply $\{1,\ldots,2n\}$, and label them in counterclockwise order. Let $\sigma=r_1 s_1|\ldots|r_n s_n$ be a fixed (unoriented) pairing of the nodes, with $R=\{r_1,\ldots,r_n\}$ and $S=\{s_1,\ldots,s_n\}$ partitioning $\mathcal{U}$. We consider the set of all spanning forests on $\G_s=\G\cup\{s\}$ consisting in $n{+}1$ trees, each of which containing a single pair of nodes $\{r_i,s_i\}$ or the root $s$ (note that some pairings cannot be realized due to the planarity of $\G$). The measure \eqref{SF_weight} on single paths in two-component spanning forests generalizes naturally to multiple unoriented paths $\gamma_i:r_i\leftrightarrow s_i$, $1\le i\le n$, in $(n{+}1)$-component spanning forests as follows:
\begin{equation}
\w_{\text{SF}}(\gamma_1,\ldots,\gamma_n)=\sum_{\text{(n{+}1)SFs }\mathcal{F}\supset\,\bigcup_i\gamma_i}\w(\mathcal{F})=\prod_{i=1}^{n}C(\gamma_i)\times\det\Delta^{(\gamma_1,\ldots,\gamma_n)},
\label{nSF_weight}
\end{equation}
where $\Delta^{(\gamma_1,\ldots,\gamma_n)}$ is the restriction of the standard Laplacian to rows and columns indexed by vertices not in $\bigcup_i\gamma_i\cup\{s\}$. Similarly, the weight on multiple oriented paths $\gamma_i:r_i\to s_i$ in oriented cycle-rooted groves, defined in \eqref{CRG_weight} for a single path, is given by
\begin{equation}
\w_{\text{CRG}}(\gamma_1,\ldots,\gamma_n)=\sum_{\text{OCRGs }\Gamma_{\vsig}\supset\,\bigcup_i\gamma_i}\w(\Gamma_{\vsig})=\prod_{i=1}^{n}C(\gamma_i)\phi(\gamma_i)^{-1}\times\det\mathbf{\Delta}^{(\gamma_1,\ldots,\gamma_n)},
\label{nCRG_weight}
\end{equation}
where $\vsig={\textstyle{s_1\atop r_1}}|\cdots|{\textstyle{s_n\atop r_n}}$ and $\mathbf{\Delta}$ is the line bundle Laplacian of $\G_s$ with Dirichlet boundary conditions at $s$. Here the sum is over oriented cycle-rooted groves consisting in $n{+}1$ trees and any number of cycle-rooted trees (the latter do not contain any nodes).

Let $f$ be a face of the graph $\G$ and $\vsig$ be an oriented pairing of the $2n$ nodes. The OCRGs $\Gamma_{\vsig}$ can be sorted into $n{+}1$ \emph{winding classes}, according to the way the $n$ paths wind around the face $f$ (i.e. if they leave $f$ to their right or to their left). As a matter of convention, we shall assume all paths to be oriented toward the node with the highest label in each pair when we refer to the oriented pairing $\vsig$, and call this the canonical orientation. To distinguish between the different winding classes, we introduce a zipper as in Section~\ref{sec3.2}, going from $f$ to the outer face of $\G$ and intercepting the boundary path between nodes $u_{2n}$ and $u_1$. As can be seen for the case $n=3$ illustrated in Fig.~\ref{3Wind_cl}, two complications arise when considering multiple paths (as opposed to single paths):
\vspace{-0.3cm}
\begin{enumerate}[(i)]
\item $n$-tuples of paths belonging to distinct winding classes may have the same product of parallel transports $\prod_{i=1}^{n}\phi(\gamma_i)$ (which is constant over each class, as for the case $n=1$). For instance, triplets of paths in the second and third classes of Fig.~\ref{3Wind_cl} are weighted by a global factor $\phi_{1\to 4}\phi_{2\to 3}\phi_{5\to 6}=z^{-1}$ in both cases.
\item For a generic oriented pairing $\vsig$, Theorem \ref{grove_thm} does not yield the partition function $\Zr[\vsig]$ directly. Instead, it only allows one to write a system of linear equations relating partition functions for distinct pairings to minors of the line bundle Green function, which is computable for most regular graphs (perturbatively around $\Phi=\mathbf{1}$, at least).
\end{enumerate}
\vspace{-0.3cm}
Due to these obstacles, writing an explicit form for Schramm's formulas for multiple paths in terms of the Green function proves to be substantially more complicated than for a single path. In what follows, we first discuss the easiest case, namely pairs of paths, using the same formalism as in Section~\ref{sec3.2}, and give explicit formulas for the three distinct winding probabilities in terms of the Green function $G$ and its derivative $G'$, defined in \eqref{Gp_def}. As we shall see, obtaining such formulas a la Schramm is already cumbersome for only two paths, so a more systematic, combinatorial approach is required for $n>2$ paths. Such a technique will be provided in Section~\ref{sec5.2}, and relies on so-called \emph{cyclic Dyck paths} and \emph{cover-inclusive Dyck tilings}, introduced respectively in \cite{KW15} and \cite{KW11b,SZJ12}.

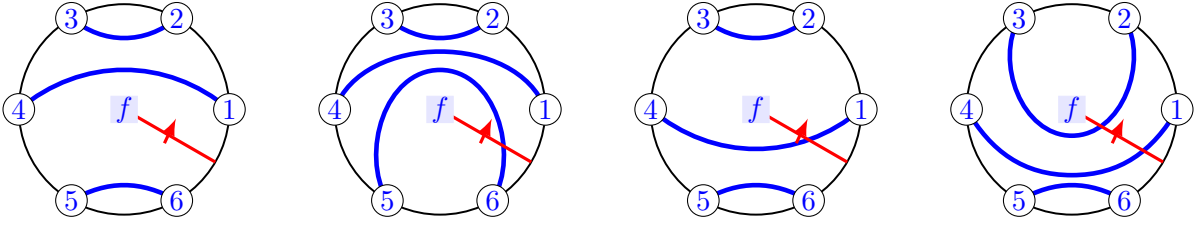
\begin{figure}[h]
\centering
\begin{tikzpicture}[scale=0.7]

\begin{scope}
\draw[line width=0.75pt] (0,0) circle (2cm);
\draw[line width=1.75pt,blue] (0:2)to[out=140,in=40](180:2);
\draw[line width=1.75pt,blue] (60:2)to[out=220,in=320](120:2);
\draw[line width=1.75pt,blue] (240:2)to[out=30,in=150](300:2);
\draw[fill=white] (0:2) circle (0.3cm) node[blue] {$1$};
\draw[fill=white] (60:2) circle (0.3cm) node[blue] {$2$};
\draw[fill=white] (120:2) circle (0.3cm) node[blue] {$3$};
\draw[fill=white] (180:2) circle (0.3cm) node[blue] {$4$};
\draw[fill=white] (240:2) circle (0.3cm) node[blue] {$5$};
\draw[fill=white] (300:2) circle (0.3cm) node[blue] {$6$};
\draw[very thick,red] (0,0)--(330:2);
\draw[very thick,red,-latex] (320:1)--(350:1);
\filldraw[blue!10!white] (-0.25,-0.25) rectangle node[blue] {$f$} (0.25,0.25);
\end{scope}

\begin{scope}[xshift=6cm]
\draw[line width=0.75pt] (0,0) circle (2cm);
\draw[line width=1.75pt,blue] (0:2)to[out=110,in=70](180:2);
\draw[line width=1.75pt,blue] (60:2)to[out=220,in=320](120:2);
\draw[line width=1.75pt,blue] (240:2)to[out=120,in=180](90:0.75)to[out=0,in=60](300:2);
\draw[fill=white] (0:2) circle (0.3cm) node[blue] {$1$};
\draw[fill=white] (60:2) circle (0.3cm) node[blue] {$2$};
\draw[fill=white] (120:2) circle (0.3cm) node[blue] {$3$};
\draw[fill=white] (180:2) circle (0.3cm) node[blue] {$4$};
\draw[fill=white] (240:2) circle (0.3cm) node[blue] {$5$};
\draw[fill=white] (300:2) circle (0.3cm) node[blue] {$6$};
\draw[very thick,red] (0,0)--(330:2);
\draw[very thick,red,-latex] (320:1)--(350:1);
\filldraw[blue!10!white] (-0.25,-0.25) rectangle node[blue] {$f$} (0.25,0.25);
\end{scope}

\begin{scope}[xshift=12cm]
\draw[line width=0.75pt] (0,0) circle (2cm);
\draw[line width=1.75pt,blue] (0:2)to[out=220,in=320](180:2);
\draw[line width=1.75pt,blue] (60:2)to[out=220,in=320](120:2);
\draw[line width=1.75pt,blue] (240:2)to[out=30,in=150](300:2);
\draw[fill=white] (0:2) circle (0.3cm) node[blue] {$1$};
\draw[fill=white] (60:2) circle (0.3cm) node[blue] {$2$};
\draw[fill=white] (120:2) circle (0.3cm) node[blue] {$3$};
\draw[fill=white] (180:2) circle (0.3cm) node[blue] {$4$};
\draw[fill=white] (240:2) circle (0.3cm) node[blue] {$5$};
\draw[fill=white] (300:2) circle (0.3cm) node[blue] {$6$};
\draw[very thick,red] (0,0)--(330:2);
\draw[very thick,red,-latex] (320:1)--(350:1);
\filldraw[blue!10!white] (-0.25,-0.25) rectangle node[blue] {$f$} (0.25,0.25);
\end{scope}

\begin{scope}[xshift=18cm]
\draw[line width=0.75pt] (0,0) circle (2cm);
\draw[line width=1.75pt,blue] (0:2)to[out=240,in=0](0,-1.25)to[out=180,in=300](180:2);
\draw[line width=1.75pt,blue] (60:2)to[out=300,in=0](0,-0.5)to[in=240,out=180](120:2);
\draw[line width=1.75pt,blue] (240:2)to[out=30,in=150](300:2);
\draw[fill=white] (0:2) circle (0.3cm) node[blue] {$1$};
\draw[fill=white] (60:2) circle (0.3cm) node[blue] {$2$};
\draw[fill=white] (120:2) circle (0.3cm) node[blue] {$3$};
\draw[fill=white] (180:2) circle (0.3cm) node[blue] {$4$};
\draw[fill=white] (240:2) circle (0.3cm) node[blue] {$5$};
\draw[fill=white] (300:2) circle (0.3cm) node[blue] {$6$};
\draw[very thick,red] (0,0)--(330:2);
\draw[very thick,red,-latex] (320:1)--(350:1);
\filldraw[blue!10!white] (-0.25,-0.25) rectangle node[blue] {$f$} (0.25,0.25);
\end{scope}

\end{tikzpicture}
\caption{Schematic representation of the four distinct winding classes for triplets of paths corresponding to the pairing $\vsig={\textstyle{{4\atop 1}|{3\atop 2}|{6\atop 5}}}$. The zipper is drawn as the red line from the marked face $f$ to the outer boundary. The oriented edges crossed by the zipper have a parallel transport $z$ in the direction of the arrow, and $z^{-1}$ in the opposite direction.}
\label{3Wind_cl}
\end{figure}


\subsection{Two paths}
\label{sec5.1}

Let us first compute Schramm's formulas for two paths between four boundary nodes, denoted by $\mathcal{U}=\{u_1,u_2,u_3,u_4\}\equiv\{1,2,3,4\}$, with respect to a given face $f$. There are two ways to connect these nodes in pairs by nonintersecting paths in OCRGs, corresponding to the partitions $12|34$ and $14|23$. The paths are further divided into three winding classes according to whether they leave $f$ to their left or to their right (recall that the paths are oriented toward the node with the higher index in each pair). We denote the partition functions associated with each of the six subclasses as follows:
\begin{equation*}
\Zr_{\L\L}[{\textstyle{2\atop 1}}|{\textstyle{4\atop 3}}],\,\,\Zr_{\L\R}[{\textstyle{2\atop 1}}|{\textstyle{4\atop 3}}],\,\,\Zr_{\R\L}[{\textstyle{2\atop 1}}|{\textstyle{4\atop 3}}],\,\,\Zr_{\L\L}[{\textstyle{4\atop 1}}|{\textstyle{3\atop 2}}],\,\,\Zr_{\R\L}[{\textstyle{4\atop 1}}|{\textstyle{3\atop 2}}],\,\,\Zr_{\R\R}[{\textstyle{4\atop 1}}|{\textstyle{3\atop 2}}],
\end{equation*}
where the indices refer to the left- or right-passage of the first and second paths with respect to $f$, respectively (some pairs of indices are not realizable on a planar graph, e.g. paths in the pairing $12|34$ cannot both leave $f$ to their right). We give an illustration of these six classes in Fig.~\ref{2Wind_cl}.

The number of winding classes for pairs of paths corresponds exactly to the number of partitions of $\mathcal{U}=\{u_1,u_2,u_3,u_4\}$ into two disjoint subsets of order two, $R$ and $S=\mathcal{U}\bs R$, that we can select in Theorem \ref{grove_thm}. For instance, if $R=\{2,4\}$ and $S=\{1,3\}$, the theorem reads
\begin{equation}
\begin{split}
\det\mathbf{\Delta}\det\Gr\textstyle{{1,3}\atop{2,4}}&=\Zr[{\textstyle{1\atop 2}}|{\textstyle{3\atop 4}}]-\Zr[{\textstyle{3\atop 2}}|{\textstyle{1\atop 4}}]\\
&=\left(\Zr_{\L\L}[{\textstyle{1\atop 2}}|{\textstyle{3\atop 4}}]+\Zr_{\L\R}[{\textstyle{1\atop 2}}|{\textstyle{3\atop 4}}]+\Zr_{\R\L}[{\textstyle{1\atop 2}}|{\textstyle{3\atop 4}}]\right)-\left(\Zr_{\L\L}[{\textstyle{3\atop 2}}|{\textstyle{1\atop 4}}]+\Zr_{\R\L}[{\textstyle{3\atop 2}}|{\textstyle{1\atop 4}}]+\Zr_{\R\R}[{\textstyle{3\atop 2}}|{\textstyle{1\atop 4}}]\right),
\end{split}
\end{equation}
where the indices refer to the left- or right-passage with respect to $f$ for paths with the canonical orientation (from the lower to the higher node index). As the product of parallel transports along the paths between the nodes is constant over each class separately, we may rewrite the preceding equation in terms of partition functions for paths with the canonical orientation:
\begin{equation}
\begin{split}
\det\mathbf{\Delta}\det\Gr\textstyle{{1,3}\atop{2,4}}&=\Zr_{\L\L}[{\textstyle{2\atop 1}}|{\textstyle{4\atop 3}}]+z^{-2}\Zr_{\L\R}[{\textstyle{2\atop 1}}|{\textstyle{4\atop 3}}]+z^{-2}\Zr_{\R\L}[{\textstyle{2\atop 1}}|{\textstyle{4\atop 3}}]\\
&\quad-\Zr_{\L\L}[{\textstyle{4\atop 1}}|{\textstyle{3\atop 2}}]-z^{-2}\Zr_{\R\L}[{\textstyle{4\atop 1}}|{\textstyle{3\atop 2}}]-z^{-2}\Zr_{\R\R}[{\textstyle{4\atop 1}}|{\textstyle{3\atop 2}}].
\end{split}
\end{equation}

Anticipating the general discussion for $n>2$ paths, we shall rather use another choice of orientation, which consists in taking the canonical orientation on paths that do not cross the zipper ($\L$ paths), and the reverse on paths that do ($\R$ paths). We shall explain the motivation behind this convention in Section~\ref{sec5.2}. With respect to this new orientation, one finds the equation
\begin{equation}
\begin{split}
\det\mathbf{\Delta}\det\Gr\textstyle{{1,3}\atop{2,4}}&=\Zr_{\L\L}[{\textstyle{2\atop 1}}|{\textstyle{4\atop 3}}]+\Zr_{\L\R}[{\textstyle{2\atop 1}}|{\textstyle{3\atop 4}}]+\Zr_{\R\L}[{\textstyle{1\atop 2}}|{\textstyle{4\atop 3}}]\\
&\quad-\Zr_{\L\L}[{\textstyle{4\atop 1}}|{\textstyle{3\atop 2}}]-\Zr_{\R\L}[{\textstyle{1\atop 4}}|{\textstyle{3\atop 2}}]-z^2\Zr_{\R\R}[{\textstyle{1\atop 4}}|{\textstyle{2\atop 3}}]
\end{split}
\end{equation}
for $R=\{2,4\}$ and $S=\{1,3\}$. Doing so for the five other subsets $R\subset\mathcal{U}$ of order two yields a system of equations represented by a matrix $\mathbf{A}_2$ given by
\begin{equation*}
\bordermatrix{& \Zr_{\L\L}[{\textstyle{2\atop 1}}|{\textstyle{4\atop 3}}] & \Zr_{\L\R}[{\textstyle{2\atop 1}}|{\textstyle{3\atop 4}}] & \Zr_{\R\L}[{\textstyle{1\atop 2}}|{\textstyle{4\atop 3}}] & \Zr_{\L\L}[{\textstyle{4\atop 1}}|{\textstyle{3\atop 2}}] & \Zr_{\R\L}[{\textstyle{1\atop 4}}|{\textstyle{3\atop 2}}] & \Zr_{\R\R}[{\textstyle{1\atop 4}}|{\textstyle{2\atop 3}}]\cr
\det\mathbf{\Delta}\det\Gr\textstyle{{3,4}\atop{1,2}} & 0 & 0 & 0 & -1 & -w & -w^2\cr
\det\mathbf{\Delta}\det\Gr\textstyle{{2,4}\atop{1,3}} & 1 & w & w & -1 & -w & -w\cr
\det\mathbf{\Delta}\det\Gr\textstyle{{2,3}\atop{1,4}} & 1 & 1 & w & 0 & 0 & 0\cr
\det\mathbf{\Delta}\det\Gr\textstyle{{1,4}\atop{2,3}} & 1 & w & 1 & 0 & 0 & 0\cr
\det\mathbf{\Delta}\det\Gr\textstyle{{1,3}\atop{2,4}} & 1 & 1 & 1 & -1 & -1 & -w\cr
\det\mathbf{\Delta}\det\Gr\textstyle{{1,2}\atop{3,4}} & 0 & 0 & 0 & -1 & -1 & -1\cr
},
\end{equation*}
where $w\equiv z^2$. $\mathbf{A}_2$ is invertible for $z\neq\pm 1$, as its determinant is equal to $(1-w)^5$. Explicitly, one obtains the inverse
\begin{equation}
\mathbf{A}_2^{-1}=\frac{1}{(1-w)^2}\begin{pmatrix}
-(1+w) & 1+w & -2w & -2w & w+w^2 & -(w+w^2)\\
1 & -1 & 1 & w & -w & w\\
1 & -1 & w & 1 & -w & w\\
-1 & w & -w & -w & w & -w^2\\
2 & -(1+w) & 1+w & 1+w & -(1+w) & 2w\\
-1 & 1 & -1 & -1 & 1 & -1
\end{pmatrix},
\end{equation}
from which the partition functions of interest can be extracted. For example, we have 
\begin{equation}
\begin{split}
&Z_{\L\R}[12|34]=\lim_{z\to 1}\Zr_{\L\R}[{\textstyle{2\atop 1}}|{\textstyle{3\atop 4}}]\\
&=\lim_{z\to 1}\frac{\det\mathbf{\Delta}}{(1-z^2)^2}\Big\{\det\Gr\textstyle{{3,4}\atop{1,2}}-\det\Gr\textstyle{{2,4}\atop{1,3}}+\det\Gr\textstyle{{2,3}\atop{1,4}}+z^2\det\Gr\textstyle{{1,4}\atop{2,3}}-z^2\det\Gr\textstyle{{1,3}\atop{2,4}}+z^2\det\Gr\textstyle{{1,2}\atop{3,4}}\Big\}\\
&=\det\Delta\big\{G_{1,2}\,G'_{3,4}-G_{1,3}\,G'_{2,4}+G_{1,4}\,G'_{2,3}-G'_{1,2}\,G'_{3,4}+G'_{1,3}\,G'_{2,4}-G'_{1,4}\,G'_{2,3}\big\},
\end{split}
\end{equation}
where $G'$ is the derivative of the Green function defined in \eqref{Gp_def}. The total partition function $Z[12|34]$ can be obtained by summing the partition functions for all winding classes $\L\L$, $\L\R$, $\R\L$ of the pairing $12|34$, or computed directly from Theorem~\ref{AMMT}, which yields
\begin{equation}
Z[12|34]=\det\Delta\det G_{1,4}^{2,3}=\det\Delta\{G_{1,2}\,G_{3,4}-G_{1,3}\,G_{2,4}\},
\end{equation}
as the pairing $13|24$ cannot be realized on a planar graph, implying $Z[13|24]=0$. After some algebra, one finds that the probability that a pair of paths in $12|34$ belongs in each of three winding classes reads
\begin{equation}
\begin{split}
\P_{\L\L}(12|34)&=1-\P_{\L\R}(12|34)-\P_{\R\L}(12|34),\\
\P_{\L\R}(12|34)&=\frac{G_{1,2}\,G'_{3,4}-G_{1,3}\,G'_{2,4}+G_{1,4}\,G'_{2,3}-G'_{1,2}\,G'_{3,4}+G'_{1,3}\,G'_{2,4}-G'_{1,4}\,G'_{2,3}}{G_{1,2}\,G_{3,4}-G_{1,3}\,G_{2,4}},\\
\P_{\R\L}(12|34)&=\frac{G_{1,4}\,G'_{2,3}-G_{2,4}\,G'_{1,3}+G_{3,4}\,G'_{1,2}-G'_{1,2}\,G'_{3,4}+G'_{1,3}G'_{2,4}-G'_{1,4}\,G'_{2,3}}{G_{1,2}\,G_{3,4}-G_{1,3}\,G_{2,4}}.
\end{split}
\label{2Schramm_1}
\end{equation}
The probabilities for the classes of paths in $14|23$, on the other hand, are given by
\begin{equation}
\begin{split}
\P_{\L\L}(14|23)&=1-\P_{\R\L}(14|23)-\P_{\R\R}(14|23),\\
\P_{\R\L}(14|23)&=\frac{G_{1,3}\,G'_{2,4}+G'_{1,3}\,G_{2,4}-G_{1,4}\,G'_{2,3}-G'_{1,4}\,G_{2,3}}{G_{1,3}\,G_{2,4}-G_{1,4}\,G_{2,3}}-2\,\P_{\R\R}(14|23),\\
\P_{\R\R}(14|23)&=\frac{-G'_{1,2}\,G'_{3,4}+G'_{1,3}\,G'_{2,4}-G'_{1,4}\,G'_{2,3}}{G_{1,3}\,G_{2,4}-G_{1,4}\,G_{2,3}}.
\label{2Schramm_2}
\end{split}
\end{equation}

\begin{figure}[h]
\centering
\begin{tikzpicture}[scale=0.6]

\begin{scope}
\draw[line width=0.75pt] (0,0) circle (2cm);
\draw[line width=1.75pt,blue](0:2)to[out=160,in=290](90:2);
\draw[line width=1.75pt,blue] (180:2)to[out=340,in=110](270:2);
\draw[fill=white] (0:2) circle (0.3cm) node[blue] {$1$};
\draw[fill=white] (90:2) circle (0.3cm) node[blue] {$2$};
\draw[fill=white] (180:2) circle (0.3cm) node[blue] {$3$};
\draw[fill=white] (270:2) circle (0.3cm) node[blue] {$4$};
\draw[very thick,red] (0,0)--(330:2);
\draw[very thick,red,-latex] (325:1.75)--(345:1.75);
\filldraw[blue!10!white] (-0.25,-0.25) rectangle node[blue] {$f$} (0.25,0.25);
\node at (0,-2.75) {LL};
\end{scope}

\begin{scope}[shift={(8cm,0)}]
\draw[line width=0.75pt] (0,0) circle (2cm);
\draw[line width=1.75pt,blue] (0:2)to[out=150,in=300](90:2);
\draw[line width=1.75pt,blue] (180:2)to[out=30,in=135](45:1)to[out=315,in=60](270:2);
\draw[fill=white] (0:2) circle (0.3cm) node[blue] {$1$};
\draw[fill=white] (90:2) circle (0.3cm) node[blue] {$2$};
\draw[fill=white] (180:2) circle (0.3cm) node[blue] {$3$};
\draw[fill=white] (270:2) circle (0.3cm) node[blue] {$4$};
\draw[very thick,red] (0,0)--(330:2);
\draw[very thick,red,-latex] (325:1.75)--(345:1.75);
\filldraw[blue!10!white] (-0.25,-0.25) rectangle node[blue] {$f$} (0.25,0.25);
\node at (0,-2.75) {LR};
\end{scope}

\begin{scope}[shift={(16cm,0)}]
\draw[line width=0.75pt] (0,0) circle (2cm);
\draw[line width=1.75pt,blue] (0:2)to[out=200,in=315](225:1)to[out=135,in=250](90:2);
\draw[line width=1.75pt,blue] (180:2)to[out=300,in=150](270:2);
\draw[fill=white] (0:2) circle (0.3cm) node[blue] {$1$};
\draw[fill=white] (90:2) circle (0.3cm) node[blue] {$2$};
\draw[fill=white] (180:2) circle (0.3cm) node[blue] {$3$};
\draw[fill=white] (270:2) circle (0.3cm) node[blue] {$4$};
\draw[very thick,red] (0,0)--(330:2);
\draw[very thick,red,-latex] (325:1.75)--(345:1.75);
\filldraw[blue!10!white] (-0.25,-0.25) rectangle node[blue] {$f$} (0.25,0.25);
\node at (0,-2.75) {RL};
\end{scope}

\begin{scope}[shift={(0,-6cm)}]
\draw[line width=0.75pt] (0,0) circle (2cm);
\draw[line width=1.75pt,blue] (0:2)to[out=100,in=45](135:1)to[out=225,in=170](270:2);
\draw[line width=1.75pt,blue] (90:2)to[out=210,in=60](180:2);
\draw[fill=white] (0:2) circle (0.3cm) node[blue] {$1$};
\draw[fill=white] (90:2) circle (0.3cm) node[blue] {$2$};
\draw[fill=white] (180:2) circle (0.3cm) node[blue] {$3$};
\draw[fill=white] (270:2) circle (0.3cm) node[blue] {$4$};
\draw[very thick,red] (0,0)--(330:2);
\draw[very thick,red,-latex] (325:1.75)--(345:1.75);
\filldraw[blue!10!white] (-0.25,-0.25) rectangle node[blue] {$f$} (0.25,0.25);
\node at (0,-2.75) {LL};
\end{scope}

\begin{scope}[shift={(8cm,-6cm)}]
\draw[line width=0.75pt] (0,0) circle (2cm);
\draw[line width=1.75pt,blue] (0:2)to[out=200,in=70](270:2);
\draw[line width=1.75pt,blue] (90:2)to[out=250,in=20](180:2);
\draw[fill=white] (0:2) circle (0.3cm) node[blue] {$1$};
\draw[fill=white] (90:2) circle (0.3cm) node[blue] {$2$};
\draw[fill=white] (180:2) circle (0.3cm) node[blue] {$3$};
\draw[fill=white] (270:2) circle (0.3cm) node[blue] {$4$};
\draw[very thick,red] (0,0)--(330:2);
\draw[very thick,red,-latex] (325:1.75)--(345:1.75);
\filldraw[blue!10!white] (-0.25,-0.25) rectangle node[blue] {$f$} (0.25,0.25);
\node at (0,-2.75) {RL};
\end{scope}

\begin{scope}[shift={(16cm,-6cm)}]
\draw[line width=0.75pt] (0,0) circle (2cm);
\draw[line width=1.75pt,blue] (0:2)to[out=210,in=60](270:2);
\draw[line width=1.75pt,blue] (90:2)to[out=340,in=45](315:1)to[out=225,in=290](180:2);
\draw[fill=white] (0:2) circle (0.3cm) node[blue] {$1$};
\draw[fill=white] (90:2) circle (0.3cm) node[blue] {$2$};
\draw[fill=white] (180:2) circle (0.3cm) node[blue] {$3$};
\draw[fill=white] (270:2) circle (0.3cm) node[blue] {$4$};
\draw[very thick,red] (0,0)--(330:2);
\draw[very thick,red,-latex] (325:1.75)--(345:1.75);
\filldraw[blue!10!white] (-0.25,-0.25) rectangle node[blue] {$f$} (0.25,0.25);
\node at (0,-2.75) {RR};
\end{scope}

\end{tikzpicture}
\caption{Illustration of the six different ways two paths between four nodes can wind around a marked face $f$. In the first (resp. second) row, the nodes are paired according to $12|34$ (resp. $14|23$). The indices $\L,\R$ indicate whether the first and second paths leave $f$ to their left or right (for the canonical orientation). The red line represents the zipper, whose arrow indicates the orientation of the edges with parallel transport $z$.}
\label{2Wind_cl}
\end{figure}
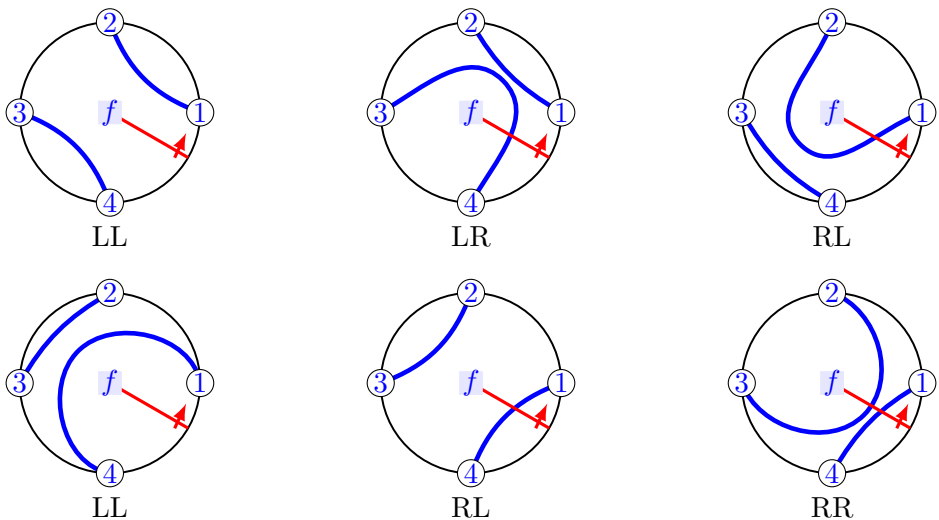


\subsection{More than two paths}
\label{sec5.2}

Let us now compute the winding probabilities for $n>2$ paths between $2n$ fixed boundary vertices in the same manner as for two paths using Theorem \ref{grove_thm}. To do so, we define a correspondence between winding classes and combinatorial objects called \emph{cyclic Dyck paths} \cite{KW15}, and give an explicit expression of the matrix $\mathbf{A}_n$ provided by Theorem \ref{grove_thm} as well as its inverse $\mathbf{A}_n^{-1}$. This will in turn allow us to write a combinatorial expression for winding partition functions of spanning forests, from which Schramm's formulas follow.

Let us first recall the definition of a standard\footnote{We use the adjective ``standard'' for such Dyck paths to distinguish them from cyclic Dyck paths, which will be introduced below.} Dyck path of length $2n$. It is a lattice path going from $(0,0)$ to $(2n,0)$ consisting in \emph{up} steps $(1,1)$ and \emph{down} steps $(1,-1)$ that never passes below the $x$ axis (such a path must therefore have an equal number of up and down steps). Equivalently, a Dyck path can be seen as a collection of heights $h_i\in\mathbb{N}$, $0\le i\le 2n$, with $|h_i{-}h_{i-1}|=1$ and $h_0=h_{2n}=0$; an up (resp. down) step at position $i$ corresponding to a pair of adjacent heights $(h_{i-1},h_i)$ such that $h_i{-}h_{i-1}=+1$ (resp. $-1$). In such a path, drawn as a mountain range as in Fig.~\ref{Dyck_ex}, each up step is paired with the closest down step to its right located at the same height; the pair being called a \emph{chord}. A third characterization of a Dyck path is given by the oriented pairing of the elements in $\mathcal{U}=\{1,2,\ldots,2n\}$ according to the chords below the path (from the up to the down step). The bijection between standard Dyck paths of length $2n$ and planar pairings of $2n$ nodes is well known (see e.g. \cite{Sta99}): each chord of a standard Dyck path pairs two labeled steps, thus defining a natural planar pairing of the nodes. This one-to-one correspondence has been used in particular to give a formula for all partition functions $Z[\sigma]$ of $(n{+}1)$-component spanning forests in terms of the Green function of the graph, where $\sigma$ is a planar pairing of boundary vertices \cite{KW11b} (see also Corollary~\ref{cor5.1} below).


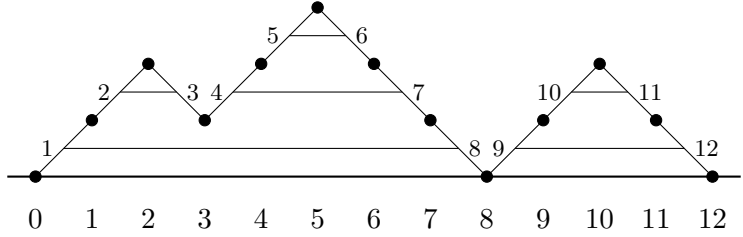
\begin{figure}[h]
\centering
\begin{tikzpicture}[scale=0.75]
\draw[thick] (-0.5,0)--(12.5,0);
\draw (0,0)--++(1,1)--++(1,1)--++(1,-1)--++(1,1)--++(1,1)--++(1,-1)--++(1,-1)--++(1,-1)--++(1,1)--++(1,1)--++(1,-1)--++(1,-1);
\foreach \p in {(0,0),(1,1),(2,2),(3,1),(4,2),(5,3),(6,2),(7,1),(8,0),(9,1),(10,2),(11,1),(12,0)}{\filldraw \p circle (0.1cm);}
\foreach \x in {0,...,12}{\node at (\x,-0.75) {$\x$};}
\draw (0.5,0.5) node[left,font=\footnotesize] {1}--(7.5,0.5) node[right,font=\footnotesize] {8};
\draw (1.5,1.5) node[left,font=\footnotesize] {2}--(2.5,1.5) node[right,font=\footnotesize] {3};
\draw (3.5,1.5) node[left,font=\footnotesize] {4}--(6.5,1.5) node[right,font=\footnotesize] {7};
\draw (4.5,2.5) node[left,font=\footnotesize] {5}--(5.5,2.5) node[right,font=\footnotesize] {6};
\draw (8.5,0.5) node[left,font=\footnotesize] {9}--(11.5,0.5) node[right,font=\footnotesize] {12};
\draw (9.5,1.5) node[left,font=\footnotesize] {10}--(10.5,1.5) node[right,font=\footnotesize] {11};
\end{tikzpicture}
\caption{Standard Dyck path $\sigma$ of length 12, which can be defined as a sequence of 6 up and 6 down steps, as a collection of 13 heights, or as a set of 6 chords: $\sigma=\text{UUDUUDDDUUDD}=\{0,1,2,1,2,3,2,1,0,1,2,1,0\}=\{(1,8),(2,3),(4,7),(5,6),(9,12),(10,11)\}$.}
\label{Dyck_ex}
\end{figure}

Since there are $(n{+}1)$ times more winding classes than standard Dyck paths, it is clear that the latter do not suffice to take into account the position of paths with respect to the face $f$ in each winding class. Following \cite{KW15}, we extend the definition of a Dyck path by allowing cyclic permutations of the labels associated with the steps, as well as of the labels of the heights (we choose as a convention the index 0 for the starting point of the step with the label 1). Such paths with numbered steps in cyclic order are called \emph{cyclic Dyck paths}\footnote{This definition differs slightly from the one of \cite{KW15}, in that we do not consider paths with a single ``flat'' step $(1,0)$ here.} (there are therefore $2n$ cyclic Dyck paths corresponding to a given mountain range). As a matter of notation, we shall denote such paths as vectors (e.g. $\vsig$), as opposed to standard Dyck paths, written as scalars (e.g. $\sigma$).

Consider now a particular winding class of a pairing of $\mathcal{U}$, with $k$ paths crossing the zipper. We choose the canonical orientation (from lower to higher node index) for the $n{-}k$ paths that do not intersect the zipper, and the opposite for the $k$ paths that do. We call this particular way of orienting paths a \emph{cyclic orientation}. The resulting oriented pairing $\vsig$ corresponds to the chords of a cyclic Dyck path, whose first step is labeled by the lowest index among the starting points of these $k$ paths. We provide an illustration of this correspondence in Fig.~\ref{cyc_Dyck_ex}. It should be noted that it is not surjective, as certain cyclic Dyck paths cannot be obtained in this manner. Namely, those that include a down step touching the $x$ axis to the left of the step with the label 1. Such paths are however equivalent to obtainable ones up to cyclic permutations of mountaintops together with their step and height labels (an example is given in the bottom panel of Fig.~\ref{cyc_Dyck_ex}).

The main benefit of choosing the cyclic orientation over the canonical one is notational. Indeed, the way the paths of a winding class wind around the face $f$ is directly encoded in a cyclic Dyck path $\vsig$ (or equivalently, in a cyclically oriented pairing $\vsig$). The corresponding winding partition functions in spanning forests and OCRGs will be denoted from now on by $Z[\vsig]$ and $\Zr[\vsig]$, respectively (dropping the cumbersome subscripts $\L,\R$ used in Sections \ref{sec3.2} and \ref{sec5.1}). In doing so, however, one finds a conflict of notation for standard Dyck paths $\sigma$, which correspond to winding classes in which no path crosses the zipper. Indeed, the notation $Z[\sigma]$ has until now been reserved for the total partition function of spanning forests in which the nodes are paired according to the unoriented pairing $\sigma$. To avoid any confusion in what follows, we shall denote from now on the winding partition function and the total partition by $Z[\sigma]$ and $Z_{\mathrm{t}}[\sigma]$, respectively.

Before we use this correspondence between winding classes and cyclic Dyck paths, let us introduce several definitions and notations associated with cyclic Dyck paths, some of which were given in \cite{KW15}:
\begin{defn}
Let $\vsig,\vt$ be two cyclic Dyck paths of length $2n$, and $R\subset\mathcal{U}=\{1,2,\ldots,2n\}$ a subset of indices of order $n$. Then
\begin{itemize}
\item[$\bullet$] $\mathcal{A}(\vsig)=\sum_{i=0}^{2n}h_i(\vsig)$ denotes the area between $\vsig$ and the $x$ axis, where the $h_i$'s are the heights of $\vsig$.
\item[$\bullet$] $W(\vsig)=h_0(\vsig)$ is the number of chords of $\vsig$ for which the index of the up step is higher than the one of the down step (in which case the chord is said to be wrapped).
\item[$\bullet$] $R\cap\vsig$ means that $R$ intersects each chord of $\vsig$ exactly once, that is, $R$ contains the index of one of the two steps belonging to each chord.
\item[$\bullet$] $R\cdot\vsig$ stands for the number of up steps of $\vsig$ whose indices lie in $R$.
\item[$\bullet$] $R:\vsig$ is the number of up steps of $\vsig$ with indices in $R$ that belong to wrapped chords.
\item[$\bullet$] $R|\vsig_{\l(1)}$ counts the steps of $\vsig$ that appear to the left of the step 1 (excluded) and whose indices appear in $R$.
\item[$\bullet$] $U^{\vsig}$ (resp. $D^{\vsig}$) denotes the set of up (resp. down) steps of $\vsig$.
\item[$\bullet$] $U^{\vsig}_{r(1)}$ is the set of up steps of $\vsig$ that appear to the right of the step 1 (included).
\item[$\bullet$] $D^{\vsig}_{\l(1)}$ denotes the set of down steps of $\vsig$ that appear to the left of the step 1 (excluded).
\item[$\bullet$] $D^{\vt}_{\l(1,\vsig)}$ is the set of down steps of $\vt$ whose indices appear to the left of the step 1 in $\vsig$ (excluded).
\end{itemize}
\vspace{-0.4cm}
\label{def5.1}
\end{defn}
As an example, let us consider the cyclic Dyck path illustrated in the top right panel of Fig.~\ref{cyc_Dyck_ex}, given by $\vsig=\{(5,2),(6,7),(8,1),(3,4)\}=\{2,1,0,1,0,0,1,2,1\}$, which yields $W(\vsig)=2$ (due to the chords $(5,2)$ and $(8,1)$). If $R=\{1,4,5,6\}$, then $R\cap\vsig$, $R\cdot\vsig=|\{5,6\}|=2$, $R:\vsig=|\{5\}|=1$ and $R|\vsig_{\l(1)}=|\{5,6\}|=2$.

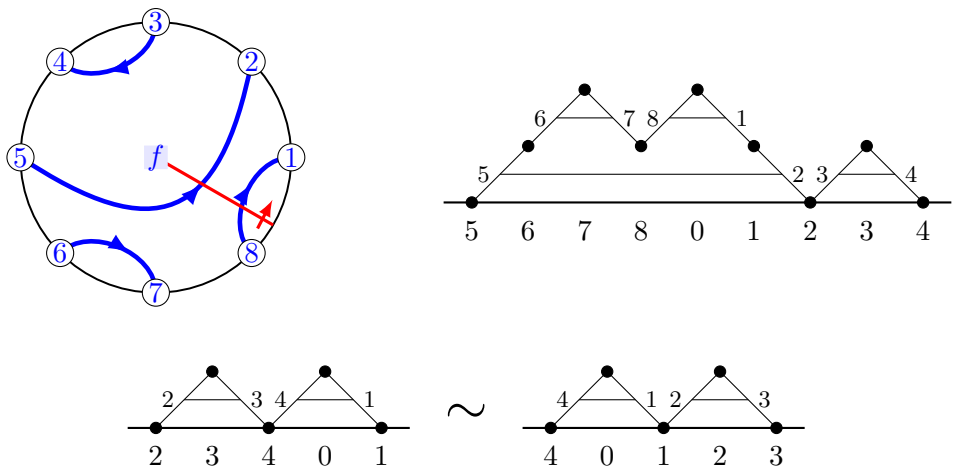
\begin{figure}[h]
\centering
\begin{tikzpicture}[scale=0.6]
\tikzset{posar/.style={decoration={markings,mark=at position 0.6 with {\arrow{#1}},},postaction={decorate}}}

\begin{scope}
\draw[line width=0.75pt] (0,0) circle (3cm);
\draw[line width=1.75pt,posar={latex},blue] (180:3)..controls(270:2)and(315:2)..(45:3);
\draw[line width=1.75pt,posar={latex},blue] (225:3)to[out=45,in=90](270:3);
\draw[line width=1.75pt,posar={latex},blue] (315:3)to[out=135,in=180](0:3);
\draw[line width=1.75pt,posar={latex},blue] (90:3)to[out=270,in=315](135:3);
\foreach \n in {1,...,8}{\draw[fill=white] ({(\n-1)*360/8}:3) circle (0.3cm) node[blue] {$\n$};}
\draw[very thick,red] (0,0)--(330:3);
\draw[very thick,red,-latex] (325:2.75)--(340:2.75);
\filldraw[blue!10!white] (-0.25,-0.25) rectangle node[blue] {$f$} (0.25,0.25);
\end{scope}

\begin{scope}[shift={(7cm,-1cm)},scale=1.25]
\draw[thick] (-0.5,0)--(8.5,0);
\draw (0,0)--++(1,1)--++(1,1)--++(1,-1)--++(1,1)--++(1,-1)--++(1,-1)--++(1,1)--++(1,-1);
\foreach \p in {(0,0),(1,1),(2,2),(3,1),(4,2),(5,1),(6,0),(7,1),(8,0)}{\filldraw \p circle (0.1cm);}
\foreach \x in {0,...,8}{\node at ({Mod(-5+\x,9)},-0.5) {$\x$};}
\draw (0.5,0.5) node[left,font=\footnotesize] {5}--(5.5,0.5) node[right,font=\footnotesize] {2};
\draw (1.5,1.5) node[left,font=\footnotesize] {6}--(2.5,1.5) node[right,font=\footnotesize] {7};
\draw (3.5,1.5) node[left,font=\footnotesize] {8}--(4.5,1.5) node[right,font=\footnotesize] {1};
\draw (6.5,0.5) node[left,font=\footnotesize] {3}--(7.5,0.5) node[right,font=\footnotesize] {4};
\end{scope}

\begin{scope}[shift={(0,-6cm)},scale=1.25]
\draw[thick] (-0.5,0)--(4.5,0);
\draw (0,0)--++(1,1)--++(1,-1)--++(1,1)--++(1,-1);
\foreach \p in {(0,0),(1,1),(2,0),(3,1),(4,0)}{\filldraw \p circle (0.1cm);}
\draw (0.5,0.5) node[left,font=\footnotesize] {2}--(1.5,0.5) node[right,font=\footnotesize] {3};
\draw (2.5,0.5) node[left,font=\footnotesize] {4}--(3.5,0.5) node[right,font=\footnotesize] {1};
\foreach \x in {0,...,4}{\node at ({Mod(-2+\x,5)},-0.5) {$\x$};}
\node at (5.5,0.25) {\huge $\sim$};
\draw[thick] (6.5,0)--(11.5,0);
\draw (7,0)--++(1,1)--++(1,-1)--++(1,1)--++(1,-1);
\foreach \p in {(7,0),(8,1),(9,0),(10,1),(11,0)}{\filldraw \p circle (0.1cm);}
\draw (7.5,0.5) node[left,font=\footnotesize] {4}--(8.5,0.5) node[right,font=\footnotesize] {1};
\draw (9.5,0.5) node[left,font=\footnotesize] {2}--(10.5,0.5) node[right,font=\footnotesize] {3};
\foreach \x in {0,...,4}{\node at ({7+Mod(-4+\x,5)},-0.5) {$\x$};}
\end{scope}

\end{tikzpicture}
\caption{On the top left is a schematic representation of a winding class of quadruples of paths in spanning forests. The paths $5\to 2$ and $8\to 1$ cross the zipper, and are therefore oriented from the higher to the lower index, as opposed to the paths $3\to 4$ and $6\to 7$. On the top right is the corresponding cyclic Dyck path, whose first step label is $\min\{5,8\}=5$. At the bottom lies a cyclic Dyck path that cannot be realized as a winding class of paths. The path obtained by cyclically rotating its two mountaintops, however, corresponds to the winding class $\R\L$ of the pairing $14|23$, pictured in Fig.~\ref{2Wind_cl}.}
\label{cyc_Dyck_ex}
\end{figure}

Let us now show how to write linear relations between winding partition functions using Theorem \ref{grove_thm}. There are $C_n=\frac{1}{n+1}\left({2n\atop n}\right)$ planar pairings of $2n$ nodes on the boundary of a planar graph, each of which containing $n{+}1$ winding classes of paths with respect to the marked face $f$; thus yielding a total of $\left({2n\atop n}\right)$ winding classes, and as many inequivalent cyclic Dyck paths. As there are as many distinct subsets $R\subset\mathcal{U}$ of order $n$, the linear equations provided by Theorem~\ref{grove_thm} relating winding partition functions to minors of the line bundle Green function $\Gr$ can be encoded in the square matrix $\mathbf{A}_n$ defined by
\begin{equation}
\det\mathbf{\Delta}\det\Gr_R^S=\sum_{\vsig}\left(\mathbf{A}_n\right)_{R,\vsig}\Zr[\vsig],
\label{grove_eq}
\end{equation}
where the sum is over all cyclic Dyck paths of length $2n$. The entries of this matrix are given by
\begin{prop}
Let $\vsig$ be a pairing of $2n$ nodes with the cyclic orientation along each path, seen as a cyclic Dyck path of length $2n$. Let $R\subset\mathcal{U}=\{1,2,\ldots,2n\}$ with $|R|=n$. If we define the new variable $w \equiv z^2$, then
\begin{equation}
\big(\mathbf{A}_n\big)_{R,\vsig}=\begin{cases}
(-1)^{\frac{1}{2}(\mathcal{A}(\vsig)-n)+(n+1)W(\vsig)}\,w^{W(\vsig)-R:\vsig}&\quad\text{if $R\cap\vsig$},\\
0&\quad\text{if $R\,\,\slashed{\cap}\,\,\vsig$}.
\end{cases}
\end{equation}
\label{prop5.1}
\end{prop}
\vspace{-0.75cm}
As an example, consider for instance the cyclic Dyck path depicted in the top right panel of Fig.~\ref{cyc_Dyck_ex}, $\vsig=\{(5,2),(6,7),(8,1),(3,4)\}$, and the subset of indices $R=\{1,4,5,6\}$. As explained above, $R$ contains an index of each chord of $\vsig$ (i.e. $R\cap\vsig$), and $W(\vsig)=2$, $R:\vsig=1$. Moreover, the area between $\vsig$ and the horizontal axis is $\mathcal{A}(\vsig)=8$, so
\begin{equation}
(\mathbf{A}_4)_{R,\vsig}=(-1)^{\frac{1}{2}(8-4)+(4+1)2}\,w^{2-1}=w.
\end{equation}

The proof of Proposition~\ref{prop5.1} is a bit technical, and is left to Appendix~\ref{a2}. A similar matrix was considered in \cite{KW15}, in which all but one node are located around a single marked face, the remaining one lying on the boundary of the outer face. The authors gave a formula for the inverse in terms of combinatorial objects called \emph{cover-inclusive Dyck tilings}, introduced in \cite{KW11b,SZJ12}. We recall their definition here, as they shall be needed to write the inverse of $\mathbf{A}_n$.

A \emph{Dyck tile} is obtained from a Dyck path by replacing each vertex of the path with a $\sqrt{2}{\times}\sqrt{2}$ square rotated by $45^\circ$, and by gluing all the squares together to form a single ribbon (see for instance the first two panels at the top of Fig.~\ref{Dyck_tiling_ex}). In particular, the simplest Dyck tile is a single square, associated with the degenerate Dyck path of length zero.

Let $\vmu,\vsig$ be two cyclic Dyck paths of length $2n$ with the same positions of indices, and let $h_i(\vmu),h_i(\vsig)$ be the heights of their vertices. The cyclic Dyck path $\vmu$ is said to dominate $\vsig$ if $h_i(\vmu)\ge h_i(\vsig)$ for $0\le i\le 2n$; which is denoted by $\vmu\ge\vsig$. A \emph{Dyck tiling} of the shape $\vsig/\vmu$ between two paths $\vmu,\vsig$ such that $\vmu\ge\vsig$ is then obtained by filling the surface between $\vmu$ and $\vsig$ with Dyck tiles. Furthermore, a Dyck tiling is said to be \emph{cover-inclusive} if, for any two tiles right above one another, the top one does not stick out horizontally with respect to the bottom one. We illustrate all cover-inclusive Dyck tilings of a given shape in the bottom panel of Fig.~\ref{Dyck_tiling_ex}.

\begin{figure}[h]
\centering
\begin{tikzpicture}[scale=0.5]

\begin{scope}
\draw[thick] (-0.5,0)--(8.5,0);
\draw (0,0)--++(1,1)--++(1,1)--++(1,-1)--++(1,1)--++(1,-1)--++(1,-1)--++(1,1)--++(1,-1);
\foreach \p in {(0,0),(1,1),(2,2),(3,1),(4,2),(5,1),(6,0),(7,1),(8,0)}{\filldraw \p circle (0.1cm);}
\end{scope}

\begin{scope}[scale=0.6,shift={(16cm,0)}]
\draw[fill=Rhodamine!70!white] (0,0)--(3,3)--(4,2)--(6,4)--(9,1)--(10,2)--(12,0)--(11,-1)--(10,0)--(9,-1)--(6,2)--(4,0)--(3,1)--(1,-1)--(0,0);
\draw[densely dotted,thick] (2,0)--(1,1);
\draw[densely dotted,thick] (3,1)--(2,2);
\draw[densely dotted,thick] (3,1)--(4,2);
\draw[densely dotted,thick] (5,1)--(4,2);
\draw[densely dotted,thick] (6,2)--(5,3);
\draw[densely dotted,thick] (6,2)--(7,3);
\draw[densely dotted,thick] (7,1)--(8,2);
\draw[densely dotted,thick] (8,0)--(9,1);
\draw[densely dotted,thick] (10,0)--(9,1);
\draw[densely dotted,thick] (10,0)--(11,1);
\end{scope}

\begin{scope}[scale=0.5,shift={(40cm,2cm)}]
\draw (-3.5,-3)--(11.5,-3);
\draw (0,0)--(4,4)--(9,-1)--(8,-2)--(7,-1)--(6,-2)--(4,0)--(2,-2)--(0,0);
\filldraw[Rhodamine!70!white] (-3,-3)--(4,4)--(11,-3)--(-3,-3);
\draw (1,-1)--(2,0);
\draw (2,0)--(3,1);
\draw (3,1)--(4,2);
\draw (4,0)--(5,1);
\draw (5,1)--(6,2);
\draw (5,-1)--(6,0);
\draw (7,-1)--(8,0);
\draw (3,-1)--(2,0);
\draw (2,0)--(1,1);
\draw (4,0)--(3,1);
\draw (3,1)--(2,2);
\draw (7,-1)--(6,0);
\draw (5,1)--(4,2);
\filldraw[white] (-3,-3)--(0,0)--(2,-2)--(4,0)--(6,-2)--(7,-1)--(8,-2)--(9,-1)--(11,-3)--(-3,-3);
\draw[thick] (-3,-3)--(4,4)--(11,-3);
\draw[thick] (0,0)--(2,-2)--(4,0)--(6,-2)--(7,-1)--(8,-2)--(9,-1);
\foreach \p in {(-3,-3),(-2,-2),(-1,-1),(0,0),(1,1),(2,2),(3,3),(4,4),(5,3),(6,2),(7,1),(8,0),(9,-1),(10,-2),(11,-3),(1,-1),(2,-2),(3,-1),(4,0),(5,-1),(6,-2),(7,-1),(8,-2)}{\filldraw \p circle (0.175cm);}
\node at (2,4) {$\vmu$};
\node at (4,-2) {$\vsig$};
\end{scope}

\begin{scope}[scale=0.45,shift={(0,-8cm)},thick]
\draw[fill=Rhodamine!70!white] (0,0)--(4,4)--(9,-1)--(8,-2)--(7,-1)--(6,-2)--(4,0)--(2,-2)--(0,0);
\draw (1,-1)--(2,0);
\draw (2,0)--(3,1);
\draw (3,1)--(4,2);
\draw (4,2)--(5,3);
\draw (4,0)--(5,1);
\draw (5,1)--(6,2);
\draw (5,-1)--(6,0);
\draw (6,0)--(7,1);
\draw (7,-1)--(8,0);
\draw (3,-1)--(2,0);
\draw (2,0)--(1,1);
\draw (4,0)--(3,1);
\draw (3,1)--(2,2);
\draw (7,-1)--(6,0);
\draw (6,0)--(5,1);
\draw (5,1)--(4,2);
\draw (4,2)--(3,3);
\end{scope}

\begin{scope}[scale=0.45,shift={(10cm,-8cm)},thick]
\draw[fill=Rhodamine!70!white] (0,0)--(4,4)--(9,-1)--(8,-2)--(7,-1)--(6,-2)--(4,0)--(2,-2)--(0,0);
\draw (1,-1)--(2,0);
\draw (2,0)--(3,1);
\draw (3,1)--(4,2);
\draw (4,2)--(5,3);
\draw (5,1)--(6,2);
\draw (5,-1)--(6,0);
\draw (6,0)--(7,1);
\draw (7,-1)--(8,0);
\draw (3,-1)--(2,0);
\draw (2,0)--(1,1);
\draw (3,1)--(2,2);
\draw (7,-1)--(6,0);
\draw (6,0)--(5,1);
\draw (5,1)--(4,2);
\draw (4,2)--(3,3);
\end{scope}

\begin{scope}[scale=0.45,shift={(20cm,-8cm)},thick]
\draw[fill=Rhodamine!70!white] (0,0)--(4,4)--(9,-1)--(8,-2)--(7,-1)--(6,-2)--(4,0)--(2,-2)--(0,0);
\draw (1,-1)--(2,0);
\draw (2,0)--(3,1);
\draw (3,1)--(4,2);
\draw (5,1)--(6,2);
\draw (5,-1)--(6,0);
\draw (6,0)--(7,1);
\draw (7,-1)--(8,0);
\draw (3,-1)--(2,0);
\draw (2,0)--(1,1);
\draw (3,1)--(2,2);
\draw (7,-1)--(6,0);
\draw (6,0)--(5,1);
\draw (5,1)--(4,2);
\end{scope}

\begin{scope}[scale=0.45,shift={(30cm,-8cm)},thick]
\draw[fill=Rhodamine!70!white] (0,0)--(4,4)--(9,-1)--(8,-2)--(7,-1)--(6,-2)--(4,0)--(2,-2)--(0,0);
\draw (1,-1)--(2,0);
\draw (2,0)--(3,1);
\draw (3,1)--(4,2);
\draw (4,2)--(5,3);
\draw (4,0)--(5,1);
\draw (5,1)--(6,2);
\draw (5,-1)--(6,0);
\draw (6,0)--(7,1);
\draw (3,-1)--(2,0);
\draw (2,0)--(1,1);
\draw (4,0)--(3,1);
\draw (3,1)--(2,2);
\draw (6,0)--(5,1);
\draw (5,1)--(4,2);
\draw (4,2)--(3,3);
\end{scope}

\begin{scope}[scale=0.45,shift={(40cm,-8cm)},thick]
\draw[fill=Rhodamine!70!white] (0,0)--(4,4)--(9,-1)--(8,-2)--(7,-1)--(6,-2)--(4,0)--(2,-2)--(0,0);
\draw (1,-1)--(2,0);
\draw (2,0)--(3,1);
\draw (3,1)--(4,2);
\draw (4,2)--(5,3);
\draw (5,1)--(6,2);
\draw (5,-1)--(6,0);
\draw (6,0)--(7,1);
\draw (3,-1)--(2,0);
\draw (2,0)--(1,1);
\draw (3,1)--(2,2);
\draw (6,0)--(5,1);
\draw (5,1)--(4,2);
\draw (4,2)--(3,3);
\end{scope}

\begin{scope}[scale=0.45,shift={(50cm,-8cm)},thick]
\draw[fill=Rhodamine!70!white] (0,0)--(4,4)--(9,-1)--(8,-2)--(7,-1)--(6,-2)--(4,0)--(2,-2)--(0,0);
\draw (1,-1)--(2,0);
\draw (2,0)--(3,1);
\draw (3,1)--(4,2);
\draw (5,1)--(6,2);
\draw (5,-1)--(6,0);
\draw (6,0)--(7,1);
\draw (3,-1)--(2,0);
\draw (2,0)--(1,1);
\draw (3,1)--(2,2);
\draw (6,0)--(5,1);
\draw (5,1)--(4,2);
\end{scope}

\begin{scope}[scale=0.45,shift={(0,-16cm)},thick]
\draw[fill=Rhodamine!70!white] (0,0)--(4,4)--(9,-1)--(8,-2)--(7,-1)--(6,-2)--(4,0)--(2,-2)--(0,0);
\draw (1,-1)--(2,0);
\draw (2,0)--(3,1);
\draw (3,1)--(4,2);
\draw (4,2)--(5,3);
\draw (5,1)--(6,2);
\draw (6,0)--(7,1);
\draw (7,-1)--(8,0);
\draw (2,0)--(1,1);
\draw (3,1)--(2,2);
\draw (7,-1)--(6,0);
\draw (6,0)--(5,1);
\draw (5,1)--(4,2);
\draw (4,2)--(3,3);
\end{scope}

\begin{scope}[scale=0.45,shift={(10cm,-16cm)},thick]
\draw[fill=Rhodamine!70!white] (0,0)--(4,4)--(9,-1)--(8,-2)--(7,-1)--(6,-2)--(4,0)--(2,-2)--(0,0);
\draw (1,-1)--(2,0);
\draw (2,0)--(3,1);
\draw (3,1)--(4,2);
\draw (5,1)--(6,2);
\draw (6,0)--(7,1);
\draw (7,-1)--(8,0);
\draw (2,0)--(1,1);
\draw (3,1)--(2,2);
\draw (7,-1)--(6,0);
\draw (6,0)--(5,1);
\draw (5,1)--(4,2);
\end{scope}

\begin{scope}[scale=0.45,shift={(20cm,-16cm)},thick]
\draw[fill=Rhodamine!70!white] (0,0)--(4,4)--(9,-1)--(8,-2)--(7,-1)--(6,-2)--(4,0)--(2,-2)--(0,0);
\draw (1,-1)--(2,0);
\draw (2,0)--(3,1);
\draw (3,1)--(4,2);
\draw (6,0)--(7,1);
\draw (7,-1)--(8,0);
\draw (2,0)--(1,1);
\draw (7,-1)--(6,0);
\draw (6,0)--(5,1);
\draw (5,1)--(4,2);
\end{scope}

\begin{scope}[scale=0.45,shift={(30cm,-16cm)},thick]
\draw[fill=Rhodamine!70!white] (0,0)--(4,4)--(9,-1)--(8,-2)--(7,-1)--(6,-2)--(4,0)--(2,-2)--(0,0);
\draw (1,-1)--(2,0);
\draw (2,0)--(3,1);
\draw (3,1)--(4,2);
\draw (4,2)--(5,3);
\draw (5,1)--(6,2);
\draw (6,0)--(7,1);
\draw (2,0)--(1,1);
\draw (3,1)--(2,2);
\draw (6,0)--(5,1);
\draw (5,1)--(4,2);
\draw (4,2)--(3,3);
\end{scope}

\begin{scope}[scale=0.45,shift={(40cm,-16cm)},thick]
\draw[fill=Rhodamine!70!white] (0,0)--(4,4)--(9,-1)--(8,-2)--(7,-1)--(6,-2)--(4,0)--(2,-2)--(0,0);
\draw (1,-1)--(2,0);
\draw (2,0)--(3,1);
\draw (3,1)--(4,2);
\draw (5,1)--(6,2);
\draw (6,0)--(7,1);
\draw (2,0)--(1,1);
\draw (3,1)--(2,2);
\draw (6,0)--(5,1);
\draw (5,1)--(4,2);
\end{scope}

\begin{scope}[scale=0.45,shift={(50cm,-16cm)},thick]
\draw[fill=Rhodamine!70!white] (0,0)--(4,4)--(9,-1)--(8,-2)--(7,-1)--(6,-2)--(4,0)--(2,-2)--(0,0);
\draw (1,-1)--(2,0);
\draw (2,0)--(3,1);
\draw (3,1)--(4,2);
\draw (6,0)--(7,1);
\draw (2,0)--(1,1);
\draw (6,0)--(5,1);
\draw (5,1)--(4,2);
\end{scope}

\end{tikzpicture}
\caption{On the left is a Dyck path, together with the associated Dyck tile (in the middle). On the right is a Dyck tiling of the shape $\vsig/\vmu$ between the paths $\vmu$ and $\vsig$, which is \emph{not} cover-inclusive (the two L shapes both cover a single square). At the bottom lie all cover-inclusive Dyck tilings of $\vsig/\vmu$.}
\label{Dyck_tiling_ex}
\end{figure}
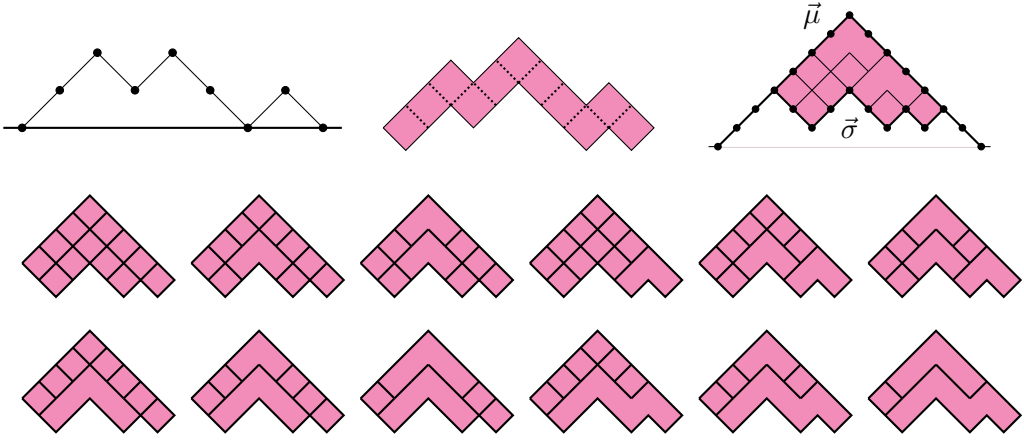

Let us now give a formula for the inverse of the matrix $\mathbf{A}_n$ defined in Proposition~\ref{prop5.1}. In terms of the quantities introduced in Definition~\ref{def5.1}, we have
\begin{prop} Let $\mathbf{B}_n$ be the square matrix defined by
\begin{equation}
\left(\mathbf{B}_n\right)_{\vsig,R}=(-1)^{\Sigma R+W(\vsig)+n}\sum_{\vmu\ge\vsig}\mathrm{ci}(\vsig/\vmu)\,w^{n-W(\vsig)-R\cdot\vmu+R|\vsig_{\l(1)}-\big|D^{\vsig}_{\l(1)}\big|},
\end{equation}
for all cyclic Dyck paths $\vsig$ of length $2n$ and for all subsets $R\subset\mathcal{U}$ of order $n$, where $\mathrm{ci}(\vsig/\vmu)$ denotes the number of cover-inclusive Dyck tilings of the shape between $\vmu$ and $\vsig$. Then for any two cyclic Dyck paths $\vsig,\vt$ of length $2n$
\begin{equation}
\sum_{R\subset\mathcal{U}:\,|R|=n}\left(\mathbf{B}_n\right)_{\vsig,R}\left(\mathbf{A}_n\right)_{R, \vt}=(1-w)^n\delta_{\vsig,\vt}.
\label{prop5.2eq1}
\end{equation}
\label{prop5.2}
\end{prop}
\vspace{-0.75cm}
Again we refer to Appendix~\ref{a2} for the proof. This proposition allows one to write winding partition functions for paths in OCRGs as
\begin{equation}
\Zr[\vsig]=\frac{\det\mathbf{\Delta}}{(1-w)^n}\sum_{R\subset\mathcal{U}:\,|R|=n}\left(\mathbf{B}_n\right)_{\vsig,R}\det\Gr_R^S,
\label{Wind_inv}
\end{equation}
where $S=\mathcal{U}\bs R$ and $\Gr=\mathbf{\Delta}^{-1}$ is the line bundle Green function of the graph. Corresponding partition functions $Z[\vsig]$ for paths in spanning forests can be obtained by taking the limit $z\to 1$ (or equivalently $w\to 1$) of $\Zr[\vsig]$. The result is given by
\begin{thm} Let $\vsig$ be a cyclic Dyck path of length $2n$. Then
\begin{equation}
Z[\vsig]=(-1)^{W(\vsig)}\det\Delta\sum_{\vmu\ge\vsig}\mathrm{ci}(\vsig/\vmu)\,\pf\mathcal{M}_{\vmu},
\end{equation}
where the matrix $\mathcal{M}_{\vmu}$ is defined as follows for $1\le i,j\le 2n$:
\begin{equation}
\left(\mathcal{M}_{\vmu}\right)_{i,j}=-G'_{i,j}+G_{i,j}\Big(\mathbf{1}_{i\in U^{\vmu}_{r(1)}}-\mathbf{1}_{i\in D^{\vmu}_{\l(1)}}-\mathbf{1}_{j\in U^{\vmu}_{r(1)}}+\mathbf{1}_{j\in D^{\vmu}_{\l(1)}}\Big).
\end{equation}
Here $\mathbf{1}$ is the indicator function and $U^{\vsig}_{\l,r(1)},D^{\vsig}_{\l,r(1)}$ are given in Definition~\ref{def5.1}.
\label{thm5.1}
\end{thm}
\begin{proof}
Taking the limit $w\to 1$ of Eq.~\eqref{Wind_inv} leads to
\begin{equation}
Z[\vsig]=(-1)^{W(\vsig)}\det\Delta\sum_{\vmu\ge\vsig}\mathrm{ci}(\vsig/\vmu)\lim_{w\to 1}\frac{w^{n-W(\vsig)-\left|D^{\vsig}_{\l(1)}\right|}}{(1-w)^n}\sum_{\substack{R\subset\mathcal{U}\\|R|=n}}(-1)^{n+\Sigma R}w^{-R\cdot\vmu+R|\vsig_{\l(1)}}\det\Gr_R^S,
\label{thm5.1eq1}
\end{equation}
where we used the explicit expression of $\mathbf{B}_n$ given by Proposition~\ref{prop5.2}. The first sum is over cyclic Dyck paths $\vmu\ge\vsig$, so $R|\vsig_{\l(1)}=R|\vmu_{\l(1)}$, as the step indices of $\vmu,\vsig$ appear in the same order. Moreover
\begin{equation}
-R\cdot\vmu+R|\vmu_{\l(1)}=\left|R\cap D^{\vmu}_{\l(1)}\right|-\left|R\cap U^{\vmu}_{r(1)}\right|,
\end{equation}
since $R\cdot\vmu$ counts the up steps of $\vmu$ in $R$ and $R|\vmu_{\l(1)}$ the steps of $\vmu$ to the left of the step 1 that belong to $R$. Substituting this relation in Eq.~\eqref{thm5.1eq1} allows one to compute the sum  over $R$ via the following
\begin{lm}[\hspace{-0.025cm}\cite{PW17a}, Lemma 3.4\hspace{-0.025cm}]
For any matrix $A$ of order $2n$ and any two subsets $X,Y\subset\mathcal{U}$ with $X\cap Y=\emptyset$,
\begin{equation}
\sum_{\substack{R\subset\mathcal{U}\\|R|=n}}(-1)^{n+\Sigma R}x^{2\left(|R\cap X|-|R\cap Y|\right)}\det A_R^S=x^{|X|-|Y|}\pf\left(\widetilde{A}-\widetilde{A}^{\,\mathrm{t}}\right),
\end{equation}
where $S=\mathcal{U}\bs R$ and $R$ are both ordered, and the matrix $\widetilde{A}$ is defined as follows:
\begin{equation}
\widetilde{A}_{i,j}=A_{i,j}\,x^{\mathbf{1}_{i\in X}+\mathbf{1}_{j\in Y}-\mathbf{1}_{i\in Y}-\mathbf{1}_{j\in X}}.
\end{equation}
\label{lm5.5}
\end{lm}
\vspace{-0.5cm}
Picking $A=\Gr$, $w=x^{1/2}$, $X=D^{\vmu}_{\l(1)}$ and $Y=U^{\vmu}_{r(1)}$ in the lemma allows one to rewrite Eq.~\eqref{thm5.1eq1} as
\begin{equation}
\begin{split}
Z[\vsig]&=(-1)^{W(\vsig)}\det\Delta\sum_{\vmu\ge\vsig}\mathrm{ci}(\vsig/\vmu)\lim_{w\to 1}w^{n-W(\vsig)-\left|D^{\vsig}_{\l(1)}\right|+\frac{1}{2}\left|D^{\vmu}_{\l(1)}\right|-\frac{1}{2}\left|U^{\vmu}_{r(1)}\right|}\,\pf\left(\frac{\widetilde{\Gr}-\widetilde{\Gr}^{\,\mathrm{t}}}{1-w}\right)\\
&=(-1)^{W(\vsig)}\det\Delta\sum_{\vmu\ge\vsig}\mathrm{ci}(\vsig/\vmu)\lim_{w\to 1}\pf\left(\frac{\widetilde{\Gr}-\widetilde{\Gr}^{\,\mathrm{t}}}{1-w}\right),
\end{split}
\end{equation}
where the derivative of the matrix $\widetilde{\Gr}$ reads
\begin{equation}
\begin{split}
\partial_w\widetilde{\Gr}_{i,j}\big|_{w=1}&=\partial_w\Gr_{i,j}\big|_{w=1}-\frac{1}{2}\Gr_{i,j}\big|_{w=1}\times\left\{\mathbf{1}_{i\in U^{\vmu}_{r(1)}}-\mathbf{1}_{i\in D^{\vmu}_{\l(1)}}-\mathbf{1}_{j\in U^{\vmu}_{r(1)}}+\mathbf{1}_{j\in D^{\vmu}_{\l(1)}}\right\}\\
&=-\frac{1}{2}\left(\mathcal{M}_{\vmu}\right)_{i,j}.
\end{split}
\end{equation}
The result of the preceding equation relies on the definition of $G,G'$ as the zeroth- and first-order derivative of the line bundle Green function $\Gr$ evaluated at $z=1$:
\begin{equation}
G=\lim_{z\to 1}\Gr=\lim_{w\to 1}\Gr,\quad G'=\lim_{z\to 1}\partial_z\Gr=\lim_{w\to 1}2\,\partial_w\Gr
\end{equation}
for $w=z^2$. Similarly we find $\partial_w\widetilde{\Gr}^{\,\mathrm{t}}_{i,j}\big|_{w=1}=\frac{1}{2}\left(\mathcal{M}_{\vmu}\right)_{i,j}$, which concludes the proof of Theorem~\ref{thm5.1}.
\end{proof}

A particular case of Theorem~\ref{thm5.1} arises for cyclic Dyck paths $\vsig$ of length $2n$ that correspond to a single mountaintop of maximal height $n$ at its peak. Indeed, the sum over cyclic Dyck paths $\vmu$ that dominate $\vsig$ includes only a single term, namely $\vmu=\vsig$. Consequently the winding partition function $Z[\vsig]$ is given (up to a multiplicative constant) by a single Pfaffian, $\pf\mathcal{M}_{\vsig}$. Among such cyclic Dyck paths is the one defined by $\vsig_{\R}=\{(2n,1),(2n{-}1,2),\ldots,(n{+}1,n)\}$, which corresponds to $n$-tuples of paths that all leave the marked face $f$ to their right when canonically oriented (from lower to higher node index). It turns out that $Z[\vsig_{\R}]$ has an exceptionally simple form, as all of its up steps appear to the left of the step 1. Indeed, this implies that $U^{\vsig_{\R}}_{r(1)}=\varnothing=D^{\vsig_{\R}}_{\l(1)}$, so the matrix $\mathcal{M}_{\vsig_{\R}}$ is equal to $-G'|_{\mathcal{U}}$, and $W(\vsig_{\R})=n$ (all chords have an up step with a higher index than the corresponding down step). The theorem therefore yields
\begin{equation}
Z[\vsig_{\R}]=\det\Delta\times\pf\left(G'|_{\mathcal{U}}\right).
\end{equation}

As an immediate corollary of Theorem~\ref{thm5.1}, we may write an explicit expression for the \emph{total} partition function $Z_{\mathrm{t}}[\sigma]$ of all paths between $2n$ nodes paired according to the pairing $\sigma$ (which we also see as a standard Dyck path), corresponding to the sum of the $n{+}1$ winding partition functions for the pairing $\sigma$. The result is as follows:
\begin{cor}[\hspace{-0.025cm}\cite{KW11b}, Theorem 3.2\hspace{-0.025cm}]
Let $\sigma$ be a standard Dyck path of length $2n$. Then
\begin{equation}
Z_{\mathrm{t}}[\sigma]=\det\Delta\sum_{\mu\ge\sigma}(-1)^{\frac{1}{2}\left(\mathcal{A}(\mu)-n\right)}\mathrm{ci}(\sigma/\mu)\det G_{U^{\mu}}^{D^{\mu}},
\end{equation}
where the sets of up and down steps of $\mu$, $U^{\mu}$ and $D^{\mu}$, are sorted in ascending order.
\label{cor5.1}
\end{cor}
\begin{proof} To compute the total partition function indexed by $\sigma$, it suffices to take the reference face $f$ as the outer face of the graph. This amounts to setting $G'=0$, so that Theorem~\ref{thm5.1} yields
\begin{equation}
Z_{\mathrm{t}}[\sigma]=\det\Delta\sum_{\mu\ge\sigma}\mathrm{ci}(\lambda/\mu)\,\pf\mathcal{M}_{\mu},
\end{equation}
where $\left(\mathcal{M}_{\mu}\right)_{i,j}=G_{i,j}\left(\mathbf{1}_{i\in U^{\mu}}-\mathbf{1}_{j\in U^{\mu}}\right)$. In other words, $\left(\mathcal{M}_{\mu}\right)_{i,j}$ is equal to $+G_{i,j}$ (resp. $-G_{i,j}$) if $i$ is an up (resp. down) step and $j$ a down (resp. up) step of $\mu$ and 0 otherwise. Let us then define the matrix $\mathcal{M}_{\pi(\mu)}$ obtained by permuting the rows and columns of $\mathcal{M}_{\mu}$ such that the indices in $U^{\mu}$ appear first and those in $D^{\mu}$ come second:
\begin{equation}
\mathcal{M}_{\mu}=P_{\pi(\mu)}\mathcal{M}_{\pi(\mu)}P_{\pi(\mu)}^{\mathrm{t}},
\end{equation}
where $\pi(\mu)$ denotes the appropriate permutation of the indices of $\mathcal{U}=\{1,2,\ldots,2n\}$ and $P_{\pi(\mu)}$ the corresponding permutation matrix. As $\mathcal{M}_{\pi(\mu)}$ is an antidiagonal block matrix, its Pfaffian can be written as
\begin{equation}
\pf\mathcal{M}_{\pi(\mu)}=(-1)^{\frac{1}{2}n(n-1)}\det G_{U^{\mu}}^{D^{\mu}}.
\end{equation}
To determine the signature of the permutation $\pi(\mu)$, consider first the maximal standard Dyck path $\mu_{\text{max}}$ whose chords are given by $(1,2n),(2,2n{-}1),\ldots,(n,n{+}1)$. The permutation $\pi(\mu_{\text{max}})$ is simply the identity. Assume next that $\mu,\mu'$ are two standard Dyck paths such that $\mu'$ is obtained from $\mu$ by turning an up step $k$ into a down step of $\mu'$, and the adjacent down step $k{+}1$ into an up step of $\mu'$ (note that all Dyck paths contain at least one chord between two adjacent steps). It follows that $\pi(\mu)$ and $\pi(\mu')$ have opposite signatures, and $\mathcal{A}(\mu')=\mathcal{A}(\mu)-2$.

Let then $\mu$ be a standard Dyck path. It is clear that it can be obtained from $\mu_{\text{max}}$ by $N$ elementary reversal operations of the type described above. Using $\mathcal{A}(\mu)=\mathcal{A}(\mu_{\text{max}})-2N=n^2-2N$, one finds the relation
\begin{equation}
\det P_{\pi(\mu)}=(-1)^N\det P_{\pi(\mu_{\text{max}})}=(-1)^{\frac{1}{2}(n^2-\mathcal{A}(\mu))},
\end{equation}
from which the result of the corollary follows.
\end{proof}

Finally, it should be noted that Theorem~\ref{thm5.1} can be adapted to graphs with free boundary conditions (i.e. in which the root $s$ is isolated, or equivalently, absent), as discussed in Section~\ref{sec2.3}. Using Proposition~\ref{grove_thm2}, one may compute the winding partitions $Z[\vsig]$ on such graphs simply by substituting
\begin{equation*}
G_{i,j}\to G_{i,j}+t^{-1}\N^{-1},\quad G'_{i,j}\to G'_{i,j}+t^{-1}\N^{-1}\widetilde{G}'_{i,j},
\end{equation*}
in Theorem~\ref{thm5.1}, where the functions $G'$ and $\widetilde{G}'$ on the right-hand side are defined in Eq.~\eqref{Gp_def2}, and then by taking the limit $t\to 0$.


\section{Correspondence with loop-erased random walks}
\label{sec6}

In this section, we recall the definitions of the standard random walk (SRW) and the loop-erased random walk (LERW) on an unoriented connected graph $\G_s$. The connection between the latter and the random spanning tree has been known for many years (\hspace{-0.025cm}\cite{Pem91} for the uniform distribution, \cite{Wil96} for a generic weighted one): for any two vertices $u,v$ of the graph, the distribution of the LERW started at $u$ and stopped at $v$ is the same as the distribution of the unique chemical path between these two vertices in the random spanning tree. This correspondence is especially apparent through Wilson's algorithm:
\begin{thm}[\hspace{-0.025cm}\cite{Wil96}, Theorem 1\hspace{-0.025cm}]
Let $\G_s$ be an unoriented connected graph, and let $s,v_1,v_2,\ldots,v_{|\V|}$ be an enumeration of its vertices. Define $\mathcal{T}_0$ as the degenerate tree consisting in the vertex $s$, and $\mathcal{T}_j$ as the tree obtained as the union of $\mathcal{T}_{j-1}$ and the loop erasure of a SRW from $v_j$ to $\mathcal{T}_{j-1}$, for $1\le j\le|\V|$. Then $\mathcal{T}_{|\V|}$ is a spanning tree on $\G_s$, occurring with the probability induced by \eqref{ST_weight}.
\label{Wil_alg}
\end{thm}
We provide here an alternative proof based on explicit formulas for the probability measures for the LERW \cite{Mar00,Bau12} and for paths in spanning forests, given by Eq.~\eqref{P_SF}. This derivation allows one to generalize in a straightforward manner the correspondence for multiple walks, whose joint probability measure was given in \cite{Hag09b}.


\subsection{Single walks}
\label{sec6.1}

Let us consider an unoriented connected graph $\G$ with vertices $\V$ and edges $\E$. Let $s$ be an additional vertex connected by a set of edges $\E_s$ to a subset of vertices $\mathcal{D}\subset\V$ in the extended graph $\G_s$, as in Section~\ref{sec2}. We assume for now that $\mathcal{D}$ is nonempty, and discuss the case $\mathcal{D}=\varnothing$ below (see Eq.~\eqref{P_LERW_free}). Let $A_s$ be a transition matrix for $\G_s$ with a root $s$, i.e. such that
\vspace{-0.3cm}
\begin{itemize}
\item $(A_s)_{u,v}\ge 0$ for any $u\in\V$, $v\in\V\cup\{s\}$, with $(A_s)_{u,v}>0$ if and only if $\{u,v\}\in\E\cup\E_s$;
\item $(A_s)_{s,v}=0$ for any $v\in\V\cup\{s\}$;
\item $\sum_{v}(A_s)_{u,v}=1$ for any $u\in\V$, where the sum is over all vertices of $\G_s$, including $s$.
\end{itemize}
\vspace{-0.3cm}
The probabilistic interpretation of the SRW is the following: starting at $v_0$, a walker moving on the edges of the graph goes toward one of the neighbors $v_1$ of $v_0$, with probability $(A_s)_{v_0,v_1}$. The walker continues moving from $v_1$ to one of its neighbors $v_2$ with probability $(A_s)_{v_1,v_2}$. The process is iterated until the walker reaches the root $s$ and remains there forever, since $(A_s)_{s,v}=0$ for any $v\in\V$.

A \emph{walk} or \emph{path} $\omega_s$ on $\G_s$ starting from a vertex $v_0\neq s$ and ending at $s$ is defined as a collection of vertices $(v_0,v_1,\ldots,v_n{=}s)$ such that $(v_i,v_{i+1})\in\E$ for $0\le i\le n{-}2$, and $(v_{n-1},v_n)\in\E_s$. Its SRW weight is defined by
\begin{equation}
\w_{\text{SRW}}(\omega_s)=\prod_{i=0}^{n-1}(A_s)_{v_i,v_{i+1}}.
\label{SRW_weight}
\end{equation}
In what follows, we shall consider a SRW starting at a given vertex $u_1$ and reaching the root $s$ through a fixed edge $(u_2,s)\in\E_s$, with $u_2\in\mathcal{D}$. The partition function can be computed by summing first over all walks of length $n$, and then summing over all positive integers $n$. Its result reads
\begin{equation}
\sum_{\substack{\omega_s:u_1\to s\\ \text{with }(u_2,s)\in\omega_s}}\w_{\text{SRW}}(\omega_s)=\sum_{n=0}^{\infty}(A_s^n)_{u_1,u_2}(A_s)_{u_2,s}=\left[(\mathbb{I}-A_s)^{-1}\right]_{u_1,u_2}\times (A_s)_{u_2,s}.
\label{Loop_weight}
\end{equation}
The convergence of the geometric series follows from the fact that the spectral radius $\rho(A)$ is strictly less than one, where $A$ is the submatrix of $A_s$ with the row and column indexed by $s$ removed. Indeed, $A$ is positive and irreducible (since $\G$ is connected and unoriented). Due to the Perron-Frobenius theorem, $A$ has an eigenvalue $r=\rho(A)$ associated with a left-eigenvector $\mathbf{p}$ whose entries are strictly positive. We can assume without loss of generality that $|\mathbf{p}|_1=\sum_{u}p_u=1$. It follows that
\begin{equation}
\begin{split}
r&=|r\mathbf{p}|_1=|\mathbf{p}A|_1=\sum_{u}\sum_{v}p_u A_{u,v}=\sum_{u}\sum_{v\neq s}p_u(A_s)_{u,v}\\
&=\sum_{u}p_u\sum_{v}(A_s)_{u,v}-\sum_{u}p_u(A_s)_{u,s}=1-\sum_{u}p_u(A_s)_{u,s}<1,
\end{split}
\end{equation}
since $p_u>0$ and there exists at least one $u$ such that $(A_s)_{u,s}>0$. It is then straightforward to see that $\rho(A)<1$ is equivalent to $\sum_{n\ge 0}A_s^n<\infty$ using the Jordan canonical form of $A_s$.

The \emph{loop-erased random walk} (LERW) was introduced in \cite{Law80} as an example of a random process that produces simple walks (i.e. with no self-intersection) and that is easier to treat analytically than the canonical self-avoiding walk; the difference between the two residing in their respective probability measures. The loop erasure of a walk $\omega_s$ is obtained by chronologically erasing loops along $\omega_s$, thus yielding a simple path. This procedure is equivalent to the following (shorter) algorithm. Let $\omega_s=(v_0,v_1,\ldots,v_n{=}s)$ be a walk on the graph $\G_s$. The loop erasure of $\omega_s$, denoted by $\gamma_s=\text{LE}(\omega_s)=(v_{i_0},v_{i_1},\ldots,v_{i_J})$, is defined inductively by
\begin{equation}
i_0=0,\quad i_{j+1}=\max_{0\le k\le n}\{k:v_k=v_{i_j}\}+1,
\label{LE_alg}
\end{equation}
which stops at $j=J$ such that $v_{i_J}=s$. An illustration of the loop erasure procedure is provided in Fig.~\ref{Loop_erasure}. The LERW weight of a simple path $\gamma_s$ is defined as the sum of the weights of all walks $\omega_s$ whose loop erasure yields $\gamma_s$:
\begin{equation}
\w_{\text{LERW}}(\gamma_s)=\sum_{\omega_s:\,\text{LE}(\omega_s)=\gamma_s}\w_{\text{SRW}}(\omega_s).
\label{LERW_weight_def}
\end{equation}

To compute an explicit form for the LERW weight in terms of the matrix $A_s$, one may use the algorithm defined in Eq.~\eqref{LE_alg} to decompose a walk $\omega_s$ such that $\text{LE}(\omega_s)=\gamma_s=(v_0,v_1,\ldots,v_n{=}s)$ as follows. If we define $\G_s^{(j)}\equiv\G_s\bs\{v_0,v_1,\ldots,v_{j}\}$, then $\omega_s$ can be seen as a (possibly trivial) loop attached to $v_0$ on $\G_s$, followed by the edge $(v_0,v_1)$; a loop attached to $v_1$ on $\G_s^{(0)}$ comes next, followed by the edge $(v_1,v_2)$; and so on. The process is iterated until the edge $(v_{n-1},v_n{=}s)$ is used, at which point the walk is stopped. If $A_s^{(j)}$ denotes the restriction of the weight matrix $A_s$ to the vertices of $\G_s^{(j)}$, then
\begin{equation}
\sum_{\substack{\omega_j:v_j\to v_j\\ \text{on }\G_s^{(j-1)}}}\w_{\text{SRW}}(\omega_j)=\left[\left(\mathbb{I}-A_s^{(j-1)}\right)^{-1}\right]_{v_j,v_j},
\end{equation}
where the sum is over all loops attached to $v_j$ conditioned to not intersect $\{v_0,v_1,\ldots,v_{j-1}\}$. The preceding equation allows one to express Eq.~\eqref{LERW_weight_def} as
\begin{equation}
\begin{split}
\w_{\text{LERW}}(\gamma_s)&=\left[\left(\mathbb{I}-A_s\right)^{-1}\right]_{v_0,v_0}(A_s)_{v_0,v_1}\left[\left(\mathbb{I}-A_s^{(0)}\right)^{-1}\right]_{v_1,v_1}(A_s)_{v_1,v_2}\left[\left(\mathbb{I}-A_s^{(1)}\right)^{-1}\right]_{v_2,v_2}\\
&\quad\times\ldots\times(A_s)_{v_{n-2},v_{n-1}}\left[\left(\mathbb{I}-A_s^{(n-2)}\right)^{-1}\right]_{v_{n-1},v_{n-1}}(A_s)_{v_{n-1},v_n}.
\end{split}
\end{equation}
Using Cramer's formula for the inverse of a matrix yields telescopic cancellations, and the result simplifies to \cite{Mar00,Bau12,Law18}
\begin{equation}
\begin{split}
\w_{\text{LERW}}(\gamma_s)&=\prod_{i=0}^{n-1}(A_s)_{v_i,v_{i+1}}\times\frac{\det\left(\mathbb{I}-A_s^{(0)}\right)}{\det\left(\mathbb{I}-A_s\right)}\times\frac{\det\left(\mathbb{I}-A_s^{(1)}\right)}{\det\left(\mathbb{I}-A_s^{(0)}\right)}\times\ldots\times\frac{\det\left(\mathbb{I}-A_s^{(n-1)}\right)}{\det\left(\mathbb{I}-A_s^{(n-2)}\right)}\\
&=\w_{\text{SRW}}(\gamma_s)\times\frac{\det\left(\mathbb{I}-A_s^{(n-1)}\right)}{\det\left(\mathbb{I}-A_s\right)}.
\label{LERW_weight}
\end{split}
\end{equation}


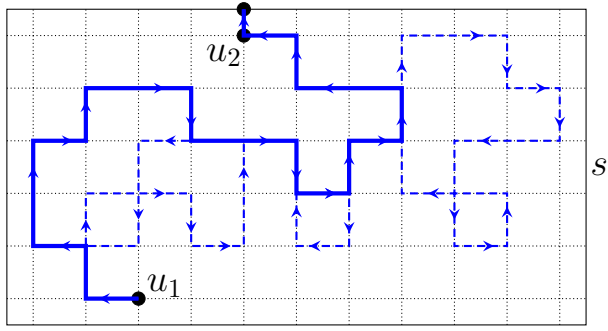
\begin{figure}[h]
\centering
\begin{tikzpicture}[scale=0.7]
\draw[densely dotted] (-0.5,-0.5) grid (10.5,5.5);
\draw (-0.5,-0.5) rectangle (10.5,5.5);
\filldraw (2,0) circle (0.125cm);
\filldraw (4,5.5) circle (0.125cm);
\filldraw (4,5) circle (0.125cm);
\node at (2.5,0.25) {\Large $u_1$};
\node at (3.625,4.625) {\Large $u_2$};
\node at (10.75,2.5) {\Large $s$};
\draw[postaction={on each segment={mid arrow}},thick,densely dashed,blue] (2,0)--(1,0)--(1,2)--(3,2)--(3,1)--(4,1)--(4,3)--(2,3)--(2,1)--(0,1)--(0,3)--(1,3)--(1,4)--(3,4)--(3,3)--(5,3)--(5,2)--(6,2)--(6,1)--(5,1)--(5,2)--(6,2)--(6,3)--(7,3)--(7,5)--(9,5)--(9,4)--(10,4)--(10,3)--(8,3)--(8,1)--(9,1)--(9,2)--(7,2)--(7,4)--(5,4)--(5,5)--(4,5)--(4,5.5);
\draw[ultra thick,blue] (2,0)--(1,0)--(1,1)--(0,1)--(0,3)--(1,3)--(1,4)--(3,4)--(3,3)--(5,3)--(5,2)--(6,2)--(6,3)--(7,3)--(7,4)--(5,4)--(5,5)--(4,5)--(4,5.5);
\end{tikzpicture}
\caption{Walk (dashed) from $u_1$ to the root $s$ (represented by the outer box) on a wired rectangular grid, that passes through the edge $(u_2,s)$. The simple path obtained by erasing all loops as they appear is drawn with a heavy line.}
\label{Loop_erasure}
\end{figure}

Let us now consider a loop-erased random walk on $\G_s$ starting at a fixed vertex $u_1\in\G$, and reaching $s$ through a specified edge $(u_2,s)$. We denote the corresponding simple paths by $u_1\to u_2\to s$. The partition function reads
\begin{equation}
\sum_{\gamma_s:u_1\to u_2\to s}\w_{\text{LERW}}(\gamma_s)=\sum_{\gamma_s:u_1\to u_2\to s}\sum_{\omega_s:\,\text{LE}(\omega_s)=\gamma_s}\w_{\text{SRW}}(\omega_s)=\sum_{\omega_s:u_1\to u_2\to s}\w_{\text{SRW}}(\omega_s),
\end{equation}
whose result is given by \eqref{Loop_weight}. The LERW probability measure on simple paths $\gamma_s:u_1\to u_2\to s$ is therefore given by
\begin{equation}
\P_{\text{LERW}}(\gamma_s)=\w_{\text{SRW}}(\gamma_s)\frac{\det\left(\mathbb{I}-A_s^{(\gamma)}\right)}{\det(\mathbb{I}-A_s)[\left(\mathbb{I}-A_s\right)^{-1}]_{u_1,u_2}(A_s)_{u_2,s}},
\label{P_LERW_s}
\end{equation}
where $\gamma\equiv\gamma_s\bs\{s\}$ and $A_s^{(\gamma)}$ denotes the restriction of $A_s$ to the rows and columns indexed by vertices not belonging to $\gamma$. Note that as the factor $(A_s)_{u_2,s}$ appears both in $\w_{\text{SRW}}(\gamma_s)$ and at the denominator, it can be simplified. Seeing this, it is natural to rewrite this measure for simple paths $\gamma:u_1\to u_2$ on the graph $\G$, which are in one-to-one correspondence with simple paths $\gamma_s:u_1\to s$ passing through $(u_2,s)$ on $\G_s$. If we let $A$ denote the submatrix of $A_s$ obtained by removing the row and column indexed by $s$, then
\begin{equation}
\mathbb{I}-A_s=\bordermatrix{& s & \V\cr
& 1 & \mathbf{0}\cr
& * & \mathbb{I}-A
},
\end{equation}
where $*$ and $\mathbf{0}$ represent $\N\times 1$ and $1\times \N$ vectors, respectively ($\N=|\V|$). From this equation, one may rewrite Eq.~\eqref{P_LERW_s} as follows for simple paths $\gamma_s=\gamma\cup\{s\}$ from $u_1\to u_2\to s$:
\begin{equation}
\begin{split}
\P_{\text{LERW}}(\gamma)&\equiv\P_{\text{LERW}}(\gamma_s)=\w_{\text{SRW}}(\gamma_s)\frac{\det\left(\mathbb{I}-A_s^{(\gamma)}\right)}{\det(\mathbb{I}-A_s)[\left(\mathbb{I}-A_s\right)^{-1}]_{u_1,u_2}(A_s)_{u_2,s}}\\
&=\w_{\text{SRW}}(\gamma)\frac{\det\left(\mathbb{I}-A^{(\gamma)}\right)}{\det(\mathbb{I}-A)[(\mathbb{I}-A)^{-1}]_{u_1,u_2}}.
\label{P_LERW}
\end{split}
\end{equation}
With Eq.~\eqref{P_LERW}, one has a probability measure on simple paths on the graph $\G$ from $u_1$ to $u_2$, in which any explicit reference to the root $s$ has been eliminated. While it is often more convenient to use this equation rather than Eq.~\eqref{P_LERW_s} for what follows, one should remember that the presence of $s$ translates into $A$ being substochastic, i.e its rows indexed by vertices $u\in\mathcal{D}$ (the neighbors of $s$ on $\G_s$) sum to strictly less than one (which ensures in particular the invertibility of $\mathbb{I}-A$).

Let us now show that Eq.~\eqref{P_LERW} coincides with the measure on paths in two-component spanning forests, defined in \eqref{P_SF}.
\begin{prop}
Let $\gamma=(v_0{=}u_1,v_1,\ldots,v_{n-1},v_n{=}u_2)$ be a simple path on $\G$, and $A$ be the transition matrix defined by
\begin{equation}
A_{u,v}=\frac{c_{u,v}}{\deg_s(u)}
\label{AC_rel}
\end{equation}
if $(u,v)\in\E$ and 0 otherwise. Here $\deg_s(u)=\sum_{w\sim u}c_{u,w}$ is the (weighted) degree of $u$ in $\G_s$, where the sum is over all neighbors of $u$ in $\G_s$. Then
\begin{equation}
\P_{\mathrm{LERW}}(\gamma)=\P_{\mathrm{SF}}(\gamma).
\end{equation}
As a consequence, any simple path $\gamma$ and the reverse path $\gamma^{-1}$ occur with the same probability with respect to the LERW measure.
\label{prop6.1}
\end{prop}
\begin{proof}
Let us define the matrix $M_{u,v}=\deg_s(u)\,\delta_{u,v}$ for $u,v\in\V$. It is straightforward to check that the Dirichlet Laplacian on $\G_s$ reads $\Delta=M(\mathbb{I}-A)$, which implies that $\det(\mathbb{I}-A)=\det\Delta/\det M$ and
\begin{equation}
\det\left(\mathbb{I}-A^{(\gamma)}\right)=\det\Delta^{(\gamma)}/\det M^{(\gamma)},\quad\text{with }\det M=\det M^{(\gamma)}\times\prod_{i=0}^{n}\deg_s(v_i).
\end{equation}
The SRW weight of the path $\gamma$ and the $(u_1,u_2)$ entry of the matrix $(\mathbb{I}-A)^{-1}$ are given by
\begin{equation}
\w_{\text{SRW}}(\gamma)=C(\gamma)\deg_s(u_2)\prod_{i=0}^{n}\frac{1}{\deg_s(v_i)},\quad[(\mathbb{I}-A)^{-1}]_{u_1,u_2}=\deg_s(u_2)\,G_{u_1,u_2}.
\end{equation}
Substituting the preceding equations into Eq.~\eqref{P_LERW} yields the result of the proposition.
\end{proof}

Before we turn to multiple walks, note that we have assumed the endpoint $u_2$ to be connected to the root $s$ for random walks (i.e. $u_2\in\mathcal{D}$), so that the walker stops moving once he reaches the root $s$ through the edge $(u_2,s)$. The probability \eqref{P_LERW} may however be extended for simple paths $\gamma:u_1\to u_2$ such that $u_2\notin\mathcal{D}$. To see this, let us denote by $\G^{\varepsilon}_s$ the graph obtained by adding an extra edge $(u_2,s)$ to $\G_s$ such that the corresponding transition matrix $A^{\varepsilon}_s$ coincides with $A_s$ everywhere except on the row indexed by $u_2$, with
\begin{equation}
(A^{\varepsilon}_s)_{u_2,v}=(1-\varepsilon)\times(A_s)_{u_2,v}\text{ for $v\neq s$},\quad(A^{\varepsilon}_s)_{u_2,s}=\varepsilon.
\end{equation}
The weight of a simple path $\gamma_s^{\varepsilon}=\gamma\cup\{s\}:u_1\to u_2\to s$ on $\G_s^{\varepsilon}$ is then given by Eq.~\eqref{LERW_weight}:
\begin{equation}
\w_{\text{LERW}}(\gamma_s^{\varepsilon})=\w_{\text{SRW}}(\gamma_s^{\varepsilon})\frac{\det\left(\mathbb{I}-(A_s^{\varepsilon})^{(\gamma)}\right)}{\det(\mathbb{I}-A_s^{\varepsilon})}=\varepsilon\,\w_{\text{SRW}}(\gamma)\frac{\det\left(\mathbb{I}-A^{(\gamma)}\right)}{\det(\mathbb{I}-A)}+\ldots
\end{equation}
at lowest order in $\varepsilon$. It is then natural to define the probability associated with a path $\gamma:u_1\to u_2$ on $\G$ as
\begin{equation}
\P_{\text{LERW}}(\gamma)=\lim_{\varepsilon\to 0^+}\P_{\text{LERW}}(\gamma_s^{\varepsilon}).
\label{P_LERW_ext}
\end{equation}
The result of Eq.~\eqref{P_LERW_ext} is given by Eq.~\eqref{P_LERW}, which holds therefore whether $u_2$ is in $\mathcal{D}$ or not, as does Proposition~\ref{prop6.1}.

Finally, the case $\mathcal{D}=\varnothing$---i.e. graphs $\G$ without a root, or equivalently, extended graphs $\G_s$ in which $s$ is isolated---can be dealt with in a similar way, namely by adding an extra edge $(u_2,s)$ with transition probability $\varepsilon$, and then taking the limit $\varepsilon\to 0^+$ of the probability distribution defined on $\G_s^{\varepsilon}$. The result reads
\begin{equation}
\P_{\text{LERW}}(\gamma)=\w_{\text{SRW}}(\gamma)\frac{\det\left(\mathbb{I}-A^{(\gamma)}\right)}{(-1)^{u_1+u_2}\det(\mathbb{I}-A)_{\V\bs\{u_2\}}^{\V\bs\{u_1\}}}.
\label{P_LERW_free}
\end{equation}
It is straightforward to show that Proposition~\ref{prop6.1} is valid in this case as well.


\subsection{Multiple walks}
\label{sec6.2}

Up to this point we have considered a single LERW on the graph $\G_s$, starting at any fixed vertex $u_1\in\G$ and ending at $s$ through a fixed edge $(u_2,s)\in\E_s$. We have shown that such paths can be traded for shorter ones from $u_1$ to $u_2$, for which we have established the link with the spanning forest measure. We now generalize this process to multiple simple paths between marked vertices (nodes) that do not intersect each other.

Two definitions of multiple LERWs seem the most natural: either one considers SRWs conditioned to not intersect each other, or one asks only that their loop erasures do not intersect (except at $s$ in both cases). We shall favor here the second definition, with the additional requirement that the $j$th SRW do not intersect the loop erasures of the first $j{-}1$ walks on $\G$, following \cite{Fom01,Hag09b}. We consider a set of nodes $\mathcal{U}=\{u_1,u_2,\ldots,u_{2n}\}\subset\V$, which we divide into two subsets, $R=\{r_1,\ldots,r_n\}$ and $S=\mathcal{U}\bs R=\{s_1,\ldots,s_n\}$. We assume that $(s_i,s)\in\E_s$ and consider LERWs from $r_i$ to $s$ conditioned to exit $\G$ through $(s_i,s)$, for $1\le i\le n$. The weight of a given $n$-tuple of paths $\Gamma_s=(\gamma_1^s,\ldots,\gamma_n^s)$ is defined as follows \cite{Hag09b},
\begin{equation}
\w^{\G_s}_{\text{LERW}}(\Gamma_s)=\w^{\G_s}_{\text{LERW}}(\gamma_1^s)\times\w^{\G_s}_{\text{LERW}}(\gamma_2^s|\gamma_1^s)\times\cdots\times\w^{\G_s}_{\text{LERW}}(\gamma_n^s|\gamma_1^s,\ldots,\gamma_{n-1}^s),
\end{equation}
where the factor $\w^{\G_s}_{\text{LERW}}(\gamma_i^s|\gamma_1^s,\ldots,\gamma_{i-1}^s)$ on the right-hand side is the sum over all walks $\omega_i^s:r_i\to s_i\to s$ with $\text{LE}(\omega_i^s)=\gamma_i^s$ that only intersect $\bigcup_{j=1}^{i-1}\gamma_j^s$ at the root $s$. Conditioning $\omega_i=\omega_i^s\bs\{s\}$ to not intersect $\gamma_j=\gamma_j^s\bs\{s\}$ for $1\le j\le i{-}1$ amounts to growing $\omega_i^s$ on $\G_s\bs\bigcup_{j=1}^{i-1}\gamma_j$ (the graph from which the vertices of $\gamma_1,\ldots,\gamma_{i-1}$ have been removed, together with their incident edges). Using Eq.~\eqref{LERW_weight}, one can therefore rewrite the preceding equation as
\begin{equation}
\begin{split}
\w^{\G_s}_{\text{LERW}}(\Gamma_s)&=\w^{\G_s}_{\text{LERW}}(\gamma_1^s)\times\cdots\times\w^{\G_s\bs\bigcup_{j=1}^{n-1}\gamma_j}_{\text{LERW}}(\gamma_n^s)\\
&=\w^{\G_s}_{\text{SRW}}(\gamma_1^s)\frac{\det\left(\mathbb{I}-A_s^{(\gamma_1)}\right)}{\det(\mathbb{I}-A_s)}\times\cdots\times\w^{\G_s\bs\bigcup_{j=1}^{n-1}\gamma_j}_{\text{SRW}}(\gamma_n^s)\frac{\det\left(\mathbb{I}-A_s^{(\gamma_1,\ldots,\gamma_n)}\right)}{\det\left(\mathbb{I}-A_s^{(\gamma_1,\ldots,\gamma_{n-1})}\right)}\\
&=(A_s)_{s_1,s}\,\w^{\G}_{\text{SRW}}(\gamma_1)\frac{\det\left(\mathbb{I}-A^{(\gamma_1)}\right)}{\det(\mathbb{I}-A)}\times\cdots\times(A_s)_{s_n,s}\,\w^{\G}_{\text{SRW}}(\gamma_n)\frac{\det\left(\mathbb{I}-A^{(\gamma_1,\ldots,\gamma_n)}\right)}{\det\left(\mathbb{I}-A^{(\gamma_1,\ldots,\gamma_{n-1})}\right)}\\
&=\prod_{j=1}^{n}(A_s)_{s_j,s}\times\w^{\G}_{\text{SRW}}(\Gamma)\frac{\det\left(\mathbb{I}-A^{(\Gamma)}\right)}{\det(\mathbb{I}-A)},
\label{nLERW_weight_s}
\end{split}
\end{equation}
where $\w_{\text{SRW}}(\Gamma)=\prod_{i=1}^{n}\w_{\text{SRW}}(\gamma_i)$ and $A^{(\Gamma)}$ denotes the restriction of $A$ to vertices not belonging to $\Gamma=\Gamma_s\bs\{s\}$. In particular, it follows from Eq.~\eqref{nLERW_weight_s} that the joint measure on multiple simple paths does not depend on the order in which the walks are grown. As the product over $j$ on the right-hand side is constant over all $n$-tuples of paths on $\G_s$, we may consider simple paths on $\G$ instead. The weight of an $n$-tuple of nonintersecting simple paths $\Gamma=(\gamma_1,\ldots,\gamma_n)$, with $\gamma:r_i\to s_i$ on $\G$ for $1\le i\le n$, is then defined as \cite{Hag09b}
\begin{equation}
\w_{\text{LERW}}(\Gamma)=\w_{\text{SRW}}(\Gamma)\frac{\det\left(\mathbb{I}-A^{(\Gamma)}\right)}{\det(\mathbb{I}-A)}.
\label{nLERW_weight}
\end{equation}

Let us now show that the SF and LERW weights on $\Gamma$, given respectively by Eq.~\eqref{nSF_weight} and Eq.~\eqref{nLERW_weight}, are proportional to each other---thus implying that both normalized measures coincide.
\begin{prop}
Let $\G_s=\G\cup\{s\}$ be an unoriented connected graph with a root $s$ and a set of nodes $\mathcal{U}=\{u_1,\ldots,u_{2n}\}\subset\V$. Let $R=\{r_1,\ldots,r_n\}\subset\mathcal{U}$ and $S=\{s_1,\ldots,s_n\}=\mathcal{U}\bs R$ such that $(s_i,s)\in\E_s$ for $1\le i\le n$. Consider $n$ nonintersecting simple paths $\gamma_i:r_i\to s_i$ on $\G$, and write $\Gamma=(\gamma_1,\ldots,\gamma_n)$. If $A_{u,v}=c_{u,v}/\deg_s(u)$ for $(u,v)\in\E$ and 0 otherwise, then
\begin{equation}
\w_{\mathrm{LERW}}(\Gamma)=\mathrm{K}\times\w_{\mathrm{SF}}(\Gamma),
\end{equation}
where $\mathrm{K}=\mathrm{K}(\G_s,S)$ is independent of $\Gamma$. As a consequence, $\P_{\mathrm{LERW}}(\Gamma)=\P_{\mathrm{SF}}(\Gamma)$.
\label{prop6.2}
\end{prop}
\begin{proof}
Let $M$ be the matrix defined by $M_{u,v}=\deg_s(u)\,\delta_{u,v}$ for $u,v\in\V$. Then $\Delta=M(\mathbb{I}-A)$ and
\begin{equation}
\w_{\mathrm{SRW}}(\Gamma)=\frac{C(\Gamma)\deg_s(S)}{\deg_s(\Gamma)},\quad\frac{\det\left(\mathbb{I}-A^{(\Gamma)}\right)}{\det(\mathbb{I}-A)}=\frac{\det\Delta^{(\Gamma)}\deg_s(\Gamma)}{\det\Delta},
\end{equation}
where $\deg_s(S)$ and $\deg_s(\Gamma)$ are the products of the degrees (on $\G_s$) of all vertices belonging to $S$ and $\Gamma$, respectively. Hence we find
\begin{equation}
\mathrm{K}=\frac{\deg_s(S)}{\det\Delta},
\end{equation}
which concludes the proof.
\end{proof}

A less explicit argument for the equality between the LERW and SF measures for multiple paths, based on Wilson's algorithm (Theorem~\ref{Wil_alg}), was given in \cite{KKP17}. The idea consists in picking an enumeration of the vertices of $\G_s$ such that $s$ and $R=\{r_1,\ldots,r_n\}$ appear first, and then growing a SRW from $r_1$ to $s$ exiting $\G$ through the edge $(s_1,s)$, followed by a second SRW from $r_2$ to $s$ through $(s_2,s)$ that does not intersect the loop erasure of the first SRW, and so on, such that the $j$th SRW from $r_j$ to $s$ through $(s_j,s)$ does not intersect the loop erasures of the first $j{-}1$ SRWs. The algorithm resumes as in Theorem~\ref{Wil_alg} when the $n$ paths from $R$ to $S$ have been constructed, and yields a spanning tree associated with a unique spanning forest with $n{+}1$ trees (each of which containing a pair $r_i,s_i$ or the root $s$), obtained by removing the prescribed edges $(s_1,s),\ldots,(s_n,s)$ from the spanning tree.

Finally, note that the requirement that the endpoints $s_1,\ldots,s_n$ of LERWs be neighbors of the root may be relaxed, allowing some or all of them to not be connected to $s$ on $\G_s$. If this is the case for $k$ elements of $S$---say $s_1,\ldots,s_k$---define a modified graph $\G_s^{\varepsilon}$ by adding $k$ edges $(s_i,s)$ to $\G_s$ such that
\begin{equation}
(A_s^{\varepsilon})_{s_i,v}=(1-\varepsilon)\times(A_s)_{s_i,v}\text{ for $v\neq s$},\quad(A_s^{\varepsilon})_{s_i,s}=\varepsilon,
\end{equation}
for $1\le i\le k$. The weight of any $n$-tuple $\Gamma_s^{\varepsilon}$ on $\G_s^{\varepsilon}$ then reads
\begin{equation}
\w_{\text{LERW}}(\Gamma_s^{\varepsilon})=\varepsilon^k\,\w_{\text{LERW}}(\Gamma)+\ldots
\end{equation}
at leading order in $\varepsilon\sim 0^+$, with $\w_{\text{LERW}}(\Gamma)$ given by Eq.~\eqref{nLERW_weight}. Consequently, the result of Proposition~\ref{prop6.2} is valid in this setting as well. A similar argument can be made if $s$ is entirely disconnected from the rest of the graph.

The equality between $\P_{\text{LERW}}$ and $\P_{\text{SF}}$ holds therefore for any positions of the extremities of the paths, on any unoriented connected graph, with or without a root. However, concrete computations of probabilities require knowledge of the explicit expressions of the partition functions for certain classes of walks from $R$ to $S$ (see Sections~\ref{sec3} and \ref{sec5}). The only case for which we have been able to derive such expressions is for paths between boundary vertices of graphs embedded on surfaces (with noncrossing edges), with unit conductances (see Section~\ref{sec4}).


\section{Comparison with SLE/CFT results}
\label{sec7}

Let us now discuss how our formulas compare, in the scaling limit, with known results obtained in the framework of Schramm-Loewner evolution (SLE) and conformal field theory (CFT). We concentrate here on the case of the upper half-plane $\mathbb{H}$ with Dirichlet boundary conditions, which is the canonical setup for both SLE and CFT on domains with a boundary.

The original formula for the left-passage probability for a single $\text{SLE}_2$ curve is due to Schramm \cite{Sch01}. Its generalization to multiple random curves was first discussed in \cite{GC05}, in which the authors provide explicit formulas for $n=2$ curves only and $\kappa=0,2,8/3,4,8$, using the correspondence between $\text{SLE}_{\kappa}$ and a CFT with central charge $c=(3\kappa-8)(6-\kappa)/(2\kappa)$ \cite{BB02,FW03} (a more rigorous SLE approach has recently been developed in \cite{LV17}, again for $n=2$ curves). Their computations make use of the fact that conditioning an SLE partition function on the existence of a curve connecting two boundary points $x_1,x_2$ is equivalent to inserting a boundary operator $\phi$ at $x_1,x_2$ in the corresponding CFT correlator \cite{Car84,BB02,FW03}. The field $\phi$ has a weight $h_{\phi}=(6-\kappa)/(2\kappa)$, and is degenerate at level two; thus yielding a second-order differential equation for any correlator containing a $\phi$ \cite{BPZ84}:
\begin{equation}
\left(\frac{3}{2(2h_{\phi}+1)}\partial_{x_1}^2-\sum_{i=2}^{k}\frac{1}{x_1-z_i}\partial_{z_i}-\sum_{i=2}^{k}\frac{h_i}{(x_1-z_i)^2}\right)\langle\phi(x_1)\mathcal{O}_2(z_2)\ldots\mathcal{O}_k(z_k)\rangle=0,
\label{BPZ_eq}
\end{equation}
where $x_1\in\partial\mathbb{H}$, $z_i\in\mathbb{H}$ and $\mathcal{O}_i$ are arbitrary fields of weight $h_i$, $2\le i\le k$. The condition that a random curve from $x_1$ to $x_2$ leaves a point $z$ to its left or right is implemented by the insertion of an indicator operator $\chi(z)$ of weight zero, treated as a local operator \cite{BB03}. Explicitly, the scaling limit of the probability that a simple path leaves $z$ to its left (or right) is then given by
\begin{equation}
\P(x_1,x_2;z)=\frac{\langle\chi(z)\phi(x_1)\phi(x_2)\rangle}{\langle\phi(x_1)\phi(x_2)\rangle}.
\end{equation}
The probabilities $\P_{\L}$ or $\P_{\R}$ are found by solving Eq.~\eqref{BPZ_eq} and imposing different boundary values for $\P(x_1,x_2;z)$ as a function of $z$. The generalization to multiple curves in straightforward in the CFT framework, and yields
\begin{equation}
\P(\mathbf{x};z)=\frac{\langle\chi(z)\Phi(\mathbf{x})\rangle}{\langle\Phi(\mathbf{x})\rangle},
\label{nP_CFT}
\end{equation}
where $\mathbf{x}=(x_1,x_2,\ldots,x_{2n})$ and $\Phi(\mathbf{x})=\phi(x_1)\phi(x_2)\ldots\phi(x_{2n})$. Here $\chi(z)$ is any of the indicator operators that each curve leaves $z$ to its left or right. As for a single curve, the boundary conditions allow us to determine which case is computed, as well as the way the $x_i$'s are paired together. Due to the level-two degeneracy of $\phi$, both the numerator and the denominator of Eq.~\eqref{nP_CFT} satisfy $2n$ BPZ equations of the form of Eq.~\eqref{BPZ_eq}. Particularizing to $\kappa=2$ (which implies that $c=-2$ and $h_{\phi}=1$), we find that the equation with respect to $x_1$ for the numerator reads:
\begin{equation}
\left(\frac{1}{2}\partial_{x_1}^2-\sum_{i=2}^{2n}\left(\frac{1}{x_1-x_i}\partial_{x_i}+\frac{1}{(x_1-x_i)^2}\right)-\frac{1}{x_1-z}\partial_z-\frac{1}{x_1-\bar{z}}\partial_{\bar{z}}\right)\langle\chi(z)\Phi(\mathbf{x})\rangle=0.
\label{BPZ_eq_part}
\end{equation}
The corresponding equation for the denominator $\langle\Phi(\mathbf{x})\rangle$ is the same except for the last two terms, which are absent. In \cite{GC05} the authors noted that solving these coupled equations is already too difficult for $n=2$, even when the endpoints of both curves are sent to infinity. They considered the limit of ``fused'' curves, namely, when the distance between the starting points of the curves goes to zero, for which they found analytical solutions of the fused BPZ equations. Their results for $\kappa=2$ read, with $z=x+\i y$ and $t=x/y$,
\begin{equation}
\begin{split}
\P_{\L\L}(t)&=\frac{1}{4}+\frac{1}{9\pi^2(1+t^2)^3}\big[(-16-9t^2+9t^4)-9\pi(t^3+t^5)\\
&\qquad+9(1+t^2)\arctan(t)\big(2t^3-\pi(1+t^2)^2+(1+t^2)^2\arctan(t)\big)\big],\\
\P_{\R\R}(t)&=\P_{\L\L}(-t),\quad\P_{\R\L}(t)=1-\P_{\L\L}(t)-\P_{\R\R}(t).
\end{split}
\label{GC_res}
\end{equation}

In comparison, note that our results for two curves on the upper half-plane are valid for any positions of the starting and ending points, provided their order is fixed: $x_1<x_2<x_3<x_4$. Their explicit expressions can be found using Eqs \eqref{2Schramm_1} and \eqref{2Schramm_2}, in which the Green function and its derivative are replaced in the scaling limit with the excursion Poisson kernel $\text{P}$ and its derivative $\text{P}'$, as argued in Section~\ref{sec4.1}. For instance, if $x_1,x_2$ are paired with $x_4,x_3$, respectively, the winding probabilities read in the scaling limit
\begin{equation}
\begin{split}
\P_{\L\L}(14|23)&=1-\P_{\R\L}(14|23)-\P_{\R\R}(14|23),\\
\P_{\R\L}(14|23)&=\frac{\text{P}_{1,3}\,\text{P}'_{2,4}+\text{P}'_{1,3}\,\text{P}_{2,4}-\text{P}_{1,4}\,\text{P}'_{2,3}-\text{P}'_{1,4}\,\text{P}_{2,3}}{\text{P}_{1,3}\,\text{P}_{2,4}-\text{P}_{1,4}\,\text{P}_{2,3}}-2\,\P_{\R\R}(14|23),\\
\P_{\R\R}(14|23)&=\frac{-\text{P}'_{1,2}\,\text{P}'_{3,4}+\text{P}'_{1,3}\,\text{P}'_{2,4}-\text{P}'_{1,4}\,\text{P}'_{2,3}}{\text{P}_{1,3}\,\text{P}_{2,4}-\text{P}_{1,4}\,\text{P}_{2,3}},
\label{2Schramm_expl}
\end{split}
\end{equation}
where $\text{P}_{i,j}\equiv\text{P}(x_i,x_j)$ and $\text{P}'_{i,j}\equiv\text{P}'(x_i,x_j;z)$ are given by
\begin{equation}
\begin{split}
\text{P}(x_i,x_j)&=\frac{1}{\pi(x_i-x_j)^2},\\
\text{P}'(x_i,x_j;z)&=-\frac{1}{\pi^2(x_i-x_j)^2}\big(\arg(z-x_i)-\arg(z-x_j)\big)\\
&\qquad+\frac{1}{\pi^2(x_i-x_j)^2}\frac{\Re[(z-x_i)(\bar{z}-x_j)]\Im[(z-x_i)(\bar{z}-x_j)]}{|z-x_i|^2\,|z-x_j|^2}.
\end{split}
\end{equation}
The probabilities \eqref{GC_res} are recovered by taking the limit $x_1,x_2\to-\infty$ and $x_3,x_4\to 0$ in Eq.\eqref{2Schramm_expl}.

A further consistency check of the CFT interpretation consists in showing that the (scaling limit of) the six distinct winding partition functions
\begin{equation*}
Z_{\L\L}[12|34],\,Z_{\L\R}[12|34],\,Z_{\R\L}[12|34],\,Z_{\L\L}[14|23],\,Z_{\R\L}[14|23],\,Z_{\R\R}[14|23]
\end{equation*}
all satisfy Eq.~\eqref{BPZ_eq_part}. It is indeed the case, and a similar check holds for the total partition functions $Z_{\mathrm{t}}[12|34]$ and $Z_{\mathrm{t}}[14|23]$. More generally, it would be interesting to provide the same validation of the CFT version of Schramm's formulas for any number of curves. While the total partition functions $Z_{\mathrm{t}}[\sigma]$ have been proved to satisfy the BPZ equation in \cite{KKP17}, the analogue check for the winding partitions functions $Z[\vsig]$ is more difficult and remains to be done.


\section{Summary}
\label{sec8}

This paper has described how a connection on a planar graph may be used to compute Schramm's formula for the loop-erased random walk, namely the left-passage probability of a path between fixed boundary vertices with respect to a marked face. Expanding on some preliminary results \cite{Ken11,KW15}, we have obtained exact probabilities for various domains and boundary conditions, and compared them with known SLE results to find a full agreement. Schramm's formula has then been generalized for multiple nonintersecting paths, for which we have given an explicit expression in terms of the Green function of the graph.

Another interesting direction to investigate would consist in the left-passage probability of a single curve with respect to \emph{two} marked faces (or points in the scaling limit). The unique case for which this probability has been computed in the SLE/CFT framework is for $\kappa=8/3$ \cite{SC09} (with more rigorous proofs given later in \cite{BV13}). For a LERW, using a complex connection supported on two zippers with respective parameters $z$ and $w$, one attached to each marked face, does not suffice to find combinatorial expressions for the partition functions of interest (indeed, the same product of parallel transports would be assigned to topologically distinct classes of paths). A more appropriate approach has been sketched in \cite{KW15}, and requires one to use a matrix-valued $\text{SL}(2,\C)$-connection (for which an equivalent of Theorem~\ref{CRSF_thm} exists \cite{Ken11}). Knowledge of the associated vector bundle Green function in terms of the standard Green function would however be required for explicit computations, and is still lacking at the moment.


\subsection*{Acknowledgments}
This work is supported by the Belgian Interuniversity Attraction Poles Program P7/18 through the network DYGEST (Dynamics, Geometry and Statistical Physics). The author wishes to thank Philippe Ruelle for his careful reading of the manuscript and many useful suggestions and comments. Moreover, the author is grateful to Christian Hagendorf for bringing the problem to his attention and commenting on an early version of this paper, and to Alexi Morin-Duchesne and Jean Li{\'e}nardy for useful remarks.


\clearpage
\appendix
\section{Jacobi's theta and elliptic functions}
\label{a1}

In this appendix, we recall the definitions of Jacobi's theta and elliptic functions, give some well-known representations of these functions and list the properties we used in this paper. In what follows, $z,\tau\in\C$ with $\Im\tau>0$, where $\tau$ is called the \emph{lattice parameter} or the \emph{half-period ratio}. The \emph{nome} $q$ defined by $q=e^{\i\pi\tau}$ is such that $0<|q|<1$.

\subsection{Jacobi's theta functions}
The four theta functions $\vartheta_a(z,q)$, with $1\le a\le 4$, are $2\pi$-periodic functions of $z$ for fixed $q$, defined by their Fourier series
\begin{align}
&\vartheta_1(z,q)=2\sum_{n=0}^{\infty}(-1)^{n}q^{(n+\frac{1}{2})^{2}}\sin[(2n+1)z],\quad&\vartheta_2(z,q)=2\sum_{n=0}^{\infty}q^{(n+\frac{1}{2})^{2}}\cos[(2n+1)z],\\
&\vartheta_3(z,q)=1+2\sum_{n=1}^{\infty}q^{n^{2}}\cos(2nz),\quad&\vartheta_4(z,q)=1+2\sum_{n=1}^{\infty}(-1)^{n}q^{n^{2}}\cos(2nz).
\end{align}
As a matter of notation, it is common to simply write $\vartheta_a$ for $\vartheta_a(0,q)$ and $\vartheta_a(z)$ for $\vartheta_a(z,q)$ in equations where a single nome $q$ appears. Using Jacobi's triple product, one can rewrite these four functions as infinite products:
\begin{align}
\vartheta_1(z,q)&=2q^{1/4}\sin z\prod_{m=1}^{\infty}(1-q^{2m})(1-q^{2m}e^{2\i z})(1-q^{2m}e^{-2\i z}),\\
\vartheta_2(z,q)&=2q^{1/4}\cos z\prod_{m=1}^{\infty}(1-q^{2m})(1+q^{2m}e^{2\i z})(1+q^{2m}e^{-2\i z}),\\
\vartheta_3(z,q)&=\prod_{m=1}^{\infty}(1-q^{2m})(1+q^{2m-1}e^{2\i z})(1+q^{2m-1}e^{-2\i z}),\\
\vartheta_4(z,q)&=\prod_{m=1}^{\infty}(1-q^{2m})(1-q^{2m-1}e^{2\i z})(1-q^{2m-1}e^{-2\i z}).
\end{align}
Let us recall as well the series representation of their logarithm, which directly follows from the infinite-product formulas:
\begin{align}
&\log\frac{\vartheta_1(z,q)}{2q^{1/6}\eta(q)\sin z}=-\sum_{k=1}^{\infty}\frac{2\cos(2kz)}{k(q^{-2k}-1)},\hspace{1cm}\log\frac{\vartheta_2(z,q)}{2q^{1/6}\eta(q)\cos z}=-\sum_{k=1}^{\infty}\frac{2(-1)^k\cos(2kz)}{k(q^{-2k}-1)},\\
&\log\frac{\vartheta_3(z,q)}{2q^{-1/12}\eta(q)}=-\sum_{k=1}^{\infty}\frac{2(-1)^k q^{-k}\cos(2kz)}{k(q^{-2k}-1)},\hspace{1cm}\log\frac{\vartheta_4(z,q)}{2q^{-1/12}\eta(q)}=-\sum_{k=1}^{\infty}\frac{2q^{-k}\cos(2kz)}{k(q^{-2k}-1)},
\end{align}
where Dedekind's eta function is defined by $\eta(q)=q^{1/12}\prod_{m=1}^{\infty}(1-q^{2m})$. Jacobi's imaginary transformation $\tau\to -1/\tau$ yields the following modular properties,
\begin{align}
\vartheta_1(z,e^{\i\pi\tau})&=-\i(-\i\tau)^{-1/2}e^{-\i z^2/(\pi\tau)}\vartheta_1(-z/\tau,e^{-\i\pi/\tau}),\\
\vartheta_2(z,e^{\i\pi\tau})&=(-\i\tau)^{-1/2}e^{-\i z^2/(\pi\tau)}\vartheta_4(-z/\tau,e^{-\i\pi/\tau}),\\
\vartheta_3(z,e^{\i\pi\tau})&=(-\i\tau)^{-1/2}e^{-\i z^2/(\pi\tau)}\vartheta_3(-z/\tau,e^{-\i\pi/\tau}),\\
\vartheta_4(z,e^{\i\pi\tau})&=(-\i\tau)^{-1/2}e^{-\i z^2/(\pi\tau)}\vartheta_2(-z/\tau,e^{-\i\pi/\tau}),
\end{align}
where $(-\i\tau)^{-1/2}$ is interpreted as satisfying $|\arg(-\i\tau)|<\pi/2$.

\subsection{Jacobi's elliptic functions}
To define Jacobi's twelve elliptic functions, one introduces the \emph{elliptic modulus} $k$ and $K=K(k)$ the complete elliptic integral of the first kind, related to the nome $q$ through
\begin{equation}
k=\vartheta_2^2/\vartheta_3^2,\quad K=\pi\vartheta_3^2/2.
\end{equation}
In addition, $k'=\sqrt{1-k^2}$ is called the \emph{complementary elliptic modulus}. There are three principal elliptic functions, defined in terms of theta functions as
\begin{equation}
\begin{split}
\text{sn}(z,k)&=\frac{\vartheta_3}{\vartheta_2}\frac{\vartheta_1(\pi z/(2K),q)}{\vartheta_4(\pi z/(2K),q)},\quad\text{cn}(z,k)=\frac{\vartheta_4}{\vartheta_2}\frac{\vartheta_2(\pi z/(2K),q)}{\vartheta_4(\pi z/(2K),q)},\\
&\hspace{1cm}\text{dn}(z,k)=\frac{\vartheta_4}{\vartheta_3}\frac{\vartheta_3(\pi z/(2K),q)}{\vartheta_4(\pi z/(2K),q)}.
\end{split}
\end{equation}
The modular relations due to Jacobi's imaginary transformation follow from those of the theta functions:
\begin{equation}
\text{sn}(\i z,k)=\i\,\text{sc}(z,k'),\quad\text{cn}(\i z,k)=\text{nc}(z,k'),\quad\text{dn}(\i z,k)=\text{dc}(z,k').
\end{equation}
The nine auxiliary elliptic functions are obtained from the principal ones using the relations
\begin{equation}
\text{pq}(z,k)=\frac{\text{pr}(z,k)}{\text{qr}(z,k)}=\frac{1}{\text{qp}(z,k)},
\end{equation}
where $\text{p},\text{q},\text{r}=\text{s},\text{c},\text{d},\text{n}$ and with the convention that $\text{pq}(z,k)\equiv 1$ if $\text{q}=\text{p}$.

Jacobi's elliptic functions possess a Fourier representation only if $z,q$ satisfy certain inequalities. We only indicate here the Fourier series of the eight functions used in Section~\ref{sec4}. When $q\,\exp(2|\Im\pi z/(2K)|)<1$,
\begin{align}
\text{cn}(z,k)&=\frac{2\pi}{Kk}\sum_{n=0}^{\infty}\frac{q^{-(n+1/2)}\cos[\pi(2n+1)z/(2K)]}{q^{-(2n+1)}+1},\\
\text{dn}(z,k)&=\frac{\pi}{2K}+\frac{2\pi}{K}\sum_{n=1}^{\infty}\frac{q^{-n}\cos(\pi nz/K)}{q^{-2n}+1},\\
\text{cd}(z,k)&=\frac{2\pi}{Kk}\sum_{n=0}^{\infty}\frac{(-1)^{n}q^{-(n+1/2)}\cos\left[\pi(2n+1)z/(2K)\right]}{q^{-(2n+1)}-1},\\
\text{nd}(z,k)&=\frac{\pi}{2Kk'}+\frac{2\pi}{Kk'}\sum_{n=1}^{\infty}\frac{(-1)^{n}q^{-n}\cos\left(\pi nz/K\right)}{q^{-2n}+1}.
\end{align}
In case the weaker condition $q\,\exp(|\Im\pi z/(2K)|)<1$ is satisfied,
\begin{align}
\text{ns}(z,k)&=\frac{\pi}{2K}\frac{1}{\sin[\pi z/(2K)]}+\frac{2\pi}{K}\sum_{n=0}^{\infty}\frac{\sin[\pi(2n+1) z/(2K)]}{q^{-(2n+1)}-1},\\
\text{cs}(z,k)&=\frac{\pi}{2K}\cot[\pi z/(2K)]-\frac{2\pi}{K}\sum_{n=1}^{\infty}\frac{\sin(\pi nz/K)}{q^{-2n}+1},\\
\text{dc}(z,k)&=\frac{\pi}{2K}\frac{1}{\cos[\pi z/(2K)]}+\frac{2\pi}{K}\sum_{n=0}^{\infty}\frac{(-1)^n\cos[\pi(2n+1)z/(2K)]}{q^{-(2n+1)}-1},\\
\text{nc}(z,k)&=\frac{\pi}{2Kk'}\frac{1}{\cos[\pi z/(2K)]}-\frac{2\pi}{Kk'}\sum_{n=0}^{\infty}\frac{(-1)^n\cos[\pi(2n+1)z/(2K)]}{q^{-(2n+1)}+1}.
\end{align}


\section{Technical proofs}
\label{a2}

In this second appendix, we give the proofs of Propositions \ref{prop5.1} and \ref{prop5.2}. The latter relies on several technical lemmas, some of which were given in \cite{KW11b,KW15}.

\textbf{Proof of Proposition~\ref{prop5.1}:}\\
Due to Theorem \ref{grove_thm}, $\left(\mathbf{A}_n\right)_{R,\vsig}$ is nonzero only if each chord of $\vsig$ links an element of $R$ to one of $S=\mathcal{U}\bs R$, i.e. if every $r_i$ is connected to $s_{\rho(i)}$ in $\vsig$ for some permutation $\rho\in\mathrm{S}_n$. Recall that for the cyclic orientation, all paths are oriented from the up to the down step of the chords of $\vsig$. All $r_i$'s that index up steps therefore contribute a factor 1 to the matrix entry. On the other hand, we have to reverse the orientation of each path $r_i{\to}s_{\rho(i)}$ when $r_i$ indexes a down step of $\vsig$. As $\phi_{r_i\to s_{\rho(i)}}=z$ (resp. 1) if $r_i<s_{\rho(i)}$ (resp. $r_i>s_{\rho(i)}$) when $r_i$ is a down step, we find for $w=z^2$ that
\begin{equation}
\left(\mathbf{A}_n\right)_{R,\vsig}=\pm w^{W(\vsig)-R:\vsig}.
\end{equation}

We now turn to the signs of the entries of $\mathbf{A}_n$, which originate from the signature of the permutations $\rho$ mapping indices of $S$ to indices of $R$  in Theorem \ref{grove_thm}. Let $\lambda$ be a standard (noncyclic) Dyck path with up steps $U=\{u_1,u_2,\ldots,u_n\}$ such that $u_i<u_j$ if $i<j$. We denote by $D=\{d_1,d_2,\ldots,d_n\}$ the collection of down steps of $\lambda$, such that $(u_i,d_i)$ is a chord of $\lambda$ for $1\le i\le n$. Since $D$ is not ordered in general, there exists a permutation $\pi$ of the indices $1,2,\ldots,n$ that sorts $D$ in ascending order; we write the result $D_{\pi}$. The signature of $\pi$ is given by $\epsilon(\pi)=(-1)^{I(D)}$, where $I(D)=\left|\big\{(d_i,d_j)|d_i>d_j\textrm{ and }i<j\big\}\right|$ is the inversion number of $D$. Observe that the condition $u_i<u_j<d_j<d_i$ means that the chord $(u_j,d_j)$ lies above $(u_i,d_i)$ in $\lambda$ (represented as a mountain range). It follows that the inversion number can be expressed as
\begin{equation}
I(D)=\sum_{\substack{\textrm{chords $(u_i,d_i)$}\\\textrm{of $\lambda$}}}\left|\big\{\textrm{chords $(u_j,d_j)$ of $\lambda$ above $(u_i,d_i)$}\big\}\right|.
\end{equation}
By tiling the area between $\lambda$ and the $x$ axis with $\sqrt{2}{\times}\sqrt{2}$ squares as depicted in Fig. \ref{boxes}, we see that the number of chords of $\lambda$ above a given chord $(u_i,d_i)$ is equal to the number of squares intersected by $(u_i,d_i)$ in their lower half. Each of these squares is intersected in its lower half by exactly one chord of $\lambda$, so $I(D)$ is the number of squares under $\lambda$ or, equivalently, $I(D)=\frac{1}{2}(\mathcal{A}(\lambda)-n)$. The sign associated with $\left(\mathbf{A}_n\right)_{U,\lambda}$ is therefore given by $\epsilon(\pi^{-1})=(-1)^{\frac{1}{2}(\mathcal{A}(\lambda)-n)}$.

\begin{figure}[h]
\centering
\begin{tikzpicture}
\filldraw[blue!10!white] (1.5,0.5)--(1,1)--(1.5,1.5)--(2,1);
\filldraw[blue!10!white] (2.5,0.5)--(2,1)--(2.5,1.5)--(3,1);
\filldraw[blue!10!white] (3.5,0.5)--(3,1)--(3.5,1.5)--(4,1);
\filldraw[blue!10!white] (4.5,0.5)--(4,1)--(4.5,1.5)--(5,1);
\draw[gray] (0.5,0.5)--(1,0);
\draw[gray] (1,1)--(2,0);
\draw[gray] (1.5,1.5)--(3,0);
\draw[gray] (3,1)--(4,0);
\draw[gray] (4,1)--(5,0);
\draw[gray] (6.5,0.5)--(7,0);
\draw[gray] (1,0)--(2.5,1.5);
\draw[gray] (2,0)--(3,1);
\draw[gray] (3,0)--(4,1);
\draw[gray] (4,0)--(5,1);
\draw[gray] (5,0)--(5.5,0.5);
\draw[gray] (7,0)--(7.5,0.5);
\draw[ultra thick,blue] (0.75,0.75)--(5.25,0.75);
\node at (0.375,0.875) {$u_i$};
\node at (5.675,0.875) {$d_i$};
\draw[very thick,densely dashed,blue] (1.25,1.25)--(2.75,1.25);
\draw[very thick,densely dashed,blue] (1.75,1.75)--(2.25,1.75);
\draw[very thick,densely dashed,blue] (3.25,1.25)--(3.75,1.25);
\draw[very thick,densely dashed,blue] (4.25,1.25)--(4.75,1.25);
\draw[thick] (-0.5,0)--(8.5,0);
\draw (0,0)--(1,1)--(2,2)--(3,1)--(3.5,1.5)--(4,1)--(4.5,1.5)--(6,0)--(7,1)--(8,0);
\filldraw (0,0) circle (0.05cm);
\filldraw (0.5,0.5) circle (0.05cm);
\filldraw (1,1) circle (0.05cm);
\filldraw (1.5,1.5) circle (0.05cm);
\filldraw (2,2) circle (0.05cm);
\filldraw (2.5,1.5) circle (0.05cm);
\filldraw (3,1) circle (0.05cm);
\filldraw (3.5,1.5) circle (0.05cm);
\filldraw (4,1) circle (0.05cm);
\filldraw (4.5,1.5) circle (0.05cm);
\filldraw (5,1) circle (0.05cm);
\filldraw (5.5,0.5) circle (0.05cm);
\filldraw (6,0) circle (0.05cm);
\filldraw (6.5,0.5) circle (0.05cm);
\filldraw (7,1) circle (0.05cm);
\filldraw (7.5,0.5) circle (0.05cm);
\filldraw (8,0) circle (0.05cm);
\end{tikzpicture}
\caption{Standard Dyck path with a marked chord $(u_i,d_i)$, which intersects the lower half of four (colored) squares. The four chords that lie above $(u_i,d_i)$ are drawn as dashed lines.}
\label{boxes}
\end{figure}
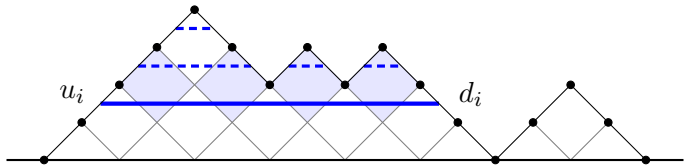

The next step of the proof consists in showing that the sign of $\left(\mathbf{A}_n\right)_{R,\lambda}$ does not depend on $R$. To do so, let us define $U,D$ as above as the sets of up and down steps of $\lambda$, and let $R$ be $U\cup\{d_\l\}\backslash\{u_\l\}$ for some index $\l$, i.e. $R=\{u_1,u_2,\ldots,u_{\l-1},d_{\l},u_{\l+1},\ldots,u_n\}$. The subset $S$ associated with $R$ is therefore $S=\{d_1,d_2,\ldots,d_{\l-1},u_{\l},d_{\l+1},\ldots,d_n\}$. To compute $I(S)-I(R)$ we need to distinguish between two cases:
\vspace{-0.4cm}
\begin{enumerate}[(a)]
\item Consider a chord $(u_j,d_j)$ that lies to the left of $(u_{\l},d_{\l})$ in $\lambda$, that is, such that $u_j<d_j<u_{\l}<d_{\l}$. Since $u_j<u_{\l},d_{\l}$ and $d_j<u_{\l},d_{\l}$, these two chords do not contribute to $I(R)$ or $I(S)$. Similarly there is no contribution from a chord to the right of $(u_{\l},d_{\l})$ or below it.
\item A chord $(u_j,d_j)$ above $(u_{\l},d_{\l})$ is such that $u_{\l}<u_j<d_j<d_{\l}$, so it contributes equally to both $I(R)$ and $I(S)$.
\end{enumerate}
\vspace{-0.4cm}
Since $I(D)-I(U)=I(D)=I(S)-I(R)$, we see that the signs of $\left(\mathbf{A}_n\right)_{U,\lambda}$ and $\left(\mathbf{A}_n\right)_{R,\lambda}$ are the same. The argument holds if $R$ is obtained from $U$ by replacing two or more up steps of $\lambda$ with the corresponding down steps. The sign is therefore the same for any $R$ such that $R\cap\lambda$.

Finally, we take into account cyclic Dyck paths. Let $\vsig$ be a cyclic Dyck path with up steps $U$ and down steps $D$. We define $\vsig'$ by shifting all step indices of $\vsig$ to the left by one unit, that is $U'=U+1\mod 2n$ and $D'=D+1\mod 2n$. Assume first that $2n\in D$, then $I(U')=I(U)$ and $I(D')-I(D)=n-1\mod 2$. If rather $2n\in U$ then $I(U')-I(U)=n-1\mod 2$ and $I(D')=I(D)$. In both cases $(-1)^{I(D')-I(U')}=(-1)^{n+1}(-1)^{I(D)-I(U)}$. Therefore if a cyclic Dyck path $\vsig$ is obtained by shifting $k$ times to the left the indices of a standard Dyck path $\lambda$, then the sign of $\left(\mathbf{A}_n\right)_{R,\vsig}$ is $(-1)^{\frac{1}{2}(\mathcal{A}(\lambda)-n)+k(n+1)}$, with $\mathcal{A}(\lambda)=\mathcal{A}(\vsig)$.

Lastly we make the following observation: a unit shift of the indices of $\vsig$ to the left to define $\vsig'$ implies that $|W(\vsig')-W(\vsig)|=1$, since $W(\vsig)=h_0(\vsig)$ is the height of the left endpoint of the step indexed by $1$ in $\vsig$. Since the parity of $W(\vsig)$ alternates between even and odd for each unit shift of the indices of $\vsig$ to the left, the sign of $\left(\mathbf{A}_n\right)_{R,\vsig}$ can be recast as $(-1)^{\frac{1}{2}(\mathcal{A}(\vsig)-n)+(n+1)W(\vsig)}$.$\hfill\square$

\begin{lm} Let $\vsig$ be a cyclic Dyck path and $R\subset\mathcal{U}$ a subset of order $n$ such that $R\cap\vsig$. If we write $\Sigma R\equiv\sum_{i=1}^{n}r_i$ then
\begin{equation}
(-1)^{\Sigma R}=(-1)^{R\cdot\vsig+\frac{1}{2}(\mathcal{A}(\vsig)-n)+n\,W(\vsig)}.
\end{equation}
\label{lmB1}
\end{lm}
\begin{proof}
Let $R\subset\mathcal{U}$ such that $R$ intersects each chord of $\vsig$ once. Let $i,j$ be the indices of a chord of $\vsig$ with $i\in R$, and define $R'\equiv\left(R\backslash\{i\}\right)\cup\{j\}$. Since $i$ and $j$ have opposite parity, $(-1)^{\Sigma R'}=-(-1)^{\Sigma R}$. Moreover, $R'\cdot\vsig=R\cdot\vsig-1$ if $i$ is an up step, and $R'\cdot\vsig=R\cdot\vsig+1$ if $i$ is a down step. It follows that $(-1)^{\Sigma R'-\Sigma R}=(-1)^{R'\cdot\vsig-R\cdot\vsig}$. This relation holds more generally for any $R'$ that intersects once each chord of $\vsig$.

Consider in particular the set $R$ of up steps of $\vsig$, which implies that $R\cdot\vsig=n$, and assume that the first step of $\vsig$ has the label $\l\ge 1$. Denote by $s_j$ the value of the step at position $j$ for $1\le j\le 2n$, with $s_j=+1$ (resp. $-1$) if $j$ is an up step (resp. down step) of $\vsig$. The height $h_j$ for $0\le j\le 2n$ is given by
\begin{equation}
h_j=\begin{cases}
h_{j-1}+s_j=h_0+\sum_{k=1}^{j}s_k&\quad\textrm{if $1\le j\le\l-1$,}\\
h_{j+1}-s_j=h_0+\sum_{k=1}^{j-1}s_k&\quad\textrm{if $\l+2\le j\le 2n$,}\\
0&\quad\textrm{if $j=\l,\l+1$,}
\end{cases}
\end{equation}
where $h_{\l}=h_{\l+1}=0$ since they correspond to the first and last vertices of the cyclic Dyck path $\vsig$. After some algebra, we can express the area $\mathcal{A}(\vsig)$ between the Dyck path and the horizontal axis in terms of the $s_j$'s as
\begin{equation}
\mathcal{A}(\vsig)=\sum_{j=0}^{2n}h_j=2n\,h_0-\sum_{j=1}^{2n}j\,s_j.
\end{equation}
By separating the last sum in the equation above for positive and negative values of the $s_j$'s, one finds
\begin{equation}
\sum_{j=1}^{2n}j\,s_j=\sum_{j:\,s_j=+1}j\,s_j+\sum_{j:\,s_j=-1}j\,s_j=\sum_{j:\,s_j=+1}2j-n(2n+1),
\end{equation}
where we used the fact that $s_1+s_2+\ldots+s_{2n}=0$ for a (cyclic) Dyck path. Hence, we obtain the following formula (recall that $h_0(\vsig)=W(\vsig)$):
\begin{equation}
\Sigma R=\sum_{j:\,s_j=+1}j=n^2+\frac{1}{2}(n-\mathcal{A}(\vsig))+n\,W(\vsig),
\end{equation}
which completes the proof.
\end{proof}

\begin{lm} Let $\lambda$ be a standard, noncyclic, Dyck path with a chord $(\l,m)$. Define $\lambda'$ as the Dyck path obtained by ``pushing down'' the chord $(\l,m)$ of $\lambda$, that is, by replacing the up step $\l$ with a down step and the down step $m$ with an up step (note that not all chords can be pushed down). The areas between the two Dyck paths are related by
\begin{equation}
\mathcal{A}(\lambda')=\mathcal{A}(\lambda)-2(m-\l).
\end{equation}
Moreover, for any subset $R\subset\mathcal{U}$ such that $R\cap\lambda$, $(-1)^{R\cdot\lambda'-R\cdot\lambda}=(-1)^{\frac{1}{2}(\mathcal{A}(\lambda')-\mathcal{A}(\lambda))}$.
\label{lmB2}
\end{lm}
\begin{proof}
Observe that the heights of the vertices $h'_i$ in $\lambda'$ are decreased by $2$ with respect to their counterparts $h_i$ in $\lambda$ for $\l\le i\le m-1$. If the up step $\l$ of $\lambda$ belongs to $R$, then $R\cdot\lambda'=R\cdot\lambda-1$ since $\l$ is a down step of $\lambda'$. Conversely if $m\in R$, $R\cdot\lambda'=R\cdot\lambda+1$ since $m$ is a down step of $\lambda$ but an up step of $\lambda'$. In both cases $R\cdot\lambda$ and $R\cdot\lambda'$ differ by one, and $m-\l=\frac{1}{2}(\mathcal{A}(\lambda)-\mathcal{A}(\lambda'))$ is odd.
\end{proof}

\begin{lm}[Adapted from \cite{KW15}, Lemma A.2\hspace{-0.025cm}] Let $\lambda$ and $\vsig$ be a standard and a cyclic Dyck paths, respectively. Then
\begin{equation} 
\sum_{R:\,R\cap\vsig}(-1)^{R\cdot\vsig}w^{-R\cdot\lambda-R:\vsig}=(-1)^{\frac{1}{2}\left(\mathcal{A}(\vsig)-\mathcal{A}(\lambda)\right)}(1-w^{-1})^n
\label{lmB3eq}
\end{equation}
if it possible to push down some chords of $\vsig$ to obtain $\lambda$, and $0$ otherwise. In particular the step labels of $\vsig$ and $\lambda$ must appear in the same order for the sum to be nonzero, i.e. $\vsig$ must be a standard Dyck path.
\label{lmB3}
\end{lm}
\begin{proof}
Let us assume first that $\vsig$ is a standard Dyck path, which we shall simply write as $\sigma$. If we can push down some chords of $\sigma$ to obtain $\lambda$, then each chord of $\sigma$ connects an up step to a down step of $\lambda$ (not necessarily belonging to the same chord). For any two subsets $R,R'\subset\mathcal{U}$ of order $n$ that intersect each chord of $\sigma$ exactly once,
\begin{equation}
(-1)^{R\cdot\sigma+R'\cdot\sigma}=(-1)^{R\cdot\lambda+R'\cdot\lambda}.
\end{equation}
Indeed, consider a chord $(i,j)$ of $\sigma$ on which $R$ and $R'$ differ. Without loss of generality, we can assume that the up step $i$ belongs to $R$ and the down step $j$ to $R'$. If $(i,j)$ was pushed down to obtain $\lambda$ then $i$ is a down step of $\lambda$ in $R$, and $j$ an up step of $\lambda$ in $R'$. Otherwise $i$ is an up step of $\lambda$ and $j$ a down step of $\lambda$. In both cases, the chord $(i,j)$ brings a factor $-1$ to both $(-1)^{R\cdot\sigma+R'\cdot\sigma}$ and $(-1)^{R\cdot\lambda+R'\cdot\lambda}$. Now denote by $R_0$ the set of down steps of $\lambda$, so $R_0\cdot\lambda=0$. Then $R_0\cdot\sigma$ is the number of chords of $\sigma$ that are pushed down to obtain $\lambda$. Lemma \ref{lmB2} then yields the equality $(-1)^{R_0\cdot\sigma}=(-1)^{\frac{1}{2}(\mathcal{A}(\sigma)-\mathcal{A}(\lambda))}$. As $(-1)^{R\cdot\sigma}=(-1)^{R_0\cdot\sigma+R\cdot\lambda}$, it follows that
\begin{equation}
\sum_{R:\,R\cap\sigma}(-1)^{R\cdot\sigma}w^{-R\cdot\lambda}=(-1)^{R_0\cdot\sigma}\sum_{R:\,R\cap\sigma}(-w^{-1})^{R\cdot\lambda}=(-1)^{\frac{1}{2}(\mathcal{A}(\sigma)-\mathcal{A}(\lambda))}(1-w^{-1})^n,
\end{equation}
where the last equality is given by the binomial expansion. On the other hand, if $\lambda$ cannot be obtained by pushing down some chords of $\sigma$, there is a chord $(i,j)$ of $\sigma$ such that both $i$ and $j$ are up steps of $\lambda$. Then for each $R$ such that $R\cap\sigma$ in the sum, define $R'$ as the symmetric difference of $R$ with $\{i,j\}$. Since $(-1)^{R\cdot\sigma}=-(-1)^{R'\cdot\sigma}$ and $R\cdot\lambda=R'\cdot\lambda$, the sum over all $R$'s such that $R\cap\sigma$ is zero.

Let us now discuss the case of a cyclic, nonstandard, Dyck path $\vsig$, for which there are several subcases:
\vspace{-0.5cm}
\begin{enumerate}[$(i)$]
\item If $(i,j)$ is a chord of $\vsig$ such that $i<j$ are both up or down steps of $\lambda$, the sum is zero. To see this, define for each $R$ such that $R\cap\vsig$ its symmetric difference with $\{i,j\}$, denoted by $R'$. Then $(-1)^{R\cdot\vsig}=-(-1)^{R'\cdot\vsig}$, $w^{-R\cdot\lambda}=w^{-R'\cdot\lambda}$ and $w^{-R:\vsig}=w^{-R':\vsig}$, so the contributions of $R$ and $R'$ cancel each other out.\vspace{-0.2cm}
\item Suppose $(2n,j)$ is a chord of $\vsig$ such that $j$ is an up step of $\lambda$. For each $R$ define $R'$ as the symmetric difference of $R$ with $\{2n,j\}$. If $j\in R$ then $w^{-R\cdot\lambda}=w^{-R'\cdot\lambda-1}$ and $w^{-R:\vsig}=w^{-R':\vsig+1}$. Conversely if $j\in R'$ then $w^{-R\cdot\lambda}=w^{-R'\cdot\lambda+1}$ and $w^{-R:\vsig}=w^{-R':\vsig-1}$. Therefore $(-1)^{R\cdot\vsig}=-(-1)^{R'\cdot\vsig}$ and $w^{-R\cdot\lambda-R:\vsig}=w^{-R'\cdot\lambda-R':\vsig}$ in both cases, so the sum is zero.\vspace{-0.2cm}
\item Assume $(i,2n)$ is a chord of $\vsig$ such that $i$ is a down step of $\lambda$. Then for $R'$ the symmetric difference of $R$ with $\{i,2n\}$, we have $(-1)^{R\cdot\vsig}=-(-1)^{R'\cdot\vsig}$, $w^{-R\cdot\lambda}=w^{-R'\cdot\lambda}$ and $w^{-R:\vsig}=w^{-R':\vsig}$. Therefore, $R,R'$ give opposite contributions to the sum, which vanishes.\vspace{-0.2cm}
\item Suppose $(2n,j)$ is a chord of $\vsig$ such that $j$ is a down step of $\lambda$. The subpath of $\vsig$ consisting in steps $\{1,\ldots,j{-}1\}$ is a standard Dyck path, so it contains the same number of up steps and down steps of $\vsig$. Consider now the subpath of $\lambda$ with the same indices. Since both $j$ and $2n$ are down steps of $\lambda$ with opposite parity, $h_j(\lambda)-h_{2n}(\lambda)$ must be odd, so the step $j$ lies higher than the step $2n$ on $\lambda$. This implies that the subset $\{1,\ldots,j{-}1\}$ contains more up steps than down steps of $\lambda$. Therefore there exists a chord $(\l,m)$ of $\vsig$ with $1\le\l<m\le j{-}1$ such that both $\l$ and $m$ are up steps of $\lambda$, so we can refer to subcase $(i)$.\vspace{-0.2cm}
\item If $(i,2n)$ is a chord of $\vsig$ with $i$ an up step of $\lambda$, then the subpath of $\vsig$ indexed by the indices $\{i{+}1,\ldots,2n{-}1\}$ is a standard Dyck path.\vspace{-0.2cm}
\begin{enumerate}[a.] 
\item If $h_{i-1}(\lambda)>h_{2n}(\lambda)=0$, that is if the up step $i$ is higher than the down step $2n$ in $\lambda$, then the subset $\{i{+}1,\ldots 2n{-}1\}$ contains more down steps than up steps of $\lambda$. Therefore there exists a chord $(\l,m)$ of $\vsig$ with $i<\l<m<2n$ such that both $\l$ and $m$ are down steps of $\lambda$, which corresponds to subcase $(i)$.
\item If on the other hand $h_{i-1}(\lambda)=h_{2n}(\lambda)=0$, then the subpath of $\lambda$ indexed by $\{1,\ldots,i{-}1\}$ is a Dyck path. In that case, if $i{-}1$ is an up step of $\vsig$, it belongs to a chord $(i{-}1,j)$ of $\vsig$ with $i{-}1>j$.
\begin{itemize}
\item If $j$ is an up step of $\lambda$, subcase $(ii)$ occurs;
\item If $j$ is an down step of $\lambda$, subcase $(iv)$ occurs.
\end{itemize}
If on the contrary $i{-}1$ is a down step of $\vsig$, it belongs to a chord $(j,i{-}1)$ of $\vsig$ with $j<i{-}1$.
\begin{itemize}
\item If $j$ is a down step of $\lambda$, subcase $(iii)$ occurs;
\item If $j$ is an up step of $\lambda$ such that $h_{j-1}(\lambda)>h_{i-1}(\lambda)=0$, subcase $(v)a.$ occurs;
\item If $j$ is an up step of $\lambda$ such that $h_{j-1}(\lambda)=h_{i-1}(\lambda)=0$, then the subpath of $\lambda$ indexed by $\{1,\ldots,j{-}1\}$ is a standard Dyck path, so we can consider subcase $(v)b.$ again with $j{-}1$ instead of $i{-}1$.
\end{itemize}
\end{enumerate}
\end{enumerate}
\vspace{-0.3cm}
For all cases, we see that the sum in Eq.~\eqref{lmB3eq} vanishes if $\vsig$ is not a standard Dyck path.
\end{proof}

\begin{lm}[\hspace{-0.025cm}\cite{KW11b}, Theorem 1.5\hspace{-0.025cm}] Let $\vmu,\vsig,\vt$ be cyclic Dyck paths, and let $M$ be the matrix defined by $M_{\vmu,\vt}=1$ if $\vmu$ can be obtained by pushing down some of the chords of $\vt$, and zero otherwise. Then
\begin{equation}
\sum_{\vmu\ge\vsig}(-1)^{\frac{1}{2}(\mathcal{A}(\vmu)-\mathcal{A}(\vsig))}\mathrm{ci}(\vsig/\vmu)M_{\vmu,\vt}=\delta_{\vsig,\vt}\,.
\end{equation}
\label{lmB4}
\end{lm}

\textbf{Proof of Proposition~\ref{prop5.2}:}\\
Let us first consider a standard Dyck path $\lambda$, for which $W(\lambda)=R|\lambda_{\l(1)}=\big|D^{\lambda}_{\l(1)}\big|=0$. The product of $\mathbf{B}_n$ and $\mathbf{A}_n$ yields
\begin{equation}
\begin{split}
\sum_{\substack{R\subset\mathcal{U}\\|R|=n}}\left(\mathbf{B}_n\right)_{\lambda,R}\left(\mathbf{A}_n\right)_{R,\vt}&=\sum_{R:\,R\cap\vt}(-1)^{\Sigma R+n}\sum_{\mu\ge\lambda}\mathrm{ci}(\lambda/\mu)\,w^{n-R\cdot\mu}\times(-1)^{\frac{1}{2}(\mathcal{A}(\vt)-n)+(n+1)W(\vt)}w^{W(\vt)-R:\vt}\\
&=(-1)^{\frac{1}{2}(\mathcal{A}(\vt)+n)+(n+1)W(\vt)}w^{n+W(\vt)}\sum_{\mu\ge\lambda}\mathrm{ci}(\lambda/\mu)\sum_{R:\,R\cap\vt}(-1)^{\Sigma R}w^{-R\cdot\mu-R:\vt}.
\end{split}
\end{equation}
We use Lemma~\ref{lmB1} to write $(-1)^{\Sigma R}$ in terms of $(-1)^{R\cdot\vt}$ on the right-hand side, and Lemma~\ref{lmB3} to compute the sum over all $R$'s such that $R\cap\vt$. The result reads
\begin{equation*}
(-w)^{n+W(\vt)}(1-w^{-1})^n\sum_{\mu\ge\lambda}(-1)^{\frac{1}{2}(\mathcal{A}(\vt)-\mathcal{A}(\mu))}\mathrm{ci}(\lambda/\mu)M_{\mu,\vt}.
\end{equation*}
Applying Lemma~\ref{lmB4} to the above expression yields
\begin{equation}
\sum_{\substack{R\subset\mathcal{U}\\|R|=n}}\left(\mathbf{B}_n\right)_{\lambda,R}\left(\mathbf{A}_n\right)_{R,\vt}=(-1)^{\frac{1}{2}(\mathcal{A}(\vt)-\mathcal{A}(\lambda))}(-w)^{W(\vt)}(1-w)^n\delta_{\lambda,\vt}=(1-w)^n\delta_{\lambda,\vt}.
\label{prop5.2eq2}
\end{equation}
Consider now two cyclic Dyck paths $\vsig,\vt$ such that the first step of $\vsig$ (resp. $\vt$) is labeled by the index $k{+}1$ (resp. $\l{+}1$). Let us define $\sigma_0$ and $\vtp$ by subtracting $k\mod 2n$ to the label of each step of $\vsig$ and $\vt$, so that $\sigma_0$ is a standard Dyck path. We define for each $R$ the subset $R'$ by subtracting $k\mod 2n$ to each index in $R$. We further assume that $R\cap\vt$, as $\left(\mathbf{A}_n\right)_{R,\vt}$ vanishes otherwise. Our goal is to rewrite the product $\left(\mathbf{B}_n\right)_{\vsig,R}\left(\mathbf{A}_n\right)_{R,\vt}$ in terms of $\left(\mathbf{B}_n\right)_{\sigma_0,R'}\left(\mathbf{A}_n\right)_{R',\vtp}$, and then use Eq.~\eqref{prop5.2eq2} to get the desired result.

Let us first discuss the signs of both products. Explicitly, their quotient reads
\begin{equation}
\frac{(-1)^{\Sigma R+W(\vsig)+n+\frac{1}{2}(\mathcal{A}(\vt)-n)+(n+1)W(\vt)}}{(-1)^{\Sigma R'+W(\sigma_0)+n+\frac{1}{2}(\mathcal{A}(\vtp)-n)+(n+1)W(\vtp)}}=(-1)^{kn+W(\vsig)+(n+1)(W(\vt)-W(\vtp))},
\label{prop5.2eq3}
\end{equation}
since $\sigma_0$ is a standard Dyck path ($W(\sigma_0)=0$) and $\Sigma R=\Sigma R'+kn\mod 2n$. Moreover, the permutation of the indices of $\vt$ to form $\vtp$ leaves the area under the path invariant: $\mathcal{A}(\vt)=\mathcal{A}(\vtp)$. In addition, observe that rotating the step indices of $\vt$ by one unit to the left changes $W(\vt)=h_0(\vt)$ by $\pm 1$. One may therefore write the equality $(-1)^{W(\vt)-W(\vtp)}=(-1)^k=(-1)^{W(\vsig)}$, which implies that Eq.~\eqref{prop5.2eq3} is equal to 1.

Applying this result together with the property $R\cdot\vmu=R'\cdot\mu_0$ for all cyclic Dyck paths $\vmu\ge\vsig$ yields the relation
\begin{equation}
\frac{\left(\mathbf{B}_n\right)_{\vsig,R}\left(\mathbf{A}_n\right)_{R,\vt}}{\left(\mathbf{B}_n\right)_{\sigma_0,R'}\left(\mathbf{A}_n\right)_{R',\vtp}}=w^{W(\vt)-W(\vtp)-W(\vsig)+R|\vsig_{\l(1)}-\left|D^{\vsig}_{\l(1)}\right|-R:\vt+R':\vtp}.
\end{equation}
To simplify this expression, let us note that for any $R$ such that $R\cap\vt$, $R:\vt=R|\vt_{\l(1)}-\left|D^{\vt}_{\l(1)}\right|$. Indeed, $R|\vt_{\l(1)}$ counts the number of indices $i\in R$ such that $(a)$ the chord $(i,j)$ belongs to $\vt$ with $i$ appearing before the step $1$ and $j$ afterwards (i.e. $i>j$), or $(b)$ the chord $(i,j)$ or the chord $(j,i)$ belongs to $\vt$ with both $i,j$ appearing before the step $1$ in $\vt$. Since $R$ intersects each chord of $\vt$ exactly once, there are as many indices of the second type in $R$ as there are down steps before the step $1$ in $\vt$.

Let us now use this decomposition to compute the difference $R':\vtp-R:\vt$. Assume first that the path $\vt$ touches the horizontal axis only at its endpoints. There are two possible cases:
\vspace{-0.3cm}
\begin{enumerate}[(a)]
\item If $\l\le k$, the step $1$ of $\vt$ is rotated to the left to obtain $\vtp$. Therefore $\left|D^{\vt}_{\l(1)}\right|\ge\left|D^{\vtp}_{\l(1)}\right|$ and $R|\vt_{\l(1)}\ge R'|\vtp_{\l(1)}$, since there are more indices before the step $1$ of $\vt$ than before that of $\vtp$. These extra steps before the step $1$ of $\vt$ are labeled by $2n,2n{-}1,\ldots,2n{-}k{+}1$: they are precisely the indices of $\vsig$ appearing before its step $1$. It follows that $\left|D^{\vt}_{\l(1)}\right|=\left|D^{\vtp}_{\l(1)}\right|+\left|D^{\vt}_{\l(1,\vsig)}\right|$ and $R|\vt_{\l(1)}=R'|\vtp_{\l(1)}+R|\vsig_{\l(1)}$, where $D^{\vt}_{\l(1,\vsig)}$ is given by Definition~\ref{def5.1}. Hence, we find the relation
\begin{equation}
R|\vsig_{\l(1)}-R:\vt+R':\vtp=\left|D^{\vt}_{\l(1,\vsig)}\right|,
\label{prop5.2eq4}
\end{equation}
which, crucially, is independent from $R$.
\item If $\l>k$, the step $1$ of $\vt$ is rotated to the right to define $\vtp$. Therefore $\left|D^{\vtp}_{\l(1)}\right|\ge\left|D^{\vt}_{\l(1)}\right|$ and $R'|\vtp_{\l(1)}\ge R|\vt_{\l(1)}$, since there are more indices before the step $1$ of $\vtp$ than before that of $\vt$. The extra steps before the step $1$ of $\vtp$ are labeled by $2n,2n{-}1,\ldots,2n{-}k{+}1$. They also correspond to the steps $1,\ldots,k$ in $\vt$, which are the indices of $\vsig$ \emph{after} its step $1$ (included). Hence, we find that $\left|D^{\vtp}_{\l(1)}\right|=\left|D^{\vt}_{\l(1)}\right|+n-\left|D^{\vt}_{\l(1,\vsig)}\right|$ and $R'|\vtp_{\l(1)}=R|\vt_{\l(1)}+n-R|\vsig_{\l(1)}$, since any cyclic Dyck path of length $2n$ has $n$ down steps and $|R|=n$. Equation~\eqref{prop5.2eq4} holds as well in this case.
\end{enumerate}
\vspace{-0.3cm}
Up to this point, we have not considered cyclic Dyck paths $\vt$ that intersect the horizontal axis at intermediary vertices distinct from its extremities (see for instance Fig.~\ref{cyc_Dyck_ex}). As explained at the beginning of Section~\ref{sec5.2}, certain permutations of the step indices yield ``forbidden'' paths, in which the step 1 is located to the right of an intersection with the horizontal axis. They can however be transformed into ``admissible'' cyclic Dyck paths by cyclically rotating connected components such that the step 1 appears in the first one. Doing so changes the number of steps located to the left of the step 1, so the above decomposition of $R|\vt_{\l(1)}$ and $\left|D^{\vt}_{\l(1)}\right|$ does not hold. Since $R\cap\vt$, it should be clear however that the difference of both quantities, which is equal to $R:\vt$, is left invariant under permutations of connected components of paths. Equation~\eqref{prop5.2eq4} is therefore valid for any cyclic Dyck paths $\vsig,\vt$.

Putting all the pieces together yields the following equation,
\begin{equation}
\begin{split}
\sum_{\substack{R\subset\mathcal{U}\\|R|=n}}\left(\mathbf{B}_n\right)_{\vsig,R}\left(\mathbf{A}_n\right)_{R,\vt}&=\sum_{\substack{R'\subset\mathcal{U}\\|R'|=n}}\left(\mathbf{B}_n\right)_{\sigma_0,R'}\left(\mathbf{A}_n\right)_{R',\vtp}w^{W(\vt)-W(\vtp)-W(\vsig)-\left|D^{\vsig}_{\l(1)}\right|+\left|D^{\vt}_{\l(1,\vsig)}\right|}\\
&=(1-w)^{n}w^{W(\vt)-W(\vtp)-W(\vsig)-\left|D^{\vsig}_{\l(1)}\right|+\left|D^{\vt}_{\l(1,\vsig)}\right|}\,\delta_{\sigma_0,\vtp}=(1-w)^{n}\delta_{\vsig,\vt}\,,
\end{split}
\end{equation}
since $D^{\vsig}_{\l(1,\vsig)}=D^{\vsig}_{\l(1)}$ and $\vsig=\vt\Leftrightarrow\sigma_0=\vtp$, which concludes the proof.$\hfill\square$



\end{document}